\PassOptionsToPackage{dvipsnames,svgnames,x11names}{xcolor}

\documentclass{lmcs}
\usepackage[utf8]{inputenc}
\pdfoutput=1

\usepackage{lastpage}
\lmcsdoi{19}{3}{2}
\lmcsheading{}{\pageref{LastPage}}{}{}%
{Dec.~01,~2021}{Jul.~07,~2023}{}

\usepackage{algpseudocode}
\usepackage{subcaption}
\usepackage{dsfont}
\usepackage{url}
\usepackage{tikz}
\usetikzlibrary{calc,arrows,positioning}
\usepackage{stmaryrd}
\usepackage{wasysym}
\usepackage{xfrac}

\allowdisplaybreaks%

\newcommand{\pstep}{\mathsf{pstep}}
\newcommand{\cstep}{\mathsf{cstep}}

\newcommand{\nats}{\mathbb{N}}
\newcommand{\borel}{{\mathcal{B}}}
\newcommand{\reals}{\mathbb{R}}

\newcommand{\rv}{\textbf{X}}

\newcommand{\Proc}{\mathcal{P}}
\newcommand{\proc}{P}
\newcommand{\D}{\mathcal{D}}
\newcommand{\ds}{\mathbf{d}}

\newcommand{\system}{\mathbf{s}}

\newcommand{\mesA}{\mathbb{A}}
\newcommand{\mesB}{\mathbb{B}}
\newcommand{\mesP}{\mathbb{P}}
\newcommand{\mesD}{\mathbb{D}}
\newcommand{\mesC}{\mathbb{C}}

\newcommand{\distrib}{\Pi}

\newcommand{\PP}{\mathit{Pr}}

\newcommand{\config}[3]{\langle#1,#2\rangle_{#3}}

\newcommand{\Var}{\mathrm{Var}}

\DeclareMathOperator{\Wasserstein}{\mathbf{W}}

\newcommand{\cmerge}{~|~}

\newcommand{\wact}[2]{(#1\rightarrow #2)}

\newcommand{\dirac}{\mathbf{\delta}}

\newcommand{\OT}{\mathrm{OT}}

\newcommand{\traccione}{evolution sequence}
\newcommand{\tracciones}{evolution sequences}
\newcommand{\Traccione}{Evolution sequence}

\newcommand{\spell}{evolution }

\newcommand{\uno}{\boldsymbol{1}}

\newcommand{\dataspace}{data space}
\newcommand{\dataspaces}{data spaces}
\newcommand{\Dataspace}{Data space}

\newcommand{\datastate}{data state}
\newcommand{\datastates}{data states}
\newcommand{\Datastate}{Data state}

\newcommand{\ES}{\mathcal{S}}

\newcommand{\C}{\mathcal C}

\newcommand{\dd}{\mathrm{d}}

\newcommand{\E}{\mathcal E}

\newcommand{\F}{\Sigma}

\newcommand{\m}{\mathfrak{m}}

\newcommand{\W}{\mathfrak{W}}
\newcommand{\w}{\mathfrak{w}}

\newcommand{\nat}{\mathbb{N}} % Naturals
\newcommand{\real}{\mathbb{R}} % Reals

\newcommand{\support}{\mathsf{supp}}

\newcommand{\var}[1]{\mathrm{var}(#1)}

\newcommand{\ropen}[1]{[#1)} % chktex 9
\newcommand{\lopen}[1]{(#1]} % chktex 9

\begin{document}

\title[A framework to measure the robustness of programs]
{A framework to measure the robustness of programs in the unpredictable environment{\rsuper*}}
\titlecomment{{\lsuper*}A preliminary version of this paper appeared as~\cite{CLT21}.}

\author[V.~Castiglioni]{Valentina Castiglioni\lmcsorcid{0000-0002-8112-6523}}[a]
\author[M.~Loreti]{Michele Loreti\lmcsorcid{0000-0003-3061-863X}}[b]
\author[S.~Tini]{Simone Tini\lmcsorcid{0000-0002-3991-5123}}[c]

\address{Reykjavik University, Iceland}
\email{valentinac@ru.is, vale.castiglioni@gmail.com}

\address{University of Camerino, Italy}
\email{michele.loreti@unicam.it}

\address{University of Insubria, Italy}
\email{simone.tini@uninsubria.it}


\begin{abstract}
Due to the diffusion of IoT, modern software systems are often thought to control and coordinate smart devices in order to manage assets and resources, and to guarantee efficient behaviours.
For this class of systems, which interact extensively with humans and with their environment, it is thus crucial to guarantee their ``correct'' behaviour in order to avoid unexpected and possibly dangerous situations.
In this paper we will present a framework that allows us to measure the robustness of systems.
This is the ability of a program to tolerate changes in the environmental conditions and preserving the original behaviour.
In the proposed framework, the interaction of a program with its environment is represented as a sequence of random variables describing how both evolve in time.
For this reason, the considered measures will be defined among probability distributions of observed data.
The proposed framework will be then used to define the notions of adaptability and reliability.
The former indicates the ability of a program to absorb perturbation on environmental conditions after a given amount of time. The latter expresses the ability of a program to maintain its intended behaviour (up-to some reasonable tolerance) despite the presence of perturbations in the environment.
Moreover, an algorithm, based on statistical inference, is proposed to evaluate the proposed metric and the aforementioned properties.
We use two case studies to the describe and evaluate the proposed approach.
\end{abstract}

\maketitle

%==============================================
% sec. intro
%===============================================

\section{Introduction}%
\label{sec:introduction}

With the advent of \emph{IoT era}~\cite{K11} we witnessed to the diffusion of a plethora of \emph{smart interconnected devices} that permit
\emph{sensing data} from the \emph{environment} and to \emph{operate} on it.
This revolution stimulated the design of increasingly pervasive systems, often named \emph{Cyber-Physical Systems} (CPSs)~\cite{RLSS10}, that are thought to control and coordinate these devices in order to manage assets and resources, and to guarantee efficient behaviours.
In particular, it is crucial to guarantee that these devices always behave ``correctly'' in their interaction with humans and/or other devices, in order to avoid unexpected and potentially dangerous situations.

Considering that the behaviour of a program controlling the device is largely determined by its extensive interplay with the surrounding environment, this means that the design of CPSs must include some formal tool allowing us to describe how the environment evolves and, if needed, how it reacts to the actions of the program.
For instance, while programming an unmanned vehicle (UV) that has to autonomously set its trajectory to avoid obstacles, we have to consider the possible wind (in case of an aerial UV), or stream (for an underwater UV) scenarios describing the conditions under which the UV will operate.
However, given the dynamical and, sometimes, unpredictable behaviour of the phenomena that constitute the environment, these scenarios cannot be fully deterministic.
Instead, the \emph{evolution of the environment} should be described in terms of a \emph{stochastic process} that gives, at each step, a probabilistic measure of the \emph{observed dynamics} of the environment, including the likelihood of events that could constitute a hazard to our device, like, e.g., a gust of wind in the case of an aerial UV\@.
Being the dynamics continuous, it is natural for the stochastic process to be defined over a \emph{continuous state space}.
At the same time, the program controlling the system might incorporate \emph{discrete random behaviours} to mitigate the uncertainty in the environmental conditions.
As a consequence, the interaction between the program and the environment will be represented in terms of a stochastic process whose states describe how the environmental conditions and the state of the program evolve along a computation.

In the literature, we can find a wealth of proposals of stochastic and probabilistic models, as, e.g., Probabilistic Automata~\cite{S95}, Stochastic Process Algebras~\cite{HHHMR94,CH08,BNL13}, Labelled Markov Chains and Stochastic Hybrid Systems~\cite{CL07,HLS00}, \emph{Markov Decision Processes}~\cite{Put05}, and \emph{ad hoc} solutions for specific application contexts, as, e.g., establishing safety guarantees for drones flying under particular circumstances~\cite{VNGV16,HJSLVBAO14}.
Yet, in these studies, either the environment is not explicitly taken into account, or it is modelled only deterministically.
In addition to that, due to the variety of applications and heterogeneity of systems, no general formal framework to deal with these challenges has been proposed so far.
The lack of concise abstractions and of an automatic support makes the analysis and verification of the considered systems difficult, laborious, and error prone.
Hence, as a first contribution, in this paper we propose a semantic framework allowing us to model the behaviour of CPSs taking into account the continuous dynamics of the environment.

We can then focus on the analysis of the obtained behaviour.
One of the main challenges here is to quantify how \emph{variations in the environmental conditions} can affect the \emph{expected behaviour} of the system.
In this paper we will present a framework that allows us to \emph{measure the robustness} of systems.
This is the ability of a program to \emph{tolerate} changes in the environmental conditions and preserving the original behaviour.
Since the interaction of the program with the enclosing environment is represented as a sequence of random variables, the considered measures will be defined among probability distributions of observed data.
The proposed framework will then be used to define the robustness notions of \emph{adaptability} and \emph{reliability}.
As an example, if we consider the above mentioned aerial UV scenario, a perturbation can be caused by a gust of wind that moves the UV out of its trajectory.
We will say that the program controlling the UV is \emph{adaptable} if it can retrieve the initial trajectory within a suitable amount of time.
In other words, we say that a program is adaptable if no matter how much its behaviour is initially affected by the perturbations, it is able to react to them and regain its intended behaviour within a given amount of time.
On the other hand, it may be the case that the UV is able to detect the presence of a gust of wind and can oppose to it, being only slightly moved from its initial trajectory.
In this case, we say that the program is \emph{reliable}.
Hence, \emph{reliability} expresses the ability of a program to maintain its intended behaviour (up-to some reasonable tolerance) despite the presence of perturbations in the environment.

%================================================================

\subsection*{Detailed outline of paper contributions}

In this paper we first propose a semantic framework that permits modelling the behaviour of programs operating in an evolving environment.
In the proposed approach, the interaction among these two elements is represented in a purely \emph{data-driven} fashion.
In fact, while the environmental conditions are (partially) available to the program as a set of data, allowing it to adjust its behaviour to the current situation, the program is also able to use data to (partially) control the environment and fulfil its tasks.
The program-environment interplay is modelled in terms of the changes they induce on a set of application-relevant data, henceforth called \emph{\dataspace}.
This approach will allow for a significant simplification in modelling the behaviour of the program, which can be \emph{isolated} from that of the environment.
Moreover, as common to favour \emph{computational tractability}~\cite{AIB11,AKLP10}, a \emph{discrete time} approach is adopted.

In our model, a \emph{system} consists in \emph{three distinct components}:
\begin{enumerate}
\item a \emph{process} $\proc$ describing the behaviour of the program,
\item a \emph{\datastate{}} $\ds$ describing the current state of the \dataspace, and
\item an \emph{environment evolution} $\E$ describing the effect of the environment on $\ds$.
\end{enumerate}
As we focus on the interaction with the environment, we abstract from the internal computation of the program and model only its activity on $\ds$.
At each step, a process can \emph{read/update} values in $\ds$ and $\E$ applies on the resulting \datastate{}, providing a new \datastate{} at the next step.
To deal with the uncertainties, we introduce \emph{probability} at two levels:
\begin{enumerate}[(i)]
\item we use the \textbf{\emph{discrete}} \emph{generative probabilistic model}~\cite{vGSS95} to define processes, and
\item $\E$ induces a \textbf{\emph{continuous}} \emph{distribution} over \datastates.
\end{enumerate}
The behaviour of the system is then entirely expressed by its \emph{\traccione}, i.e., the sequence of distributions over \datastates{} obtained at each step.
Given the novelties of our model, as a side contribution we show that this behaviour defines a \emph{Markov process}.

\begin{rem}%
\label{rmk:continuity}
We remark that it is not possible to replace the continuous distribution induced by $\E$ with a discrete one without introducing limiting assumptions on the behaviour of the environment.
In fact, in the continuous setting, each \datastate{} in the support of the distribution is potentially reachable.
Conversely, when a discrete (or finite state) setting is considered, it is necessary to partition the \dataspace, so that some \datastates{} have to be identified in a single state, or they simply become unreachable.
Hence, a simplification to the discrete setting would result in a loss of information on the behaviour of the environment, thus making the model unfeasible to deal with the slight modifications on data induced by the uncertainties.
\end{rem}

The second contribution is the definition of a \emph{metric semantics} in terms of a (time-dependent) distance on the \tracciones{} of systems, which we call the \emph{\spell metric}.
The \emph{\spell metric} will allow us to:
\begin{enumerate}
\item verify how well a program is fulfilling its tasks by comparing it with its specification,
\item\label{item:two_programs} compare the activity of different programs in the same environment,
\item\label{item:two_environments} compare the behaviour of one program with respect to different environments and changes in the initial conditions.
\end{enumerate}
The \spell metric will consist of two components: a \emph{metric on \datastates{}} and the \emph{Wasserstein metric}~\cite{W69}.
The former is defined in terms of a (time-dependent) \emph{penalty function} allowing us to compare two \datastates{} only on the base of the objectives of the program.
The latter lifts the metric on \datastates{} to a metric on distributions on \datastates.
We then obtain a metric on \tracciones{} by considering the maximal of the Wasserstein distances over time.
The proposed metric is then used to formalise the notions of \emph{robustness}, \emph{adaptability}, and \emph{reliability}.

\begin{rem}
Notice that the analysis of behaviour presented in items (\ref{item:two_programs}) and (\ref{item:two_environments}) above are made possible by our choice of deriving the Markov process from the specification of the program $P$ and the environment $\E$, instead of taking it as given.
Hence, it is enough for us to substitute the specification of a process (respectively, the environment) with another one, and consider the \traccione{} induced by it to obtain the new behaviour.
This is what we do, for instance, in the case study in Section~\ref{sec:engine}, where we compare a genuine program with a compromised one working in the same environment.
\end{rem}

A \emph{randomised algorithm} that permits estimating the \tracciones{} of systems and thus for the evaluation of the \spell metric is then introduced.
Following~\cite{TK09}, the Wasserstein metric is evaluated in time $O(N \log N)$, where $N$ is the (maximum) number of samples.
We already adopted this approach in~\cite{CLT20a} in the context of finite-states self-organising collective systems, without any notion of environment or data space.
Moreover, we will also show how the proposed algorithm can be used to estimate the \emph{robustness}, \emph{adaptability}, and \emph{reliability}, of a system.

%========================================

\subsection*{Organisation of contents}

After reviewing some background mathematical notions in Section~\ref{sec:background}, we devote Section~\ref{sec:calcolo} to the presentation of our \emph{model}.
In Section~\ref{sec:metriche} we discuss the \emph{\spell metric} and all the ingredients necessary to define it.
Then, in Section~\ref{sec:computing} we provide the \emph{algorithm}, based on statistical inference, for the evaluation of the metric.
The notions of \emph{adaptability} and \emph{reliability} of programs are formally introduced in Section~\ref{sec:properties} together with an example of application of our algorithm to their evaluation. In these sections a running scenario is used to clarify the proposed approach. This consists of a stochastic variant of the three-tanks laboratory experiment described in~\cite{RKOMI97}.
In Section~\ref{sec:engine} our methodology is used to analyse the \emph{engine system} proposed in~\cite{LMMT21}. This consists of two supervised self-coordinating refrigerated engine systems that are subject to cyber-physical attacks aiming to inflict \emph{overstress of equipment}~\cite{GGIKLW2015}.
Finally, Section~\ref{sec:conclusion} concludes the paper with a discussion on related work and future research directions.

%===============================================
% sec- back
%===============================================

\section{Background}%
\label{sec:background}

In this section we introduce the mathematical background on which we build our contribution.
We remark that we present in detail only the notions that are fundamental to allow any reader to understand our work, like, e.g., those of probability measure and metric.
Other notions that are needed exclusively to guarantee that the Mathematical constructs we use are well-defined (e.g., the notion of Radon measure for Wasserstein hemimetrics) are only mentioned.
The interested reader can find all the formal definitions in any Analysis textbook.

%====================================

\subsection*{Measurable spaces and measurable functions}

A \emph{$\sigma$-algebra} over a set $\Omega$ is a family $\F$ of subsets of $\Omega$ such that:
\begin{enumerate*}
\item $\Omega \in \F$,
\item $\F$ is closed under complementation, and
\item $\F$ is closed under countable union.
\end{enumerate*}
The pair $(\Omega, \Sigma)$ is called a \emph{measurable space} and the sets in $\Sigma$ are called \emph{measurable sets}, ranged over by $\mesA,\mesB,\dots$.

For an arbitrary family $\Phi$ of subsets of $\Omega$, the $\sigma$-algebra \emph{generated} by $\Phi$ is the smallest $\sigma$-algebra over $\Omega$ containing $\Phi$.
In particular, we recall that given a topology $T$ over $\Omega$, the \emph{Borel $\sigma$-algebra} over $\Omega$, denoted $\borel(\Omega)$, is the $\sigma$-algebra generated by the open sets in $T$.
%\footnote{\textcolor{blue}{It is a common convention to identify a topological space $(\Omega,T)$ with the set $\Omega$ when $T$ is fixed, or unspecified, and thus to write $\borel(\Omega)$ in place of $\borel(\Omega,T)$ (see, e.g.,~\cite{Bo07}, page 1).}}.
For instance, given $n\in \nats^+$, $\borel(\reals ^n)$ is the $\sigma$-algebra generated by the open intervals in $\reals^n$.

Given two measurable spaces $(\Omega_i,\Sigma_i)$, $i = 1,2$, the \emph{product $\sigma$-algebra} $\Sigma_1 \otimes \Sigma_2$ is the $\sigma$-algebra on $\Omega_1 \times \Omega_2$ generated by the sets $\{\mesA_1 \times \mesA_2 \mid \mesA_i \in \Sigma_i\}$.

Given measurable spaces $(\Omega_1,\Sigma_1), (\Omega_2,\Sigma_2)$,
a function $f \colon \Omega_1 \to \Omega_2$ is said to be $\Sigma_1$-\emph{measurable} if $f^{-1}\!(\mesA_2) \!\in\! \Sigma_1$ for all $\mesA_2 \!\in\! \Sigma_2$, with $f^{-1}(\mesA_2) \! = \! \{\omega \!\in\! \Omega_1 \mid f(\omega) \in \mesA_2\}$.

%====================================

\subsection*{Probability spaces and random variables}

A \emph{probability measure} on a measurable space $(\Omega,\Sigma)$ is a function $\mu \colon \Sigma \to [0,1]$ such that $\mu(\Omega) = 1$, $\mu(\mesA) \ge 0$ for all $\mesA \in \Sigma$, and $\mu( \bigcup_{i \in I} \mesA_i) = \sum_{i \in I}\mu(\mesA_i)$ for every countable family of pairwise disjoint measurable sets $\{\mesA_i\}_{i\in I} \subseteq \Sigma$.
Then $(\Omega,\Sigma, \mu)$ is called a \emph{probability space}.
We let $\distrib{(\Omega,\F)}$ denote the set of all probability measures over $(\Omega,\F)$.

For $\omega \in \Omega$, the \emph{Dirac measure} $\dirac_{\omega}$ is defined by $\dirac_\omega (\mesA) = 1$, if $\omega \in \mesA$, and $\dirac_{\omega}(\mesA) = 0$, otherwise, for all $\mesA \in \F$.
Given a countable set of reals $(p_i)_{i \in I}$ with $p_i \ge 0$ and $\sum_{i \in I}p_i = 1$, the \emph{convex combination} of the probability measures $\{\mu_i\}_{i \in I} \subseteq \distrib{(\Omega,\F)}$ is the probability measure $\sum_{i \in I} p_i \cdot \mu_i$ in $\distrib{(\Omega,\F)}$ defined by $(\sum_{i \in I} p_i \cdot \mu_i)(\mesA) = \sum_{i \in I} p_i \mu_i(\mesA)$, for all $\mesA \in \F$.
A probability measure $\mu \in \distrib{(\Omega,\F)}$ is called \emph{discrete} if $\mu=\sum_{i \in I}p_i \cdot \dirac_{\omega_i}$, with $\omega_i \in \Omega$, for some countable set of indexes $I$.
The \emph{support} of a discrete probability measure $\mu$ is defined as $\support(\mu) = \{\omega \in \Omega \mid \mu(\omega) > 0\}$.

Assume a probability space $(\Omega, \F, \mu)$ and a measurable space $(\Omega', \F')$.
A function $X \colon \Omega \to \Omega'$ is called a \emph{random variable} if it is $\Sigma$-measurable.
The \emph{distribution measure}, or \emph{cumulative distribution function} (cdf), of $X$ is the probability measure $\mu_X$ on $(\Omega',\F')$ defined by $\mu_X(\mesA)=\mu(X^{-1}(\mesA))$ for all $\mesA \in \F'$.
Given random variables $X_i$ from $(\Omega_i,\F_i, \mu_i)$ to $(\Omega'_i,\F'_i)$, $i=1,\ldots,n$, the collection $\rv = [X_1,\ldots,X_n]$ is called a \emph{random vector}.
The cdf of a random vector $\rv$ is given by the \emph{joint} distribution of the random variables in it.

\begin{nota}
With a slight abuse of notation, in the examples and explanations throughout the paper, we will sometimes use directly the cdf of a random variable rather than formally introducing the probability measure defined on the domain space.
In fact, since we will consider Borel sets over $\real^n$, for some $n \ge 1$, it is usually easier to define the cdfs than the probability measures on them.
Similarly, when the cdf is absolutely continuous with respect to the Lebesgue measure, then we shall reason directly on the \emph{probability density function} (pdf) of the random variable, namely the Radon-Nikodym derivative of the cdf with respect to the Lebesgue measure.

Consequently, we shall also henceforth use the more suggestive term \emph{\bf distribution} in place of the terms probability measure, cdf and pdf.
\end{nota}

%==================================

\subsection*{The Wasserstein hemimetric}

A \emph{metric} on a set $\Omega$ is a function $m \colon \Omega \times \Omega \to \real^{\ge0}$ s.t.\ $m(\omega_1,\omega_2) = 0$ if{f} $\omega_1 = \omega_2$, $m(\omega_1,\omega_2) = m(\omega_2,\omega_1)$, and $m(\omega_1,\omega_2) \le m(\omega_1,\omega_3) + m(\omega_3,\omega_2)$, for all $\omega_1,\omega_2,\omega_3 \in \Omega$.
We obtain a \emph{pseudometric} by relaxing the first property to $m(\omega_1,\omega_2) = 0$ if $\omega_1=\omega_2$.
Then, a \emph{hemimetric} is a pseudometric that is not necessarily symmetric.
A (pseudo-, hemi-)metric $m$ is $l$-\emph{bounded} if $m(\omega_1,\omega_2) \le l$ for all $\omega_1,\omega_2 \in\Omega$. % chktex 36
For a (pseudo-, hemi-)metric on $\Omega$, the pair $(\Omega,m)$ is a (\emph{pseudo}-, \emph{hemi-})\emph{metric space}. % chktex 36

Given a (pseudo-, hemi-)metric space $(\Omega,m)$, the (pseudo-, hemi-)metric $m$ induces a natural topology over $\Omega$, namely the topology generated by the open $\varepsilon$-balls, for $\varepsilon > 0$, $B_{m}(\omega,\varepsilon) = \{\omega' \in \Omega \mid m(\omega,\omega') < \varepsilon\}$. % chktex 36
We can then naturally obtain the Borel $\sigma$-algebra over $\Omega$ from this topology.

In this paper we are interested in defining a \emph{hemimetric on distributions}.
To this end we will make use of the Wasserstein lifting~\cite{W69}, which evaluates the infimum, with respect to the joint distributions, of the expected value of the distance over the two distributions, and whose definition is based on the following notions and results.
Given a set $\Omega$ and a topology $T$ on $\Omega$, the topological space $(\Omega,T)$ is said to be \emph{completely metrisable} if there exists at least one metric $m$ on $\Omega$ such that $(\Omega, m)$ is a complete metric space and $m$ induces the topology $T$.
A \emph{Polish space} is a separable completely metrisable topological space.
In particular, we recall that:
\begin{enumerate*}[(i)]
\item $\real$ is a Polish space; and
\item every closed subset of a Polish space is in turn a Polish space.
\end{enumerate*}
Moreover, for any $n \in N$, if $\Omega_1, \dots, \Omega_n$ are Polish spaces, then the Borel $\sigma$-algebra on their product coincides with the product $\sigma$-algebra generated by their Borel $\sigma$-algebras, namely
\[
\borel( \bigtimes_{i=1}^n \Omega_i) = \bigotimes_{i=1}^n \borel(\Omega_i).
\]
(This is proved, e.g., in~\cite{Bo07} as Lemma 6.4.2 whose hypothesis are satisfied by Polish spaces since they are second countable.)
These properties of Polish spaces are interesting for us since they guarantee that all the distributions we consider in this paper are Radon measures and, thus, the Wasserstein lifting is well-defined on them.
For this reason, we also directly present the Wasserstein hemimetric by considering only distributions on Borel sets.

\begin{defi}
[Wasserstein hemimetric]%
\label{def:Wasserstein}
Consider a Polish space $\Omega$ and let $m$ be a hemimetric on $\Omega$.
For any two distributions $\mu$ and $\nu$ on $(\Omega,\borel(\Omega))$, the \emph{Wasserstein lifting} of $m$ to a distance between $\mu$ and $\nu$ is defined by
\[
\Wasserstein(m)(\mu,\nu) = \inf_{\w \in \W(\mu,\nu)} \int_{\Omega \times \Omega} m(\omega,\omega') \dd\w(\omega,\omega')
\]
where $\W(\mu,\nu)$ is the set of the \emph{couplings of $\mu$ and $\nu$}, namely the set of joint distributions $\w$ over the product space $(\Omega \times \Omega, \borel(\Omega \times \Omega))$ having $\mu$ and $\nu$ as left and right marginal, respectively, namely $\w(\mesA \times \Omega) = \mu(\mesA)$ and $\w(\Omega \times \mesA) = \nu(\mesA)$, for all $\mesA \in \borel(\Omega)$.
\end{defi}

Despite the original version of the Wasserstein distance being defined on a metric on $\Omega$, the Wasserstein hemimetric given above is well-defined.
We refer the interested reader to~\cite{FR18} and the references therein for a formal proof of this fact.
In particular, the Wasserstein hemimetric is given in~\cite{FR18} as Definition 7 (considering the compound risk excess metric defined in Equation (31) of that paper), and Proposition 4 in~\cite{FR18} guarantees that it is indeed a well-defined hemimetric on $\distrib(\Omega,\borel(\Omega))$.
Moreover, Proposition 6 in~\cite{FR18} guarantees that the same result holds for the hemimetric $m(x,y) = \max\{y-x,0\}$ which, as we will see, plays an important role in our work (cf.\ Definition~\ref{def:metric_DS} below).

\begin{rem}
As elsewhere in the literature, for simplicity and brevity, we shall henceforth use the term \emph{metric} in place of the term hemimetric.
\end{rem}

%==========================================
% sec - Calculus
%==========================================

\section{The model}%
\label{sec:calcolo}

We devote this section to introduce the three components of our systems, namely the \emph{\dataspace} $\D$, the \emph{process} $\proc$ describing the behaviour of the program, and the \emph{environment evolution} $\E$ describing the effects of the environment.
As already discussed in the Introduction, our choice of modelling program and environment separately will allow us to isolate programs from the environment and favour their analysis.

In order to help the reader grasp the details of our approach, we introduce a stochastic variant of the three-tanks laboratory experiment, that embodies the kind of program-environment interactions we are interested in.
Several variants (with different number of tanks) of the $n$-tanks experiment have been widely used in control program education
(see, among the others,~\cite{AL92,RKOMI97,Jo00,ALGAC06}).
Moreover, some recently proposed cyber-physical security testbeds, like \emph{SWaT}~\cite{MT16,AGAOT17}, can be considered as an evolution of the tanks experiment.
Here, we consider a variant of the three-tanks laboratory experiment
described in~\cite{RKOMI97}.

Please notice that this example is meant to serve as a toy example allowing us to showcase the various features of our tool, without carrying out a heavy mathematical formulation.
We delay to Section~\ref{sec:engine} the presentation of a case-study showing how our framework can be applied to the analysis of (more complex) real-world systems.

\begin{nota}
In the upcoming examples we will slightly abuse of notation and use a variable name $x$ to denote all: the variable $x$, the (possible) function describing the evolution in time of the values assumed by $x$, and the (possible) random variable describing the distribution of the values that can be assumed by $x$ at a given time.
The role of the variable name $x$ will always be clear from the context.
\end{nota}

\begin{figure}[tbp]
\centering
\begin{tikzpicture}
\draw[-](0,4)--(0,0.5);
\draw[-](0,0.5)--(9.7,0.5);
\draw[-](10.7,0.5)--(11,0.5);
\draw[-](2.5,4)--(2.5,1);
\draw[-](2.5,1)--(3.5,1);
\draw[-](3.5,1)--(3.5,4);
\draw[-](6,4)--(6,1);
\draw[-](6,1)--(7,1);
\draw[-](7,1)--(7,4);
\draw[-](9.5,4)--(9.5,1);
\draw[-](9.5,1)--(9.7,1);
\draw[-](10.7,1)--(11,1);
\node at (1.25,0){Tank 1};
\node at (4.75,0){Tank 2};
\node at (8.25,0){Tank 3};
\draw[-latex](2.75,0.6)--(3.25,0.6);
\node at (3,0.8){$q_{12}$};
\draw[-latex](6.25,0.6)--(6.75,0.6);
\node at (6.5,0.8){$q_{23}$};
\draw[-latex](10.8,0.6)--(11.3,0.6);
\node at (11.05,0.8){$q_3$};
\draw[dashed,thick,red](-0.1,3.5)--(2.5,3.5);
\node at (-0.35,3.7){\textcolor{red}{$l_{max}$}};
\draw[dashed,thick,red](-0.1,0.6)--(2.5,0.6);
\node at (-0.35,0.8){\textcolor{red}{$l_{min}$}};
\draw[dashed,thick,red](3.5,3.5)--(6,3.5);
\draw[dashed,thick,red](3.5,0.6)--(6,0.6);
\draw[dashed,thick,red](7,3.5)--(9.5,3.5);
\draw[dashed,thick,red](7,0.6)--(9.5,0.6);
\draw[dashed,thick,ForestGreen](0,2.5)--(2.6,2.5);
\node at (2.9,2.7){\textcolor{ForestGreen}{$l_{goal}$}};
\draw[dashed,thick,ForestGreen](3.5,2.5)--(6,2.5);
\draw[dashed,thick,ForestGreen](7,2.5)--(9.5,2.5);
\draw[blue](-0.1,3)--(2.5,3);
\node at(-0.2,3.2){\textcolor{blue}{$l_1$}};
\draw[blue](3.4,2)--(6,2);
\node at(3.2,2.2){\textcolor{blue}{$l_2$}};
\draw[blue](6.9,2.2)--(9.5,2.2);
\node at(6.8,2.4){\textcolor{blue}{$l_3$}};
\draw[-](0.8,6)--(0.8,5.5);
\draw[-](1.7,6)--(1.7,5.5);
\draw[-](0.2,5.5)--(2.3,5.5);
\draw[-](0.2,5.5)--(0.2,4.5);
\draw[-](2.3,5.5)--(2.3,4.5);
\draw[-](0.2,4.5)--(2.3,4.5);
\draw[-](0.8,4.5)--(0.8,4);
\draw[-](1.7,4.5)--(1.7,4);
\node at (1.25,5){pump};
\draw[-latex](1.25,4.25)--(1.25,3.75);
\node at (1.45,4.1){$q_1$};
\draw[->](3.5,5)--(2.5,5);
\node at (3.5,5.8){\scalebox{0.85}{Flow rate under}};
\node at (3.5,5.5){\scalebox{0.85}{the control of}};
\node at (3.5,5.2){\scalebox{0.85}{the software}};
\draw[-](7.8,6)--(7.8,4);
\draw[-](8.7,6)--(8.7,4);
\draw[-latex](8.25,4.25)--(8.25,3.75);
\node at (8.45,4.1){$q_2$};
\draw[->](6.5,5)--(7.5,5);
\node at (6.5,5.8){\scalebox{0.85}{Flow rate under}};
\node at (6.5,5.5){\scalebox{0.85}{the control of}};
\node at (6.5,5.2){\scalebox{0.85}{the environment}};
\draw[-](9.7,0.3)--(9.7,1.2);
\draw[-](10.7,0.3)--(10.7,1.2);
\draw[-](9.7,0.3)--(10.7,0.3);
\draw[-](9.7,1.2)--(10.7,1.2);
\node at (10.2,0.75){pump};
\draw[->](10.2,2.3)--(10.2,1.3);
\node at (10.8,3.1){\scalebox{0.85}{Flow rate under}};
\node at (10.8,2.8){\scalebox{0.85}{the control of}};
\node at (10.8,2.5){\scalebox{0.85}{the software}};
\draw[fill,cyan,opacity=0.2] (1,5.5) rectangle (1.5,6);
\draw[fill,cyan,opacity=0.2] (1.15,4.1) rectangle (1.35,4.5);
\draw[fill,cyan,opacity=0.2] (8,4.1) rectangle (8.5,6);
\draw[fill,cyan,opacity=0.2] (0,0.5) rectangle (2.5,3);
\draw[fill,cyan,opacity=0.2] (3.5,0.5) rectangle (6,2);
\draw[fill,cyan,opacity=0.2] (7,0.5) rectangle (9.5,2.2);
\draw[fill,cyan,opacity=0.2] (2.5,0.5) rectangle (3.5,1);
\draw[fill,cyan,opacity=0.2] (6,0.5) rectangle (7,1);
\draw[fill,cyan,opacity=0.2] (9.5,0.5) rectangle (9.7,1);
\draw[fill,cyan,opacity=0.2] (10.7,0.5) rectangle (11,0.8);
\end{tikzpicture}
\caption{Schema of the three-tanks scenario.}%
\label{fig:threetanks}
\end{figure}
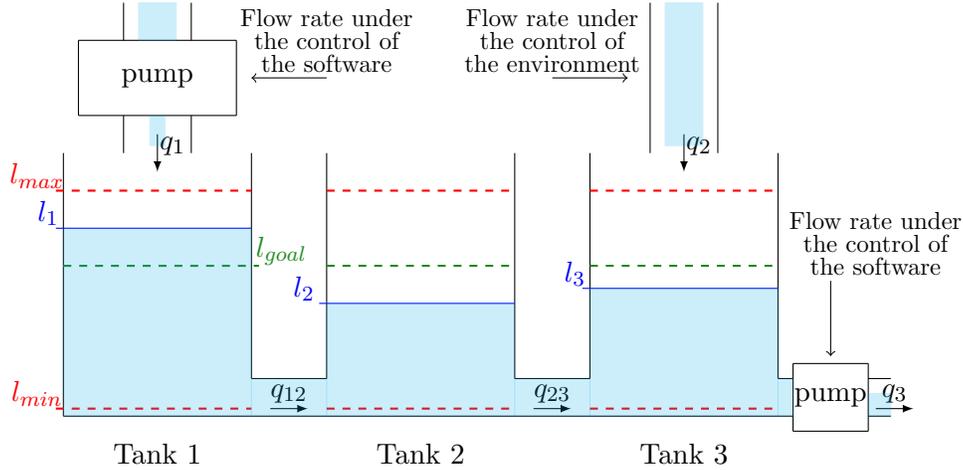

\begin{exa}%
\label{ex:tanks_I}
As schematised in Figure~\ref{fig:threetanks}, there are three identical tanks connected by two pipes.
Water enters in the first and in the last tank by means of a pump and an incoming pipe, respectively.
The last tank is equipped with an outlet pump.
We assume that water flows through the incoming pipe with a rate $q_2$ that is determined by the environment, whereas the flow rates through the two pumps (respectively, $q_1$ and $q_3$) are under the control of the program controlling the experiment.
The rates $q_1,q_2,q_3$ can assume values in the range $[0,q_{max}]$, for a given a maximal flow rate $q_{max}$.
The task of the program consists in guaranteeing that the levels of water in the three tanks fulfil some given requirements.

The level of water in tank $i$ at time $\tau$ is denoted by $l_i(\tau)$, for $i=1,2,3$, and is always in the range $[l_{min},l_{max}]$, for suitable $l_{min}$ and $l_{max}$ giving, respectively, the minimum and maximum level of water in the tanks.

The principal difference between our version of the three-tanks and its classical formulation is the flow rate $q_2$.
In fact, we assume that $q_2$ cannot be controlled by the program, that can only indirectly observe the effects it has on $l_3$ and react to those.
In our setting, this is equivalent to say that $q_2$ is under the control of the environment.
Moreover, to obtain a more complex scenario, we assume that the value of $q_2$ (at a given time) can be described only probabilistically.
This can be interpreted either as the effect of uncertainties (like, e.g., measurement errors of the sensors having to determine $q_2$), or as the consequence of an unpredictable behaviour of the environment.
Consequently, the exact value of $q_2(\tau)$ is unknown and will be evident only at execution time.
Still, we can consider different scenarios that render the assumptions we might have on the environment.
For instance, we can assume that the flow rate of the incoming pipe is normally distributed with mean $q_{med} \in [0,q_{max}]$ and variance $\Delta_q > 0$, expressed, assuming sampling time interval $\Delta\tau$, by:
\begin{equation}%
\label{eq:scenario1}
\begin{array}{rcl}
q_2(\tau+\Delta\tau) & \sim & N(q_{med},\Delta_q)
\enspace .
\end{array}
\end{equation}
In a more elaborated scenario, we could assume that $q_2$ varies at each step by a value $v$ that is normally distributed with mean $0$ and variance $1$.
In this case, we have:
\begin{equation}%
\label{eq:scenario2}
\begin{array}{rcl}
v(\tau) & \sim & N(0,1) \\
q_2(\tau+\Delta\tau) & = & \min\left\{ \max\left\{ 0 , q_2(\tau)+v(\tau) \right\} , q_{max} \right\}
\enspace .
\end{array}
\end{equation}

Conversely, the two pumps are controlled by the program that, by reading the values of $l_1(\tau)$ and $l_3(\tau)$, can select the value of $q_1(\tau+\Delta\tau)$ and $q_3(\tau+\Delta\tau)$.
In detail, the equations describing the behaviour of the program are the following:
\begin{equation}%
\label{eq:q1}
q_1(\tau+\Delta\tau) \! = \! \left\{ \!\!
\begin{array}{ll}
\max\{ 0 , q_1(\tau) - q_{step} \} & \! \text{if } l_1(\tau)>l_{goal}+\Delta_l \\
\min\{ q_{max} ,  q_1(\tau) + q_{step} \} & \! \text{if } l_1(\tau)<l_{goal}-\Delta_l\\
q_1(\tau) & \! \text{otherwise;}
\end{array}
\right.
\end{equation}
\begin{equation}%
\label{eq:q3}
q_3(\tau+\Delta\tau) \! = \! \left\{ \!\!
\begin{array}{ll}
\min\{ q_{max} , q_3(\tau) + q_{step} \} & \! \text{if } l_3(\tau)>l_{goal}+\Delta_l \\
\max\{ 0 ,  q_3(\tau) - q_{step} \} & \! \text{if } l_3(\tau)<l_{goal}-\Delta_l\\
q_3(\tau) & \! \text{otherwise.}
\end{array}
\right.
\end{equation}
Above, $l_{goal}$ is the desired level of water in the tanks, while $q_{step} \in \ropen{0,q_{max}}$ is the variation of the flow rate that is controllable by the pump in one time step $\Delta\tau$.
The idea behind Equation~\eqref{eq:q1} is that when $l_1$ is greater than $l_{goal}+\Delta_l$, the flow rate of the pump is decreased by $q_{step}$.
Similarly, when $l_1$ is less than $l_{goal}-\Delta_l$, that rate is increased by $q_{step}$.
The reasoning for Equation~\eqref{eq:q3} is symmetrical.
In both cases, we use the threshold $\Delta_l > 0$ to avoid continuous contrasting updates.

From now on, we will say that we are in \emph{scenario 1} if $q_2$ is determined by Equation (\ref{eq:scenario1}), and that we are in \emph{scenario 2} in case $q_2$ is determined by Equation (\ref{eq:scenario2}).
In both scenarios, $q_1$ is determined by Equation (\ref{eq:q1}) and $q_3$ by Equation (\ref{eq:q3}).

The dynamics of $l_i(\tau)$ can be modelled via the following set of stochastic difference equations, with sampling time interval $\Delta\tau$:
\begin{equation}%
\label{eq:tankstynamics}
\begin{array}{rcl}
l_1(\tau+\Delta\tau) & = & l_1(\tau) + \Delta\tau \cdot (q_1(\tau) - q_{12}(\tau))\\
l_2(\tau+\Delta\tau) & = & l_2(\tau) + \Delta\tau \cdot (q_{12}(\tau) - q_{23}(\tau))\\
l_3(\tau+\Delta\tau) & = & l_3(\tau) + \Delta\tau \cdot (q_2(\tau) + q_{23}(\tau) - q_3(\tau))
\end{array}
\end{equation}
where $q_1,q_2,q_3$ are as above (i.e., the flow rates through, respectively, the pump connected to the first tank, the incoming pipe, and the outlet pump), and $q_{ij}$ denotes the flow rate from tank $i$ to tank $j$.
The evaluation of $q_{ij}(\tau)$ depends on $l_i(\tau),l_j(\tau)$, and on the physical dimensions of the tanks, and it is obtained as follows.
Let $A$ be the cross sectional area of each tank, and $a$ be the cross sectional area of the connecting and outlet pipes.
The volume balance difference equations of each tank are then the following:
\begin{equation}%
\label{eq:volumes}
\begin{array}{rcl}
\displaystyle\frac{A}{\Delta\tau}\cdot(l_1(\tau+\Delta\tau)-l_1(\tau)) & = & q_1(\tau) - q_{12}(\tau) \\[.3cm]
\displaystyle\frac{A}{\Delta\tau}\cdot(l_2(\tau+\Delta\tau)-l_2(\tau)) & = & q_{12}(\tau) - q_{23}(\tau) \\[.3cm]
\displaystyle\frac{A}{\Delta\tau}\cdot(l_3(\tau+\Delta\tau)-l_3(\tau)) & = & q_2(\tau) + q_{23}(\tau) - q_{3}(\tau). \\
\end{array}
\end{equation}
We then apply Torricelli's law to the equations in~\eqref{eq:volumes} to obtain the flow rates $q_{12}$ and $q_{23}$:
\begin{equation}%
\label{eq:pipes}
\begin{split}
& q_{12}(\tau) =
\begin{cases}
a_{12}(\tau) \cdot a \cdot \sqrt{2g\big(l_1(\tau) - l_2(\tau)\big)} & \text{ if } l_1(\tau)\ge l_2(\tau) \\[.1cm]
-a_{12}(\tau) \cdot a \cdot \sqrt{2g\big(l_2(\tau) - l_1(\tau)\big)} & \text{ otherwise;}
\end{cases}
\\
& q_{23}(\tau) =
\begin{cases}
a_{23}(\tau) \cdot a \cdot \sqrt{2g\big(l_2(\tau) - l_3(\tau)\big)} & \text{ if } l_2(\tau)\ge l_3(\tau) \\[.1cm]
-a_{23}(\tau) \cdot a \cdot \sqrt{2g\big(l_3(\tau) - l_2(\tau)\big)} & \text{ otherwise;}
\end{cases}
\end{split}
\end{equation}
where $g$ is the gravitational constant, and $a_{12},a_{23}$ are the loss coefficients of the respective pipes.
These coefficients are represented as time dependent functions since they depend on the geometry of the pipes and on the water level in the tanks.
In our experiments, we use the approximation $g = 9.81$, and the following values for the aforementioned coefficients:
\begin{enumerate*}[(i)]
\item $a = 0.5$,
\item $a_{12} = 0.75$,
\item $a_{23} = 0.75$.
\end{enumerate*}
\hfill\scalebox{0.85}{$\LHD$}
\end{exa}

We now proceed to formally introduce the three components of our systems and their interaction.

%==================================
% data
%===================================

\subsection{Modelling the \dataspace}%
\label{sec:ds}

We define the \dataspace{} by means of a \emph{finite} set of \emph{variables} $\mathrm{Var}$ representing:
\begin{enumerate*}[(i)]
\item \emph{environmental conditions}, such as pressure, temperature, humidity, etc.,
\item \emph{values perceived by sensors}, which depend on the value of environmental conditions and are unavoidably affected by imprecision and approximations introduced by sensors, and
\item \emph{state of actuators}, which are usually elements in a discrete domain, like $\{\mathit{on},\mathit{off}\}$ (which can be easily mapped to reals as, e.g., $\{0,1\}$).
\end{enumerate*}
For each $x \in \mathrm{Var}$ we assume a measurable space $(\D_x,\borel_x)$, with $\D_x \subseteq \real$ the domain of $x$ and $\borel_x$ the Borel $\sigma$-algebra on $\D_x$.
Without loosing generality, we can assume that each domain $\D_x$ is either a \emph{finite set} or a \emph{compact} subset of $\real$.
In particular, this means that each $\D_x$ is a Polish space.

As $\Var$ is a finite set, we can always assume it to be ordered, namely $\Var=\{x_1,\dots,x_{n}\}$ for a suitable $n \in \nats$.

\begin{defi}
[\Dataspace]
We define the \emph{\dataspace{}} over $\mathrm{Var}$, notation $\D_{\Var}$, as the Cartesian product of the variables domains, namely $\D_{\Var} = \bigtimes_{i = 1}^n \D_{x_i}$.
Then, as a $\sigma$-algebra on $\D_\Var$ we consider the the product $\sigma$-algebra $\borel_{\D_\Var} = \bigotimes_{i=1}^n \borel_{x_i}$.
\end{defi}

When no confusion arises, we will use $\D$ and $\borel_\D$ in place of $\D_{\mathrm{Var}}$ and $\borel_{\D_{\mathrm{Var}}}$, respectively.

\begin{exa}%
\label{ex:tanks_II}
The \dataspace{} $\D_{3T}$ for the considered three-tanks system is defined on the set of variables $\Var=\{q_1,q_2,q_3,l_1,l_2,l_3\}$, with domains $q_i \in \D_q = [0,q_{max}]$ and $l_i \in \D_l = [l_{min},l_{max}]$, for $i =1,2,3$.
We remark that $\Delta_l$ is not included in the \dataspace.
This is due to the fact that, in our setting, $\Delta_l$ is never modified by the environment (it is also constant in time), and it is an inherent property of the program: different programs might have different values of $\Delta_l$, according to how capable they are of controlling and modifying the flow rates, or of reading the levels of water in the tanks.
\hfill\scalebox{0.85}{$\LHD$}
\end{exa}

Elements in $\D$ are $n$-ples of the form $(v_1,\dots,v_n)$, with $v_i \in \D_{x_i}$, which can be also identified by means of functions $\ds \colon \mathrm{Var} \to \real$ from variables to values, with $\ds(x) \in \D_x$ for all $x \in \mathrm{Var}$.
Each function $\ds$ identifies a particular configuration of the data in the \dataspace, and it is thus called a \emph{\datastate}.

\begin{defi}
[\Datastate]%
\label{def:datastate}
A \emph{\datastate{}} is a mapping $\ds \colon \mathrm{Var} \to \real $ from state variables to values, with $\ds(x) \in \D_x$ for all $x \in \mathrm{Var}$.
\end{defi}

For simplicity, we shall write $\ds \in \D$ in place of $(\ds(x_1),\dots,\ds(x_n)) \in \D$.
Since program and environment interact on the basis of the \emph{current} values of data, we have that at each step there is a \datastate{} $\ds$ that identifies the \emph{current state of the \dataspace{}} on which the next computation step is built.
Given a \datastate{} $\ds$, we let $\ds[x=v]$ denote the \datastate{} $\ds'$ associating $v$ with $x$, and $\ds(y)$ with $y$ for any $y \neq x$.
For $\overline{x}=(x_{i_1},\ldots,x_{i_k})$ and $\overline{v}=(v_{i_1},\ldots,v_{i_k})$, we let $\ds[\overline{x}=\overline{v}]$ denote $\ds[x_{i_1}=v_{i_1}]\cdots [x_{i_k}=v_{i_k}]$.

%===========================================================
% processes
%===========================================================

\subsection{Modelling processes}%
\label{sec:proc}

We introduce a simple process calculus allowing us to specify programs that interact with a \datastate{} $\ds$ in a given environment.
We assume that the action performed by a process at a given computation step is determined probabilistically, according to the \emph{generative} probabilistic model~\cite{vGSS95}.

\begin{defi}
[Syntax of processes]%
\label{def:syntax}
We let $\Proc$ be the set of \emph{processes} $\proc$ defined by:
\begin{align*}
\proc \;\;::=\;\; &
\wact{\overline{e}}{\overline{x}}.\proc' \;\;\big|\;\;
\mathrm{if}\; [e]\; \proc_1 \;\mathrm{else}\; \proc_2 \;\;\big|\;\;
\sum_{i \in I} p_i \cdot \proc_i \;\;\big|\;\;
\proc_1 \mathbin{\|}_p \proc_2 \;\;\big|\;\;
\proc_1 \cmerge \proc_2 \;\;\big|\;\;
A\\
e \;\;::=\;\; &
v \in \mathbb{R} \;\;\big|\;\;
x\in \mathrm{Var} \;\;\big|\;\;
op_{k}(e_1,\ldots,e_k)
\end{align*}
where $p,p_1,\ldots$ range over \emph{probability weights} in $[0,1]$, $I$ is finite, $A$ ranges over \emph{process variables}, $op_k$ indicates a \emph{measurable operator} $\mathbb{R}^{k}\rightarrow \mathbb{R}$, and $\overline{\cdot}$ denotes a finite sequence of elements.
We assume that for each process variable $A$ we have a single definition $A \stackrel{\mathit{def}}{=} \proc$.
Moreover, we require that $\sum_{i \in I}p_i =1$ for any process $\sum_{i \in I} p_i \cdot \proc_i$, and that $\var{\proc_1} \cap \var{\proc_2} = \emptyset$ for any process $\proc_1 \cmerge \proc_2$, where $\var{\proc} = \{x \in \Var \mid  x \in \overline{x} \text{ for any } \wact{\overline{e}}{\overline{x}} \text{ occurring in } \proc\}$.
\end{defi}

As usual, the expression $\wact{\overline{e}}{\overline{x}}$ in $\wact{\overline{e}}{\overline{x}}.\proc'$ is called a \emph{prefix}.
In a single action a process can read and update a set of state variables.
This is done by process $\wact{\overline{e}}{\overline{x}}.\proc$ that first evaluates the sequence of expressions $\overline{e}$ in the current \datastate{} and then assigns the results to sequence of variables $\overline{x}$, which are stored in the \datastate{} and evolve according to the environment evolution $\E$ (Definition~\ref{def:environment} below) before being available at next step, when the process will behave as $\proc$.
We use $\surd$ to denote the empty prefix, i.e., the prefix in which the sequences $\overline{e}$ and $\overline{x}$ are empty, meaning that, in the current time step, the process does not read or update variables.
Process $\mathrm{if}\; [e]\; \proc_1 \;\mathrm{else}\; \proc_2$ behaves either as $\proc_1$, if $e$ evaluates to $1$, or as $\proc_2$, if $e$ evaluates to $0$.
Then, $\sum_{i\in I}p_i \cdot \proc_i$ is the \emph{generative probabilistic choice}:
it has probability $p_i$ to behave as $P_i$.
We may write $\sum_{i=1}^{n} p_i \cdot \proc_i$ as $p_1 \cdot \proc_1 + \cdots + p_n \cdot \proc_n$.
The \emph{generative probabilistic interleaving} construct $\proc_1 \mathbin{\|}_p \proc_2$ lets the two argument processes to interleave their actions, where at each step the first process moves with probability $p$ and the second with probability $(1-p)$.
We also provide a \emph{synchronous parallel composition} construct $\proc_1 \cmerge \proc_2$ that lets the two arguments perform their actions in a synchronous fashion, provided $\proc_1$ and $\proc_2$ do not modify the same variables.
Finally, process variable $A$ allows us to specify recursive behaviours by equations of the form $A \stackrel{\mathit{def}}{=} \proc$.
To avoid Zeno behaviours we assume that all occurrences of process variables appear \emph{guarded} by prefixing constructs in $\proc$.
We assume the standard notion of \emph{free} and \emph{bound} process variable and we only consider \emph{closed process terms}, that is terms without \emph{free} variables.

Actions performed by a process can be abstracted in terms of the \emph{effects} they have on the \datastate{}, i.e., via \emph{substitutions} of the form $\theta = [x_{i_1}\leftarrow v_{i_1},\ldots, x_{i_k}\leftarrow v_{i_k}]$, also denoted by $\overline{x}\leftarrow \overline{v}$ for $\overline{x}=x_{i_1},\ldots,x_{i_k}$ and $\overline{v}=v_{i_1},\ldots,v_{i_k}$.
Since in Definition~\ref{def:syntax} process operations on data and expressions are assumed to be measurable, we can model the effects as $\borel_\D$-measurable functions $\theta \colon \D \to \D$ such that $\theta(\ds) = \ds[\overline{x}=\overline{v}]$ whenever $\theta = \overline{x} \leftarrow \overline{v}$.
We denote by $\Theta$ the set of effects.
The behaviour of a process can then be defined by means of a function $\pstep \colon \Proc \times \D \to \distrib(\Theta\times\Proc)$ that given a process $\proc$ and a \datastate{} $\ds$ yields a \emph{discrete} distribution over $\Theta \times \Proc$.
The distributions in $\distrib(\Theta\times\Proc)$ are ranged over by $\pi,\pi',\dots$.
Function $\pstep$ is formally defined in Table~\ref{tab:pstepfunction}.

\begin{table}[tbp]
\begin{displaymath}
\begin{array}{lcl}
(\scalebox{0.9}{PR1}) & \qquad &
\pstep(\wact{\overline{e}}{\overline{x}}.\proc',\ds)
=
\delta_{(\overline{x}\leftarrow \llbracket \overline{e}\rrbracket_{\ds},\proc')}
\\[.1cm]
(\scalebox{0.9}{PR2}) & \qquad &
\pstep(\mathrm{if}\; [e]\; \proc_1 \;\mathrm{else}\; \proc_2, \ds)
=
\begin{cases}
\pstep(\proc_1,\ds) & \text{ if } \llbracket e\rrbracket_{\ds} = 1 \\
\pstep(\proc_2,\ds) & \text{ if } \llbracket e\rrbracket_{\ds} = 0
\end{cases}
\\[.1cm]
(\scalebox{0.9}{PR3}) & \qquad &
\pstep(\sum_{i}p_i \cdot \proc_i,\ds)
=
\sum_{i} p_i\cdot \pstep(\proc_i,\ds)
\\[.1cm]
(\scalebox{0.9}{PR4}) & \qquad &
\pstep(\proc_1 \mathbin{\|}_p \proc_2,\ds)
=
p\cdot (\pstep(\proc_1,\ds)\mathbin{\|}_p\proc_2) +(1-p)\cdot(\proc_1\mathbin{\|}_p \pstep(\proc_2,\ds))
\\[.1cm]
(\scalebox{0.9}{PR5}) & \qquad & \pstep(\proc_1 \cmerge \proc_2, \ds) = \pstep(\proc_1,\ds) \cmerge \pstep(\proc_2,\ds)
\\[.1cm]
(\scalebox{0.9}{PR6}) & \qquad &
\pstep(A,\ds)
=
\pstep(\proc,\ds) \qquad (\text{if } A\stackrel{\mathit{def}}{=}\proc)\\
\end{array}
\end{displaymath}
\caption{\label{tab:pstepfunction} Process Semantics}
\end{table}

In detail, rule (PR1) expresses that in a single action a process can read and update a set of variables:
process $\wact{\overline{e}}{\overline{x}}.\proc$ first evaluates the sequence of expressions $\overline{e}$ in the \datastate{} $\ds$ and then
stores the results to $\overline{x}$.
These \emph{new} data will evolve according to $\E$ before being available at next step, when the process will behave as $\proc$.
The evaluation of an expression $e$ in $\ds$, denoted by $\llbracket e\rrbracket_{\ds}$, is obtained by replacing each variable $x$ in $e$ with $\ds(x)$ and by evaluating the resulting ground expression.
Then, process $\wact{\overline{e}}{\overline{x}}.\proc$ evolves to the Dirac's distribution on the pair $(\overline{x}\leftarrow \llbracket \overline{e}\rrbracket_{\ds},\proc)$.
Rule (PR2) models the behaviour of $\mathrm{if}\; [e]\; \proc_1 \;\mathrm{else}\; \proc_2$.
We assume that boolean operations, like those in the $\mathrm{if}$-guard, evaluate to $1$ or $0$, standing respectively for $\top$ and $\bot$.
Hence, $\mathrm{if}\; [e]\; \proc_1 \;\mathrm{else}\; \proc_2$ behaves as $\proc_1$ if $\llbracket e \rrbracket_{\ds} = 1$, and it behaves as $\proc_2$ otherwise.
Rule (PR3) follows the \emph{generative} approach to probabilistic choice: the behaviour of process $\proc_i$ is selected with probability $p_i$.
Rule (PR4) states that process $\proc_1 \mathbin{\|}_p \proc_2$ interleaves the moves by $\proc_1$ and $\proc_2$, where at each step $\proc_1$ moves with probability $p$ and $\proc_2$ with probability $(1-p)$.
Given $\pi \in \distrib(\Theta \times \Proc)$, we let $\pi \mathbin{\|}_p \proc$ (resp. $\proc \mathbin{\|}_p \pi$) denote the probability distribution $\pi'$ over $(\Theta \times \Proc)$ such that: $\pi'(\theta,\proc')=\pi(\theta,\proc'')$, whenever $\proc'=\proc'' \mathbin{\|}_p \proc$ (resp. $\proc'=\proc \mathbin{\|}_p \proc''$), and $0$, otherwise.
Rule (PR5) gives us the semantics of the \emph{synchronous} parallel composition: for $\pi_1,\pi_2 \in \distrib(\Theta \times \Proc)$, we let $\pi_1 \cmerge \pi_2$ denote the probability distribution defined, for all $\theta_i \in \Theta$ and $\proc_i \in \Proc$, $i = 1,2$, by:
\[
(\pi_1 \cmerge \pi_2) (\theta_1 \theta_2,\proc_1 \cmerge \proc_2) =
\pi_1(\theta_1,\proc_1) \cdot \pi_2(\theta_2,\proc_2)
\]
where the effect $\theta_1\theta_2$ is given by the concatenation of the effects $\theta_1$ and $\theta_2$.
Notice that the concatenation is well-defined and it does not introduce contrasting effects on variables since we are assuming that $\var{\proc_1} \cap \var{\proc_2} = \emptyset$.
Finally, rule (PR6) states that $A$ behaves like $\proc$ whenever $A\stackrel{\mathit{def}}{=}\proc$.

We can observe that $\pstep$ returns a discrete distribution over $\Theta \times \Proc$.

\begin{prop}
[Properties of process semantics]%
\label{prop:processes_are_generative}
Let $\proc \in \Proc$ and $\ds \in \D$.
Then $\pstep(\proc,\ds)$ is a distribution with finite support, namely the set $\support(\pstep(\proc,\ds))=\{ (\theta,\proc') \mid (\pstep(\proc,\ds))(\theta,\proc') >0\}$ is finite and $\sum_{(\theta,\proc')\in\support(\pstep(\proc,\ds))} (\pstep(\proc,\ds))(\theta,\proc')=1$.
\end{prop}

%------------------------------------------------------------------------
\begin{proof}
The proof follows by a simple induction on the construction of $\pstep(\proc,\ds)$.
\end{proof}
%--------------------------------------------------------------------------

\begin{exa}%
\label{ex:tanks_III}
We proceed to define the program $\mathsf{P_{Tanks}}$ responsible for controlling the two pumps in order to achieve (and maintain) the desired level of water in the three tanks.
Intuitively, $\mathsf{P_{Tanks}}$ will consist of a process $\mathsf{P_{in}}$ controlling the flow rate through the pump attached to the first tank, and of a process $\mathsf{P_{out}}$ controlling the flow rate through the outlet pump.
The behaviour of $\mathsf{P_{in}}$ and $\mathsf{P_{out}}$ can be obtained as a straightforward translation of, respectively, Equation (\ref{eq:q1}) and Equation (\ref{eq:q3}) into our simple process calculus.
Formally:
\[
\begin{array}{lll}
\mathsf{P_{in}}
&
\stackrel{\mathit{def}}{=}
&
\mathrm{if}\; [l_1 > l_{goal} + \Delta_l]\;
(\max\{0, q_1 - q_{step}\} \rightarrow q_1).\mathsf{P_{in}}
\\[.1cm]
& &
\mathrm{else} \quad \mathrm{if}\; [l_1 < l_{goal} - \Delta_l]\;
(\min\{q_{max}, q_1 + q_{step}\} \rightarrow q_1).\mathsf{P_{in}}
\\[.1cm]
& &
\phantom{\mathrm{else}} \quad \mathrm{else}\; \surd.\mathsf{P_{in}}
\\[.3cm]
\mathsf{P_{out}}
&
\stackrel{\mathit{def}}{=}
&
\mathrm{if}\; [l_3 > l_{goal} + \Delta_l]\;
(\min\{q_{max}, q_3 + q_{step}\} \rightarrow q_3).\mathsf{P_{out}}
\\[.1cm]
& &
\mathrm{else} \quad \mathrm{if}\; [l_3 < l_{goal} - \Delta_l]\;
(\max\{0, q_3 - q_{step}\} \rightarrow q_3).\mathsf{P_{out}}
\\[.1cm]
& &
\phantom{\mathrm{else}} \quad \mathrm{else}\; \surd.\mathsf{P_{out}}
\end{array}
\]
However, for $\mathsf{P_{Tanks}}$ to work properly, we need to allow $\mathsf{P_{in}}$ and $\mathsf{P_{out}}$ to execute simultaneously.
Since the effects of read/update operators in $\mathsf{P_{in}}$ and $\mathsf{P_{out}}$ are applied to distinct variables (respectively, $q_1$ and $q_3$), we can obtain the desired execution by defining $\mathsf{P_{Tanks}}$ as the synchronous parallel composition of $\mathsf{P_{in}}$ and $\mathsf{P_{out}}$:
\[
\mathsf{P_{Tanks}}
\stackrel{\mathit{def}}{=}
\mathsf{P_{in}} \cmerge \mathsf{P_{out}}
\]
\hfill\scalebox{0.85}{$\LHD$}
\end{exa}

%=======================================
% environment
%========================================

\subsection{Modelling the environment}%
\label{sec:E}

To model the action of the environment on data we use a mapping $\E$, called \emph{environment evolution}, taking a \datastate{} to a distribution over \datastates{}.

\begin{defi}
[Environment evolution]%
\label{def:environment}
An \emph{environment evolution} is a function $\E \colon \D \to \distrib(\D,\borel_\D)$ s.t.\ for each $\mesD \in \borel_\D$ the mapping $\ds \mapsto \E(\ds)(\mesD)$ is $\borel_\D$-measurable.
\end{defi}

We notice that, due to the interaction with the program, in our model the probability induced by $\E$ at the next time step \emph{depends only} on the current state of the \dataspace.
It is then natural to assume that the behaviour of the environment is modelled as a discrete time Markov process, and to define function $\E$ we only need the initial distribution on $\D$ (which in our setting will be the Dirac's distribution on a \datastate{} in $\D$) and the random vector $\mathbf{X} = [X_1,\dots,X_n]$ in which each random variable $X_i$ models the changes induced by the environment on the value of variable $x_i$.

Since the environment is an ensemble of physical phenomena, it is represented by a system of equations of the form $x' = f(x)$, for $x \in \Var$, and thus with the obvious syntax.
Hence, for simplicity, we do not explicitly present a syntactic definition of the environment, but we model it directly as a function over \dataspaces.
The following example shows how the environment evolution is obtained from the system of equations.

\begin{exa}%
\label{ex:tanks_IV}
To conclude the encoding of the three-tanks experiment from Example~\ref{ex:tanks_I} into our model, we need to define a suitable environment evolution.
This can be derived from Equations~\eqref{eq:tankstynamics},~\eqref{eq:scenario1} (or~\eqref{eq:scenario2}), and~\eqref{eq:pipes} in the obvious way.
For sake of completeness, below we present
functions $\E_1$ and $\E_2$ modelling the evolution of the environment, respectively, in scenario 1 (with Equation (\ref{eq:scenario1})) and scenario 2 (with Equation (\ref{eq:scenario2})).
We remark that:
\begin{itemize}
\item The only stochastic variable in our setting is $q_2$, giving thus that the distribution induced by $\E_i$ is determined by the distribution of $q_2$, for $i = 1,2$.
\item The values of $q_1$ and $q_3$ are modified by $\mathsf{P_{in}}$ and $\mathsf{P_{out}}$, respectively.
In particular, according to Definition~\ref{def:cstep}, this means that the environment does not affect the values of $q_1$ and $q_3$ and, moreover, when we evaluate $\E_i(\ds)$, the effects of the two process over $q_1$ and $q_3$ have already been taken into account in $\ds$, for $i = 1,2$.
\end{itemize}
Define $f_{3T} \colon [0,q_{max}] \times \D_{3T} \to \D_{3T}$ as the function such that, for any $x \in [0,q_{max}]$ and \datastate{} $\ds$, gives us the \datastate{} $f_{3T}(x,\ds) = \ds'$ obtained from the solution of the following system of equations:
\[
\begin{cases}
q_{12} =
\begin{cases}
a_{12} \cdot a \cdot \sqrt{2g\big(\ds(l_1) - \ds(l_2)\big)} & \text{ if } \ds(l_1)\ge \ds(l_2) \\[.1cm]
-a_{12} \cdot a \cdot \sqrt{2g\big(\ds(l_2) - \ds(l_1)\big)} & \text{ otherwise;}
\end{cases}
\\
q_{23} =
\begin{cases}
a_{23} \cdot a \cdot \sqrt{2g\big(\ds(l_2) - \ds(l_3)\big)} & \text{ if } \ds(l_2)\ge \ds(l_3) \\[.1cm]
-a_{23} \cdot a \cdot \sqrt{2g\big(\ds(l_3) - \ds(l_2)\big)} & \text{ otherwise;}
\end{cases}
\\
\ds'(l_1) = \ds(l_1) + \Delta\tau \cdot (\ds(q_1) - q_{12})
\\
\ds'(l_2) = \ds(l_2) + \Delta\tau \cdot (q_{12} - q_{23})
\\
\ds'(l_3) = \ds(l_3) + \Delta\tau \cdot (x + q_{23} - \ds(q_3))
\\
\ds'(q_1) = \ds(q_1)
\\
\ds'(q_2) = x
\\
\ds'(q_3) = \ds(q_3)
\end{cases}
\]
where $a,a_{12},a_{23},g,\Delta\tau$ are constants.
Then, consider the random variables $X \sim N(q_{med},\Delta_q)$ and $Y \sim N(0,1)$.
From them, we define the random variables $Q_2^1(\ds) = f_{3T}(X, \ds)$ and $Q_2^2(\ds) = f_{3T}(\max\{q_{max},\ds(q_2)+Y\}, \ds)$.
Then, for $i =1,2$ and for each \datastate{} $\ds$, the distribution $\E_i(\ds)$ is obtained as the distribution of $Q_2^i$.
\hfill\scalebox{0.85}{$\LHD$}
\end{exa}

%=========================================
% configurations
%=======================================

\subsection{Modelling system's behaviour}%
\label{sec:config}

We use the notion of \emph{configuration} to model the state of the system at each time step.

\begin{defi}[Configuration]
A \emph{configuration} is a triple $c = \config{\proc}{\ds}{\E}$, where $\proc$ is a process, $\ds$ is a \datastate{} and $\E$ is an environment evolution.
We denote by $\C_{\Proc,\D,\E}$ the set of configurations defined over $\Proc,\D$ and $\E$.
\end{defi}
When no confusion shall arise, we shall write $\C$ in place of $\C_{\Proc,\D,\E}$.

Let $(\Proc,\Sigma_{\Proc})$ be the measurable space of processes, where $\Sigma_{\Proc}$ is the power set of $\Proc$, and $(\D,\borel_\D)$ be the measurable space of \datastates.
As $\E$ is fixed, we can identify $\C$ with $\Proc \times \D$ and equip it with the product $\sigma$-algebra $\Sigma_\C = \Sigma_{\Proc} \otimes \borel_\D$:
$\Sigma_{\C}$ is generated by the sets $\{\config{\mesP}{\mesD}{\E} \mid \mesP \in \Sigma_{\Proc}, \mesD \in \borel_\D\}$, where $\config{\mesP}{\mesD}{\E}=\{\config{\proc}{\ds}{\E} \mid \proc \in \mesP, \ds\in \mesD\}$.

\begin{nota}%
\label{notazione_distrib}
For $\mu_{\Proc} \in \distrib(\Proc,\Sigma_\Proc)$ and $\mu_{\D} \in \distrib(\D,\borel_\D)$ we let $\mu=\config{\mu_{\Proc}}{\mu_{\D}}{\E}$ denote the product distribution on $(\C,\Sigma_{\C})$, i.e., $\mu(\config{\mesP}{\mesD}{\E})=\mu_{\Proc}(\mesP)\cdot \mu_{\D}(\mesD)$ for all  $\mesP \in \Sigma_\Proc$ and $\mesD \in \borel_\D$.
If $\mu_{\Proc} = \delta_{\proc}$ for some $\proc\in\Proc$, we shall denote $\config{\delta_{\proc}}{\mu_{\D}}{\E}$ simply by $\config{\proc}{\mu_{\D}}{\E}$.
\end{nota}

Our aim is to express the behaviour of a system in terms of the changes on data.
We start with the \emph{one-step} behaviour of a configuration, in which we combine the effects on the \datastate{} induced by the activity of the process (given by $\pstep$) and the subsequent action by the environment.
Formally, we define a function $\cstep$ that, given a configuration $c$, yields a distribution on $(\C,\Sigma_{\C})$ (Definition~\ref{def:cstep} below).
Then, we use $\cstep$ to define the \emph{multi-step} behaviour of configuration $c$ as a sequence $\ES_{c,0}^{\C}, \ES_{c,1}^{\C},\ldots$ of distributions on $(\C,\Sigma_{\C})$.
To this end, we show that $\cstep$ is a Markov kernel (Proposition~\ref{prop:cstep} below).
Finally, to abstract from processes and focus only on data, from the sequence $\ES_{c,0}^{\C},\ES_{c,1}^{\C},\ldots$, we obtain a sequence of distributions $\ES_{c,0}^{\D}, \ES_{c,1}^{\D},\ldots$ on $(\D,\borel_\D)$ called the \emph{\traccione{}} of the system (Definition~\ref{def:traccione} below).

\begin{defi}
[One-step semantics of configurations]%
\label{def:cstep}
Function $\cstep \colon \C \to \distrib(\C,\Sigma_\C)$ is defined for all configurations
$\config{\proc}{\ds}{\E} \in \C$ by
\begin{equation}%
\label{ex:function_cstep}
\cstep(\config{\proc}{\ds}{\E}) =
\sum_{(\theta,\proc')\in \support(\pstep(\proc,\ds))} \pstep(\proc,\ds)(\theta,\proc')\cdot \config{\proc'}{\E(\theta(\ds))}{\E}.
\end{equation}
\end{defi}

The next result follows by $\E(\theta(\ds)) \in \distrib(\D,\borel_\D)$ and Proposition~\ref{prop:processes_are_generative}.

\begin{prop}
[Properties of one-step configuration semantics]%
\label{prop:configurations_are_generative}
Assume any configuration $\config{\proc}{\ds}{\E} \in \C$.
Then $\cstep(\config{\proc}{\ds}{\E})$ is a distribution on $(\C,\Sigma_{\C})$.
\end{prop}

%----------------------------------------------
\begin{proof}
Being $\ds$ a \datastate{} in $\D$ and $\theta$ an effect, $\theta(\ds)$ is a \datastate.
Then, $\E(\theta(\ds))$ is a distribution on $(\D,\borel_\D)$, thus implying that
$\config{\proc'}{\E(\theta(\ds))}{\E}$ is a distribution on $(\C,\Sigma_{\C})$ (see Notation~\ref{notazione_distrib}).
Finally, by Proposition~\ref{prop:processes_are_generative} we get that $\cstep(\config{\proc}{\ds}{\E})$ is a convex combination of distributions on $(\C,\Sigma_{\C})$ whose weights sum up to 1, which gives the thesis.
\end{proof}
%----------------------------------------------

Since $\cstep(c) \in \distrib(\C,\Sigma_{\C})$ for each $c \in \C$, we can rewrite $\cstep \colon \C \times \Sigma_{\C} \to [0,1]$, so that for each configuration $c \in \C$ and measurable set $\mesC \in \Sigma_{\C}$, $\cstep(c)(\mesC)$ denotes the probability of reaching in one step a configuration in $\mesC$ starting from $c$.
We can prove that $\cstep$ is the Markov kernel of the Markov process modelling our system.
This follows by Proposition~\ref{prop:configurations_are_generative} and by proving that for each $\mesC \in \Sigma_\C$, the mapping $c \mapsto \cstep(c)(\mesC)$ is $\Sigma_\C$-measurable.

\begin{prop}%
\label{prop:cstep}
The function $\cstep$ is a Markov kernel.
\end{prop}

%---------------------------------------------------------
\begin{proof}
We need to show that $\cstep$ satisfies the two properties of Markov kernels, namely
\begin{enumerate}
\item\label{kernel_1}
For each configuration $c \in \C$, the mapping $\mesC \mapsto \cstep(c)(\mesC)$ is a distribution on $(\C, \Sigma_\C)$.
\item\label{kernel_2}
For each measurable set $\mesC \in \Sigma_\C$, the mapping $c \mapsto \cstep(c)(\mesC)$ is $\Sigma_\C$-measurable.
\end{enumerate}
Item~\ref{kernel_1} follows directly by Proposition~\ref{prop:configurations_are_generative}.

Let us focus on item~\ref{kernel_2}.
By Definition~\ref{def:cstep}, for each $\config{\proc}{\ds}{\E} \in \C$ we have that
\[
\cstep(\config{\proc}{\ds}{\E})(\mesC) =
\sum_{(\theta,\proc')\in \support(\pstep(\proc,\ds))}
\pstep(\proc,\ds)(\theta,\proc') \cdot
\config{\proc'}{\E(\theta(\ds))}{\E}(\mesC)
\enspace .
\]
As, by definition, each $\theta \in \Theta$ and $\E$ are $\borel_\D$-measurable functions, we can infer that also their composition $\E(\theta(\cdot))$ is a $\borel_\D$-measurable function.
Since, moreover, $\Sigma_\C$ is the (smallest) $\sigma$-algebra generated by $\Sigma_{\Proc} \times \borel_\D$ and every subset of $\Proc$ is a measurable set in $\Sigma_{\Proc}$, we can also infer that $\config{(\cdot)}{\E(\theta(\cdot))}{\E}$ is a $\Sigma_\C$-measurable function.
Finally, we recall that by Proposition~\ref{prop:processes_are_generative} we have that $\support(\pstep(\proc,\ds))$ is finite.
Therefore, we can conclude that $\cstep$ is a $\Sigma_\C$-measurable function as linear combination of a finite collection of $\Sigma_\C$-measurable functions (see, e.g.,~\cite[Chapter 3.5, Proposition 19]{Ro88}).
\end{proof}
%--------------------------------------------------------

Hence, the multi-step behaviour of configuration $c$ can be defined as a time homogeneous Markov process having $\cstep$ as Markov kernel and $\delta_c$ as initial distribution.

\begin{defi}
[Multi-step semantics of configurations]%
\label{def:multi-step_cstep}
Let $c\in \C$ be a configuration.
The \emph{multi-step behaviour} of $c$ is the sequence of distributions $\ES_{c,0}^{\C},\ES_{c,1}^{\C}, \ldots$ on $(\C,\Sigma_{\C})$ defined inductively as follows:
\begin{align*}
\ES_{c,0}^{\C}(\mesC) = & \; \delta_{c}(\mesC) \text{, for all } \mesC \in \Sigma_{\C}\\
\ES_{c,i+1}^{\C}(\mesC) = & \int_{\C} \cstep(b)(\mesC) \;\dd(\ES_{c,i}^{\C}(b))  \text{, for all } \mesC \in \Sigma_{\C}.
\end{align*}
\end{defi}

We can prove that $\ES_{c,0}^{\C},\ES_{c,1}^{\C}, \ldots$ are well defined, namely they are distributions on $(\C,\Sigma_{\C})$.

\begin{prop}
[Properties of configuration multi-step semantics]%
\label{prop:configurations_are_super_generative}
For any $c \in \C$, all $\ES_{c,0}^{\C},\ES_{c,1}^{\C}, \ldots$ are distributions on $(\C,\Sigma_{\C})$.
\end{prop}

%-----------------------------------------------------------------
\begin{proof}
The proof follows by induction.
The base case for $\ES_{c,0}^{\C}$ is immediate.
The inductive step follows by Proposition~\ref{prop:cstep}.
\end{proof}
%---------------------------------------------------------------------

As the program-environment interplay can be observed only in the changes they induce on the \datastates, we define the \emph{\traccione{}} of a configuration as the sequence of distributions over \datastates{} that are reached by it, step-by-step.

\begin{defi}
[\Traccione]%
\label{def:traccione}
The \emph{\traccione{}} of a configuration $c = \config{\proc}{\ds}{\E}$ is a sequence $\ES^\D_c \in \distrib(\D,\borel_\D)^{\omega}$ of distributions on $(\D,\borel_{\D})$ s.t.\ $\ES^\D_c = \ES^\D_{c,0} \dots \ES^\D_{c,n} \dots$ if and only if for all $i \ge 0$ and for all $\mesD \in \borel_\D$,
\[
\ES^\D_{c,i}(\mesD) =  \ES_{c,i}^{\C}(\config{\Proc}{\mesD}{\E}).
\]
\end{defi}

%===============================================
% sec - metriche
%===============================================

\section{Towards a metric for systems}%
\label{sec:metriche}

As outlined in the Introduction, we aim at defining a \emph{distance} over the systems described in the previous section, called the \emph{\spell metric}, allowing us to do the following:
\begin{enumerate}
\item Measure how well a program is fulfilling its tasks.
\item Establish whether one program behaves better than another one in an environment.
\item Compare the interactions of a program with different environments.
\end{enumerate}
As we will see, these three objectives can be naturally obtained thanks to the possibility of modelling the program in isolation from the environment typical of our model, and to our purely data-driven system semantics.
Intuitively, since the behaviour of a system is entirely described by its \traccione, the \spell metric $\m$ will indeed be defined as a distance on the \tracciones{} of systems.
However, in order to obtain the proper technical definition of $\m$, some considerations are due.

Firstly, we notice that in most applications, not only the program-environment interplay, but also the tasks of the program can be expressed in a purely data-driven fashion.
For instance, if we are monitoring a particular feature of the system then we can identify a (time-dependent) \emph{specific value of interest} of a state variable, like an upper bound on the energy consumption at a given time step.
Similarly, if we aim at controlling several interacting features or there is a range of parameters that can be considered acceptable, like in the case of our running example, then we can specify a (time-dependent) \emph{set of values of interest}.
At any time step, any difference between them and the data actually obtained can be interpreted as a flaw in systems behaviour.
We use a \emph{penalty function} $\rho$ to quantify these differences.
From the penalty function we can obtain a \emph{distance on \datastates}, namely a $1$-bounded \emph{hemimetric} $m^\D$ expressing how much a \datastate{} $\ds_2$ is worse than a \datastate{} $\ds_1$ according to parameters of interests.
The idea is that if $\ds_1$ is obtained by the interaction of $\proc_1$ with $\E$ and $\ds_2$ is obtained from the one of $\proc_2$ with $\E$, $\ds_2$ being worse than $\ds_1$ means that $\proc_2$ is performing worse than $\proc_1$ in this case.
Similarly, if $\ds_1$ and $\ds_2$ are obtained from the interaction of $\proc$ with $\E_1$ and $\E_2$, respectively, we can express whether in this particular configuration and time step $\proc$ behaves worse in one environment with respect to the other.

Secondly, we recall that the \traccione{} of a system consists in a sequence of \emph{distributions} over \datastates.
Hence, we use the \emph{Wasserstein metric} to lift $m^\D$ to a distance $\Wasserstein(m^\D)$ over distributions over \datastates.
Informally, the Wasserstein metric gives the expected value of the ground distance between the elements in the support of the distributions.
Thus, in our case, the lifting expresses how much worse a configuration is expected to behave with respect to another one at a given time.

Finally, we need to lift $\Wasserstein(m^\D)$ to a distance on the entire \tracciones{} of systems.
For our purposes, a reasonable choice is to take the maximum over time of the pointwise (with respect to time) Wasserstein distances (see Remark~\ref{rmk:why_max} below for further details on this choice).
Actually, to favour computational tractability, our metric will \emph{not} be evaluated on \emph{all} the distributions in the \tracciones, but \emph{only} on those that are reached at certain time steps, called the \emph{observation times} ($\OT$).

We devote the remaining of this section to a formalisation of the intuitions above.

%===========================================

\subsection{A metric on \datastates}%
\label{sec:metric_on_data}

We start by proposing a metric on \datastates, seen as \emph{static components} in isolation from processes and environment.
To this end, we introduce a \emph{penalty function} $\rho \colon \D \to [0,1]$, a continuous function that assigns to each \datastate{} $\ds$ a penalty in $[0,1]$ expressing how far the values of the parameters of interest in $\ds$ are from their desired ones (hence $\rho(\ds) = 0$ if $\ds$ respects all the parameters).
Since sometimes the parameters can be time-dependent, we also introduce a time-dependent version of $\rho$: at any time step $\tau$, the $\tau$-penalty function $\rho_\tau$ compares the \datastates{} with respect to the values of the parameters expected at time $\tau$.
When the value of the penalty function is independent from the time step, we simply omit the subscript $\tau$.

\begin{exa}%
\label{ex:tanks_V}
A requirement on the behaviour of the three-tanks system from Example~\ref{ex:tanks_I} can be that the level of water $l_{i}$ should be at the level $l_{goal}$, for $i = 1,2,3$.
This can be easily rendered in terms of the penalty functions $\rho^{l_i}$, for $i=1,2,3$, defined as the normalised distance between the current level of water $\ds(l_i)$ and $l_{goal}$, namely:
\begin{equation}%
\label{eq:single_penalty_function}
\rho^{l_i}(\ds) =
\frac{|\ds(l_i) - l_{goal}|}{\max\{l_{max}-l_{goal},l_{goal}-l_{min}\}}
\enspace .
\end{equation}
Clearly, we can also consider as a requirement that \emph{all} the $l_i$ are at level $l_{goal}$.
The penalty function thus becomes
\[
\rho^{\max}(\ds)=\max_{i =1,2,3}\rho^{l_i}(\ds)
\enspace .
\]
The ones proposed above are just a few simple examples of requirements, and related penalty functions, for the three-tanks experiment.
Yet, in their simplicity, they will allow us to showcase a number of different applications of our framework while remaining in a setting where the reader can easily foresee the expected results.
\hfill\scalebox{0.85}{$\LHD$}
\end{exa}

A formal definition of the penalty function is beyond the purpose of this paper, also due to its context-dependent nature.
Besides, notice that we can assume that $\rho$ already includes some tolerances with respect to the exact values of the parameters in its evaluation, and thus we do not consider them.
The (\emph{timed}) \emph{metric on \datastates{}} is then defined as the asymmetric difference between the penalties assigned to them by the penalty function.

\begin{defi}
[Metric on \datastates{}]%
\label{def:metric_DS}
For any time step $\tau$, let $\rho_{\tau} \colon \D \rightarrow [0,1]$ be the $\tau$-penalty function on $\D$.
The $\tau$-\emph{metric on \datastates{}} in $\D$, $m_{\rho,\tau}^{\D} \colon \D \times \D \to [0,1]$, is defined, for all $\ds_1,\ds_2 \in \D$, by
\[
m_{\rho,\tau}^{\D}(\ds_1,\ds_2) = \max\{\rho_{\tau}(\ds_2)-\rho_{\tau}(\ds_1), 0 \}.
\]
\end{defi}

Notice that $m_{\rho,\tau}^{\D}(\ds_1,\ds_2) > 0$ if{f} $\rho_\tau(\ds_2) > \rho_\tau(\ds_1)$, i.e., the penalty assigned to $\ds_2$ is higher than that assigned to $\ds_1$.
For this reason, we say that $m_{\rho,\tau}^{\D}(\ds_1,\ds_2)$ expresses \emph{how worse} $\ds_2$ is than $\ds_1$ with respect to the objectives of the system.
It is not hard to see that for all $\ds_1,\ds_2,\ds_3 \in \D$ we have $m_{\rho,\tau}^{\D}(\ds_1,\ds_2)\le 1$, $m_{\rho,\tau}^{\D}(\ds_1,\ds_1)=0$, and $m_{\rho,\tau}^{\D}(\ds_1,\ds_2) \le m_{\rho,\tau}^{\D}(\ds_1,\ds_3) + m_{\rho,\tau}^{\D}(\ds_3,\ds_2)$, thus ensuring that $m_{\rho,\tau}^{\D}$ is a 1-bounded hemimetric.

\begin{prop}%
\label{prop:m_metrica}
Function $m_{\rho,\tau}^{\D}$ is a 1-bounded hemimetric on $\D$.
\end{prop}

When no confusion shall arise, we shall drop the $\rho,\tau$ subscripts.
We remark that penalty functions allow us to define the distance between two \datastates, which are elements of $\real^n$, in terms of a distance on $\real$.
As we discuss in Section~\ref{sec:computing}, this feature significantly lowers the complexity of the evaluation the \spell metric.

%=================================================

\subsection{Lifting \texorpdfstring{$m^\D$}{mD} to distributions}%
\label{sec:Wass}

The second step to obtain the \spell metric consists in lifting $m^\D$ to a metric on distributions on \datastates.
In the literature, we can find a wealth of notions of distances on distributions (see~\cite{RKSF13} for a survey).
For our purposes, the most suitable one is the \emph{Wasserstein metric}~\cite{W69}.

According to Definition~\ref{def:Wasserstein} (given in Section~\ref{sec:background}), for any two distributions $\mu,\nu$ on $(\D,\borel_\D)$, the Wasserstein lifting of $m^{\D}$ to a distance between $\mu$ and $\nu$ is defined by
\[
\Wasserstein(m^{\D})(\mu,\nu) = \inf_{\w \in \W(\mu,\nu)} \int_{\D \times \D} m^{\D}(\ds,\ds') \;\dd\w(\ds,\ds')
\]
where $\W(\mu,\nu)$ is the set of the \emph{couplings of $\mu$ and $\nu$}, namely the set of joint distributions $\w$ over the product space $(\D \times \D, \borel_\D \otimes \borel_\D)$ having $\mu$ and $\nu$ as left and right marginal.

As outlined in Section~\ref{sec:background},~\cite[Proposition 4, Proposition 6]{FR18} guarantee that $\Wasserstein(m^\D)$ is a well-defined hemimetric on $\distrib(\D,\borel_\D)$.

\begin{prop}%
\label{prop:W_metrica}
Function $\Wasserstein(m^\D_{\rho,\tau})$ is a $1$-bounded hemimetric on $\distrib(\D,\borel_\D)$.
\end{prop}

%==========================================================

\subsection{The \spell metric}%
\label{sec:spell_metric}

We now need to lift $\Wasserstein(m^{\D})$ to a distance on \tracciones{}.
To this end, we observe that the \traccione{} of a configuration includes the distributions over \datastates{} induced after \emph{each} computation step.
However, it could be the case that the changes on data induced by the environment can be appreciated only along wider time intervals.
To deal with this kind of situations, we introduce the notion of \emph{observation times}, namely a \emph{discrete} set $\OT$ of time steps at which the modifications induced by the program-environment interplay give us useful information on the evolution of the system (with respect to the considered objectives).
Hence, a comparison of the \tracciones{} based on the differences in the distributions reached at the times in $\OT$ can be considered meaningful.
Moreover, considering only the differences at the observation times will favour the computational tractability of the \spell metric.

To compare \tracciones, we propose a sort of \emph{weighted infinity norm} of the tuple of the Wasserstein distances between the distributions in them.
As weight we consider a non-increasing function $\lambda \colon \OT \to \lopen{0,1}$ allowing us to express how much the distance at time $\tau$, namely $\Wasserstein(m^\D_{\rho,\tau})(\ES^\D_{c_1,\tau},\ES^\D_{c_2,\tau})$, affects the overall distance between configurations $c_1$ and $c_2$.
The idea is that, in certain application contexts, the differences in the behaviour of systems that are detected after wide time intervals, can be less relevant than the differences in the initial behaviour.
Hence, we can use the weight $\lambda(\tau)$ to \emph{mitigate} the distance of events at time $\tau$.
Following the terminology introduced in the context of behavioural metrics~\cite{AHM03,DGJP04}, we refer to $\lambda$ as to the \emph{discount function}, and to $\lambda(\tau)$ as to the \emph{discount factor at time} $\tau$.
Clearly, the constant function $\lambda(\tau) = 1$, for all $\tau \in \OT$, means that no discount is applied.

\begin{defi}
[Evolution metric]%
\label{def:spell_metric}
Assume a set $\OT$ of observation times and a discount function $\lambda$.
For each $\tau \in \OT$, let $\rho_{\tau}$ be a $\tau$-penalty function and let $m^\D_{\rho,\tau}$ be the $\tau$-metric on \datastates{} defined on it.
Then, the $\lambda$-\emph{\spell metric} over $\rho$ and $\OT$ is the mapping $\m^{\lambda}_{\rho,\OT} \colon \C \times \C \to [0,1]$ defined, for all configurations $c_1,c_2 \in \C$, by
\[
\m^{\lambda}_{\rho,\OT}(c_1,c_2) = \sup_{\tau \in \OT}\, \lambda(\tau) \cdot \Wasserstein(m^\D_{\rho,\tau})\Big( \ES^\D_{c_1,\tau}, \ES^\D_{c_2,\tau} \Big).
\]
\end{defi}

Since $\Wasserstein(m^{\D}_{\rho,\tau})$ is a 1-bounded hemimetric (Proposition~\ref{prop:W_metrica}) we can easily derive the same property for $\m^{\lambda}_{\rho,\OT}$.

\begin{prop}
Function $\m^{\lambda}_{\rho,\OT}$ is a 1-bounded hemimetric on $\C$.
\end{prop}

Notice that if $\lambda$ is a \emph{strictly} non-increasing function, then to obtain upper bounds on the \spell metric only a \emph{finite} number of observations is needed.

\begin{rem}%
\label{rmk:why_max}
Usually, due to the presence of uncertainties, the behaviour of a system can be considered acceptable even if it differs from its intended one \emph{up-to a certain tolerance}.
Similarly, the properties of adaptability, reliability, and robustness that we aim to study will check whether a program is able to perform well in a perturbed environment \emph{up-to a given tolerance}.
In this setting, the choice of defining the \spell metric as the pointwise maximal distance in time between the \tracciones{} of systems is natural and reasonable: if in the worst case (the maximal distance) the program keeps the parameters of interest within the given tolerance, then its entire behaviour can be considered acceptable.
However, with this approach we have that a program is only as good as its worst performance, and one could argue that there are application contexts in which our \spell metric would be less meaningful.
For these reasons, we remark that we could have given a \emph{parametric} version of Definition~\ref{def:spell_metric} and defining the \spell metric in terms of a generic \emph{aggregation function} $f$ over the tuple of Wasserstein distances.
Then, one could choose the best instance for $f$ according to the chosen application context.
The use of a parametric definition would have not affected the technical development of our paper.
However, in order to keep the notation and presentation as simple as possible, we opted to define $\m^{\lambda}_{\rho,\OT}$ directly in the weighted infinity norm form.
A similar reasoning applies to the definition of the penalty function that we gave in Example~\ref{ex:tanks_V}, and to the definition of the metric over \datastates{} (Definition~\ref{def:metric_DS}).
\end{rem}

%==============================================
% sec - computing
%===============================================

\section{Estimating the \spell metric}%
\label{sec:computing}

In this section we show how the \spell metric can be estimated via statistical techniques.
Firstly, in Section~\ref{sec:empirical_tracciones} we show how we can estimate the \traccione{} of a given configuration $c$.
Then, in Section~\ref{sec:computing_distance} we evaluate the distance between two configurations $c_1$ and $c_2$ on their estimated \tracciones.

%===============================================

\subsection{Computing empirical \tracciones}%
\label{sec:empirical_tracciones}

Given a configuration $\config{\proc}{\ds}{\E}$ and an integer $k$ we can use the function $\Call{Simulate}{}$, defined in Figure~\ref{alg:simstep}, to sample a sequence of configurations of the form
\[
\config{\proc_0}{\ds_0}{\E},
\config{\proc_1}{\ds_1}{\E},
\ldots ,
\config{\proc_k}{\ds_k}{\E}.
\]
This sequence represents $k$-steps of a computation starting from $\config{\proc}{\ds}{\E}=\config{\proc_0}{\ds_0}{\E}$.
Each step of the sequence is computed by using function $\Call{SimStep}{}$, also defined in Figure~\ref{alg:simstep}.
There we let $\Call{rand}{}$ be a function that allows us to get a uniform random number in $\lopen{0,1}$ while $\Call{sample}{\E,\ds}$ is used to sample a \datastate{} in the distribution $\E(\ds)$.
We assume that for any measurable set of \datastates{} $\mesD \in \borel_\D$ the probability of $\Call{sample}{\E,\ds}$ to be in $\mesD$ is equal to the probability of $\mesD$ induced by $\E(\ds)$, namely:
\begin{equation}%
\label{eq:fsample_property}
\PP(\Call{sample}{\E,\ds}\in \mesD)=\E(\ds)(\mesD)
\enspace .
\end{equation}

\begin{exa}
If we consider the environment evolution from Examples~\ref{ex:tanks_I} and~\ref{ex:tanks_IV}, function $\Call{sample}{\E,\ds}$ would sample a random value $x$ for $q_2$ from the distribution chosen for it (according to the considered scenario), and then solve the system of equations identified by $f_{3T}(x,\ds)$.
\hfill\scalebox{0.85}{$\LHD$}
\end{exa}

\begin{figure}[tbp]
\begin{minipage}[t]{0.5\textwidth}
\small
\begin{algorithmic}[1]
\Function{Simulate}{$\config{\proc}{\ds}{\E},k$}
\State $i\gets 1$
\State $c\gets \config{\proc}{\ds}{\E}$
\State $a \gets c$
\While{$i\leq k$}
\State $c \gets \Call{SimStep}{c}$
\State $a \gets a,c$
\State $i \gets i+1$
\EndWhile
\State \Return $a$
\EndFunction
\end{algorithmic}
\end{minipage}\hfill
\begin{minipage}[t]{0.5\textwidth}
\small
\begin{algorithmic}[1]
\Function{SimStep}{$\config{\proc}{\ds}{\E}$}
\State $\sum_{i=1}^n p_i (\theta_i,\proc_i)
\gets\pstep(\proc,\ds)$
\State $u\gets \Call{rand}{ }$
\State let $i$ s.t. $\sum_{j=1}^{i-1} p_j< u\leq \sum_{j=1}^{i} p_j$
\State $\ds_i'\gets \Call{sample}{\E,\theta_i(\ds_i)}$
\State \Return $\config{\proc_i}{\ds_i'}{\E}$
\EndFunction
\end{algorithmic}
\end{minipage}\hfill
\caption{Functions used to simulate behaviour of a configuration.}%
\label{alg:simstep}
\end{figure}

We can then extend this property to configurations and function $\Call{SimStep}{}$:

\begin{lem}%
\label{lem:SimStep}
For any configuration $\config{\proc}{\ds}{\E}\in \C$, and for any measurable set $\mesC \in \Sigma_{\C}$:
\[
\PP(\Call{SimStep}{\config{\proc}{\ds}{\E}}\in \mesC)=\cstep(\config{\proc}{\ds}{\E})(\mesC).
\]
\end{lem}

%---------------------------------------------------
\begin{proof}
To prove the thesis it is enough to show that the property
\[
\PP(\Call{SimStep}{\config{\proc}{\ds}{\E}}\in \mesC)=\cstep(\config{\proc}{\ds}{\E})(\mesC)
\]
holds on the generators of the $\sigma$-algebra $\Sigma_\C$.
Hence, assume that $\mesC = \config{\mesP}{\mesD}{\E}$ for some $\mesP \in \Sigma_{\Proc}$ and $\mesD \in \borel_\D$.
Then we have
\begin{align*}
\cstep(\config{\proc}{\ds}{\E})(\mesC)
= \quad
&
\sum_{(\theta,\proc')\in \support(\pstep(\proc,\ds))} \left( \pstep(\proc,\ds)(\theta,\proc')\cdot \config{\proc'}{\E(\theta(\ds))}{\E} \right) (\mesC)
\\
= \quad
&
 \sum_{i=1}^{n} \left( p_i \cdot \config{\proc_i}{\E(\theta_i(\ds))}{\E} \right) (\mesC)
\\
= \quad
&
 \sum_{i=1}^{n} p_i \cdot \mathds{1}_\mesP\{\proc_i\}
 \cdot \E(\theta_i(\ds))(\mesD)
\\
= \quad
&
\sum_{\{i | \proc_i \in \mesP\}} p_i \cdot \E(\theta_i(\ds)) (\mesD)
\\
= \quad
&
\sum_{\{i | \proc_i \in \mesP\}} p_i \cdot \PP(\Call{Sample}{\E,\theta_i(\ds)}\in \mesD)
\\
= \quad
&
\PP(\Call{SimStep}{\config{\proc}{\ds}{\E}}\in \mesC)
\end{align*}
where:
\begin{itemize}
\item The first step follows by the definition of $\cstep$ (Definition~\ref{def:cstep}).
\item The second step follows by $\pstep(\proc,\ds) = p_1 \cdot (\theta_1,\proc_1) + \cdots + p_n \cdot (\theta_n,\proc_n)$, for some $p_i$, $\theta_i$ and $\proc_i$ (see Proposition~\ref{prop:processes_are_generative}).
\item The third  and fourth steps follow by the definition of the distribution $\config{\mu_{\Proc}}{\mu_\D}{\E}$ (Notation~\ref{notazione_distrib}), in which $\mathds{1}_\mesP$ denotes the characteristic function of the set $\mesP$, i.e., $\mathds{1}_\mesP\{\proc'\} = 1$ if $\proc' \in \mesP$ and $\mathds{1}_\mesP\{\proc'\} = 0$ otherwise.
\item The fifth step follows by the assumption on the function $\Call{Sample}{\_}$ in Equation~\eqref{eq:fsample_property}.
\item The last step follows by the definition of function $\Call{SimpStep}{\_}$.
\qedhere
\end{itemize}
\end{proof}
%-------------------------------------------------

\begin{figure}[tbp]
\small
\begin{algorithmic}[1]
\Function{Estimate}{$\config{\proc}{\ds}{\E},k,N$}
\State $\forall i:(0\leq i\leq k): E_i \gets \emptyset$
\State $counter\gets 0$
\While{$counter< N$}
\State $(c_0,\ldots,c_k) \gets \Call{Simulate}{\config{\proc}{\ds}{\E},k}$
\State $\forall i: E_i \gets E_i,c_i$
\State $counter \gets counter + 1$
\EndWhile
\State \Return $E_0,\ldots,E_k$
\EndFunction
\end{algorithmic}
\caption{\label{fig:estimate} Function used to obtain $N$ samples of the \traccione{} of a configuration.}
\end{figure}

\noindent
To compute the empirical \traccione{} of a configuration $c$ the function $\Call{Estimate}{}$ in Figure~\ref{fig:estimate} can be used.

The function $\Call{Estimate}{c,k,N}$ invokes $N$ times the function $\Call{Simulate}{}$ in Figure~\ref{alg:simstep} to sample a sequence of configurations $c_0,\dots,c_k$, modelling $k$ steps of a computation from $c=c_0$.
Then, a sequence of observations $E_0,\ldots,E_k$ is computed, where each $E_i$ is the tuple $c_i^{1},\ldots,c_i^{N}$ of configurations observed at time $i$ in each of the $N$ sampled computations.

Notice that, for each $i \in \{0,\dots,k\}$, the samples $c_i^1,\dots,c_i^N$ are independent and moreover, by Lemma~\ref{lem:SimStep}, they are identically distributed.
Each $E_i$ can be used to estimate the distribution $\ES^\C_{c,i}$.
For any $i$, with $0\leq i\leq k$, we let $\hat{\ES}_{c,i}^{\C,N}$ be the distribution such that for any measurable set of configurations $\mesC \in \Sigma_{\C}$ we have
\[
\hat{\ES}_{c,i}^{\C,N}(\mesC)=\frac{|E_i \cap \mesC|}{N}.
\]
Finally, we let $\hat{\ES}_{c}^{\D,N}=\hat{\ES}_{c,0}^{\D,N}\ldots \hat{\ES}_{c,k}^{\D,N}$ be the \emph{empirical evolution sequence} such that for any measurable set of \datastates{} $\mesD\in \borel_\D$ we have
\[
\hat{\ES}_{c,i}^{\D,N}(\mesD)=\hat{\ES}_{c,i}^{\C,N}(\config{\Proc}{\mesD}{\E}).
\]
Then, by applying the weak law of large numbers to the i.i.d samples, we get that when $N$ goes to infinite both $\hat{\ES}_{c,i}^{\C,N}$ and $\hat{\ES}_{c,i}^{\D,N}$ converge weakly to $\ES^\C_{c,i}$ and $\ES_{c,i}^{\D}$ respectively:
\begin{equation}%
\label{eq:weak_convergence}
\lim_{N\rightarrow \infty}\hat{\ES}_{c,i}^{\C,N} = \ES_{c,i}^\C
\qquad\qquad
\lim_{N\rightarrow \infty}\hat{\ES}_{c,i}^{\D,N} = \ES_{c,i}^{\D}.
\end{equation}

The algorithms in Figures~\ref{alg:simstep} and~\ref{fig:estimate} have been implemented in Python to support analysis of systems.
The tool is available at \url{https://github.com/quasylab/spear}.

\begin{figure}[t]
\includegraphics[scale=0.45]{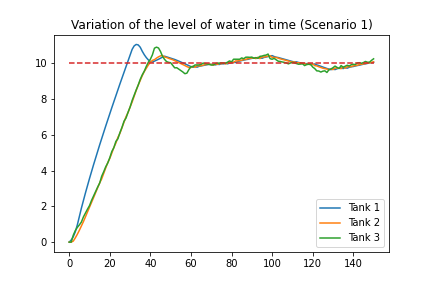}
\includegraphics[scale=0.45]{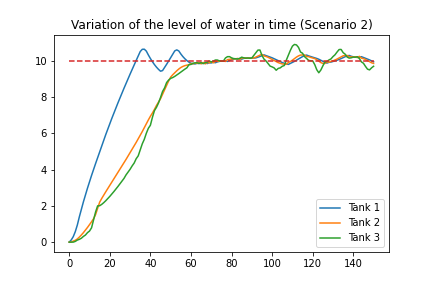}
\caption{\label{fig:simulation_tanks} Simulation results.
The \textcolor{red}{dashed red} line corresponds to $l_{goal} = 10$.}
\end{figure}

\begin{figure}[t]
\includegraphics[scale=0.32]{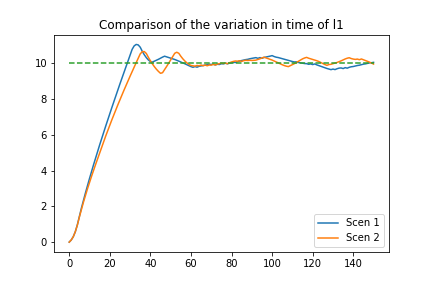}
\includegraphics[scale=0.32]{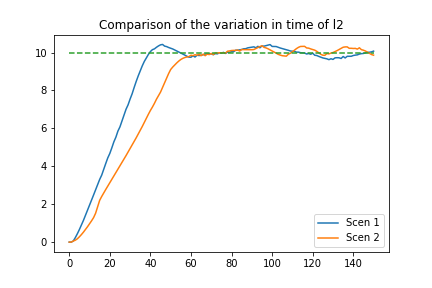}
\includegraphics[scale=0.32]{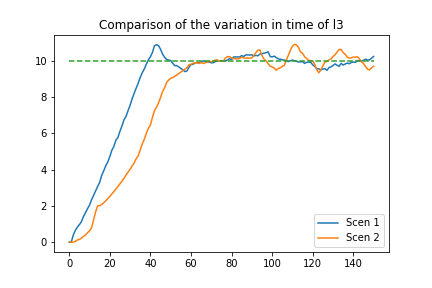}
\caption{\label{fig:simulations_tank_by_tank} Simulation of the variation of the level of water in each tank in the two scenarios. The \textcolor{ForestGreen}{dashed green} line corresponds to $l_{goal} = 10$.}
\end{figure}

\begin{figure}[t]
\begin{subfigure}{1\textwidth}
\includegraphics[scale=0.3]{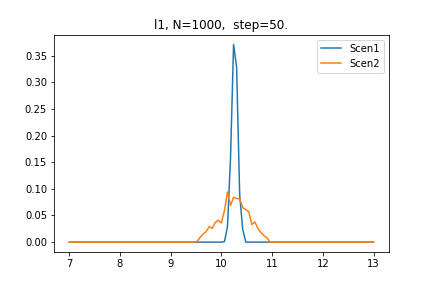}
\includegraphics[scale=0.3]{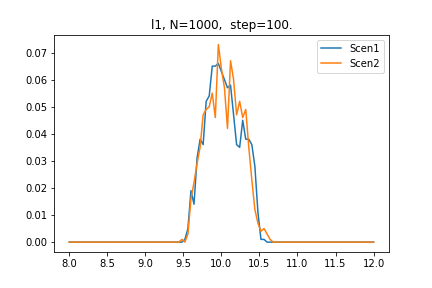}
\includegraphics[scale=0.3]{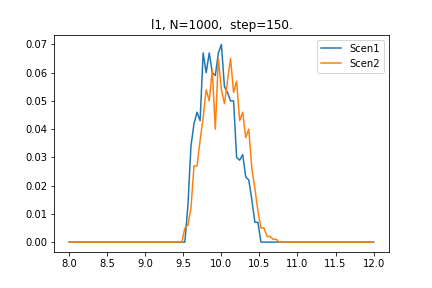}
\end{subfigure}
\hfill
\begin{subfigure}{1\textwidth}
\includegraphics[scale=0.3]{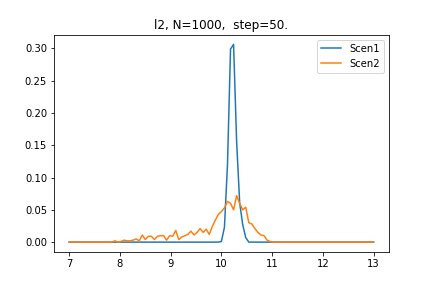}
\includegraphics[scale=0.3]{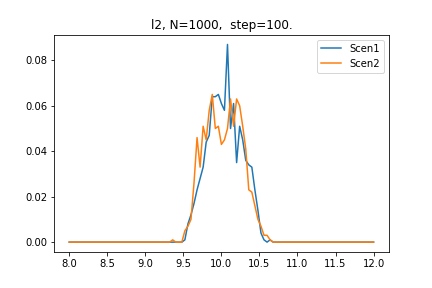}
\includegraphics[scale=0.3]{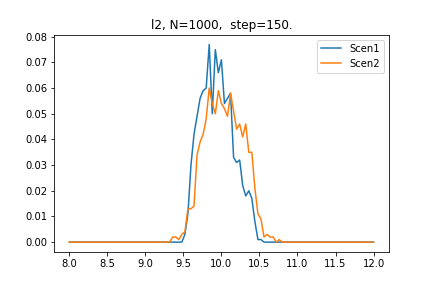}
\end{subfigure}
\hfill
\begin{subfigure}{1\textwidth}
\includegraphics[scale=0.3]{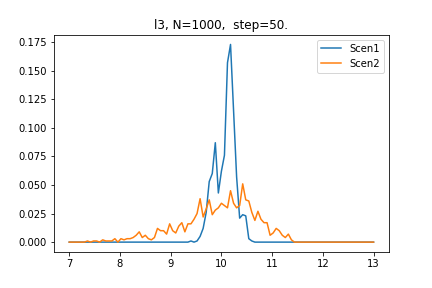}
\includegraphics[scale=0.3]{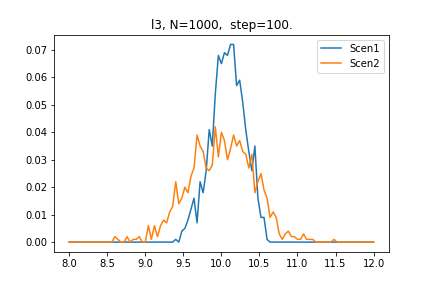}
\includegraphics[scale=0.3]{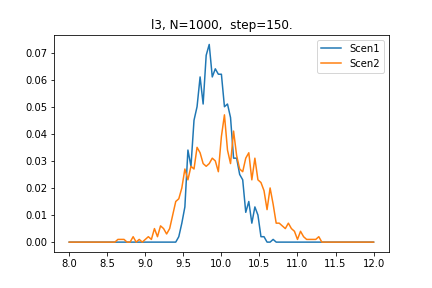}
\end{subfigure}
\caption{\label{fig:probestimation} Estimated distributions of the water levels in the tanks in the two scenarios.}
\end{figure}

\begin{exa}%
\label{ex:simulation}
A simulation of the three-tanks laboratory experiment is given in Figure~\ref{fig:simulation_tanks}, with the following setting:
\begin{enumerate*}[(i)]
\item $l_{min} = 0$,
\item $l_{max} = 20$,
\item $l_{goal} = 10$,
\item $\Delta_l = 0.5$,
\item $q_{max} = 6$,
\item $q_{step} = q_{max} / 5$,
\item $\Delta_q = 0.5$,
\item $\Delta\tau = 0.1$.
\end{enumerate*}
Moreover, we assume an initial \datastate{} $\ds_{0}$ with $l_i = l_{min}$ and $q_i = 0$, for $i = 1,2,3$.
The plot on the left hand side of Figure~\ref{fig:simulation_tanks} depicts a simulation run of the system $\system_1$, namely the system in which the environment evolution follows scenario 1 and having $c_1 = \config{\mathsf{P_{Tanks}}}{\ds_0}{\E_1}$ as initial configuration.
On the right hand side we have the corresponding plot for system $\system_2$, whose environment evolution follows scenario 2 and the initial configuration is thus $c_2 = \config{\mathsf{P_{Tanks}}}{\ds_0}{\E_2}$.

In Figure~\ref{fig:simulations_tank_by_tank} we consider the same simulations, but we compare the variations of the level of water, in time, in the same tank in the two scenarios.
This can help us to highlight the potential differences in the evolution of the system with respect to the two environments.

In Figure~\ref{fig:probestimation} we give an estimation of the distributions of $l_1,l_2,l_3$ obtained from our simulations of $\system_1$ and $\system_2$.
In detail, the three plots in the top row report the comparison of the distributions over $l_1$ obtained in the two scenarios at time step, respectively, $50, 100$ and $150$ when using $N=1000$ samples.
The three plots in the central row are the corresponding ones for $l_2$, and in the bottom row we have the plots related to $l_3$.
In both scenarios, we obtain Gaussian-like distributions, as expected given the probabilistic behaviour of the environment defined in Examples~\ref{ex:tanks_I} and~\ref{ex:tanks_IV}.
However, we observe that while in both cases the mean of the distributions is close to $l_{goal} = 10$, there is a significant difference in their variance: the variance of the distributions obtained in scenario 1 is generally smaller than that of those in scenario 2 (the curves for scenario 1 are ``slimmer'' than those for scenario 2).
We will see how this disparity is captured by the \spell metric.
\hfill\scalebox{0.85}{$\LHD$}
\end{exa}

As usual in the related literature, we can apply the classic \emph{standard error approach} to analyse the approximation error of our statistical estimation of the evolution sequences.
Briefly, we let $x \in \{l_1,l_2,l_3\}$ and we focus on
the distribution of the means of our samples (in particular we consider the cases $N = 1000, 5000, 10000$) for each variable $x$.
In each case, for each $i \in \{0,\dots,k\}$ we compute the mean $\overline{E_i}(x) = \frac{1}{N}\sum_{j = 1}^N \ds_i^j(x)$ of the sampled data (recall that $E_i$ is the tuple of the sampled configurations $c_i^1,\dots,c_i^N$, with $c_i^j=\config{\proc_i^j}{\ds_i^j}{\E}$), and we evaluate their \emph{standard deviation} $\tilde{\sigma}_{i,N}(x) = \sqrt{\frac{\sum_{j = 1}^{N} (\ds_i^j(x) - \overline{E_i}(x))^2}{N - 1}}$ (see Figure~\ref{fig:standard_deviation} for the variation in time of the standard deviation of the distribution of $x = l_3$).
From $\tilde{\sigma}_{i,N}(x)$ we obtain the \emph{standard error of the mean} $\overline{\sigma}_{i,N}(x) = \frac{\tilde{\sigma}_{i,N}(x)}{\sqrt{N}}$ (see Figure~\ref{fig:standard_error} for the variation in time of the standard error for the distribution of $x = l_3$).
Finally, we proceed to compute the $z$-\emph{score} of our sampled distribution as follows: $z_{i,N}(x) = \frac{\overline{E_i}(x) - \mathbb{E}(x)}{\overline{\sigma}_{i,N}(x)}$, where $\mathbb{E}(x)$ is the mean (or expected value) of the real distribution over $x$.
In Figure~\ref{fig:z_score} we report the variation in time of the $z$-score of the distribution over $x = l_3$: the dashed red lines correspond to $z = \pm 1.96$, namely the value of the $z$-score corresponding to a confidence interval of the $95\%$.
We can see that our results can already be given with a $95\%$ confidence in the case of $N = 1000$.
Please notice that the oscillation in time of the values of the $z$-scores is due to the perturbations introduced by the environment in the simulations and by the natural oscillation in the interval $[l_{goal}-\Delta_l,l_{goal}+\Delta_l]$ of the water levels in the considered experiment (see Figure~\ref{fig:simulation_tanks}).
A similar analysis, with analogous results, can be carried out for the distributions of $l_1$ and $l_2$.
In Figure~\ref{fig:total_z_score} we report the variation in time of the $z$-scores of the distributions of the three variables, in the case $N = 1000$.

\begin{figure}
\begin{subfigure}{0.47\textwidth}
\includegraphics[scale=0.35]{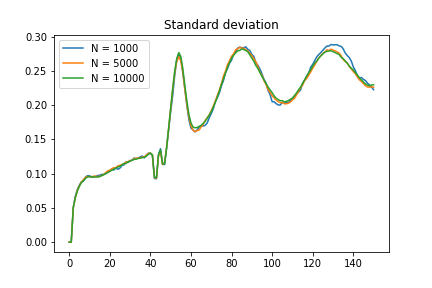}
\caption{Standard deviation for $l_3$.}%
\label{fig:standard_deviation}
\end{subfigure}
\hfill
\begin{subfigure}{0.47\textwidth}
\includegraphics[scale=0.37]{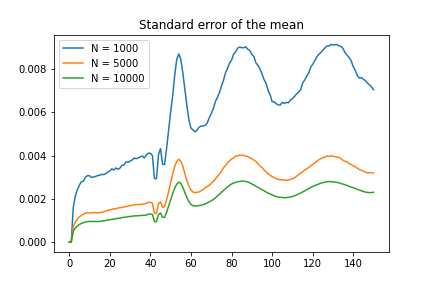}
\caption{Standard error for $l_3$.}%
\label{fig:standard_error}
\end{subfigure}
\begin{subfigure}{0.47\textwidth}
\includegraphics[scale=0.37]{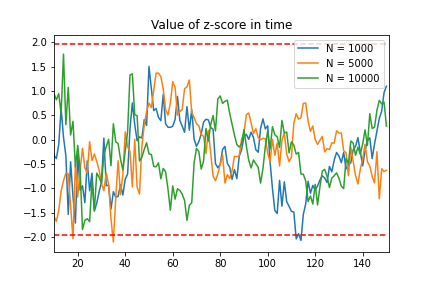}
\caption{$z$-score for $l_3$.}%
\label{fig:z_score}
\end{subfigure}
\hfill
\begin{subfigure}{0.47\textwidth}
\includegraphics[scale=0.35]{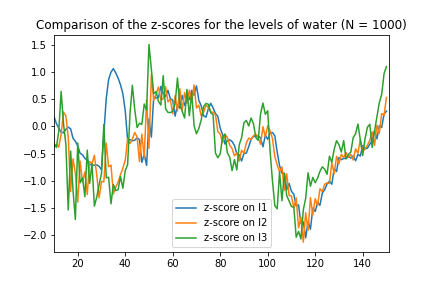}
\caption{$z$-scores for $l_1,l_2,l_3$, $N = 1000$.}%
\label{fig:total_z_score}
\end{subfigure}
\caption{Analysis of the approximation error, over time, of the distributions of $l_1,l_2,l_3$, in scenario 2, for $N = 1000,5000,10000$.}%
\label{fig:statistical_error}
\end{figure}

%=====================================================

\subsection{Computing distance between two configurations}%
\label{sec:computing_distance}

Function $\Call{Estimate}{}$ allows us to collect independent samples at each time step $i$ from $0$ to a deadline $k$. These samples can be used to estimate the distance between two configurations $c_1$ and $c_2$.
Following an approach similar to the one presented in~\cite{TK09}, to estimate the Wasserstein distance $\Wasserstein(m^{\D}_{\rho,i})$ between two (unknown) distributions $\ES^{\D}_{c_1,i}$ and $\ES^{\D}_{c_2,i}$ we can use $N$ independent samples $\{c^1_1,\ldots,c^N_1\}$ taken from $\ES^\C_{c_1,i}$ and $\ell\cdot N$ independent samples $\{c^1_2,\ldots,c_2^{\ell \cdot N}\}$ taken from $\ES^\C_{c_2,i}$.
After that, we exploit the $i$-penalty function $\rho_i$ and we consider the two sequences of values
$\omega_1,\dots,\omega_{N}$
and
$\nu_1,\dots,\nu_{\ell \cdot N}$
defined as follows:
\[
\{ \omega_j=\rho_i(\ds_1^j) \mid \config{\proc_1^j}{\ds_1^j}{\E_1}=c_1^j\}
\qquad
\{ \nu_h=\rho_i(\ds_2^h) \mid \config{\proc^h_2}{\ds^h_2}{\E_2}=c_2^h\}.
\]
We can assume, without loss of generality that these sequences are ordered, i.e., $\omega_j\leq \omega_{j+1}$ and $\nu_{h}\leq \nu_{h+1}$.
The value $\Wasserstein(m^{\D}_{\rho,i})(\ES^\D_{c_1,i},\ES^\D_{c_2,i})$ can be approximated as $\frac{1}{\ell N}\sum_{h=1}^{\ell N}\max\{\nu_{h} - \omega_{\lceil \frac{h}{\ell}\rceil},0\}$.
The next theorem, based on results in~\cite{TK09,Vil08,V74}, ensures that the larger the number of samplings the closer the gap between the estimated value and the exact one.

\begin{thm}%
\label{thm:estimate}
Let $\ES^\C_{c_1,i},\ES^\C_{c_2,i} \in \distrib(\C,\Sigma_\C)$ be unknown.
Let $\{c^1_1,\ldots,c^N_1\}$ be independent samples taken from $\ES^\C_{c_1,i}$, and $\{c^1_2,\ldots,c_2^{\ell \cdot N}\}$ independent samples taken from $\ES^\C_{c_2,i}$.
Let $\{ \omega_j = \rho_i(\ds^j_1)\}$ and $\{ \nu_h = \rho_i(\ds^h_2)\}$ be the ordered sequences obtained from the samples and the $i$-penalty function.
Then, the Wasserstein distance between $\ES^{\D}_{c_1,i}$ and $\ES^{\D}_{c_2,i}$ is equal, a.s., to
\[
\Wasserstein(m^{\D}_{\rho,i})(\ES^{\D}_{c_1,i},\ES^{\D}_{c_2,i})
=
\lim_{N\rightarrow \infty} \frac{1}{\ell N}\sum_{h=1}^{\ell N}\max\left\{\nu_{h} - \omega_{\lceil \frac{h}{\ell}\rceil},0\right\}.
\]
\end{thm}

%--------------------------------------------------
\begin{proof}
Let $c_1^j =\config{\proc^j_1}{\ds^j_1}{\E}$, for all $j=1,\dots,N$,
and $c_2^j =\config{\proc^j_2}{\ds^j_2}{\E}$, for all $j=1,\dots,\ell \cdot N$.
We split the proof into two parts showing respectively:
\begin{equation}%
\label{proof:one}
\Wasserstein(m^\D_{\rho,i})(\ES^\D_{c_1,i},\ES^\D_{c_2,i})
=
\lim_{N \to \infty} \Wasserstein(m^\D_{\rho,i})(\hat{\ES}^{\D,N}_{c_1,i},\hat{\ES}^{\D,\ell N}_{c_2,i})
\enspace .
\end{equation}
\begin{equation}%
\label{proof:two}
\Wasserstein(m^\D_{\rho,i})(\hat{\ES}^{\D,N}_{c_1,i},\hat{\ES}^{\D,\ell N}_{c_2,i})
=
\frac{1}{\ell N} \sum_{h = 1}^{\ell N} \max\left\{ \nu_h - \omega_{\lceil \frac{h}{\ell} \rceil},0 \right\}
\enspace .
\end{equation}
\begin{itemize}
\item {\sc Proof of Equation~\eqref{proof:one}.}

We recall that the sequence $\{\hat{\ES}^{\D,N}_{c_l,i}\}$ converges weakly to $\ES^\D_{c_l,i}$ for $l \in \{1,2\}$ (see Equation~\eqref{eq:weak_convergence}).
Moreover, we can prove that these sequences converge weakly in $\distrib(\D,\borel_\D)$ in the sense of~\cite[Definition 6.8]{Vil08}.
In fact, given the $i$-ranking function $\rho_i$, the existence of a \datastate{} $\tilde{\ds}$ such that $\rho_i(\tilde{\ds}) = 0$ is guaranteed (remember that the constraints used to define $\rho_i$ are on the possible values of state variables and a \datastate{} fulfilling all the requirements is assigned value $0$).
Thus, for any $\ds \in \D$ we have that
\[
m^\D_{\rho,i}(\tilde{\ds},\ds) = \max\{\rho_i(\ds) - \rho_i(\tilde{\ds}),0\} = \rho_i(\ds)
\enspace .
\]
Since, moreover, by definition $\rho_i$ is continuous and bounded, the weak convergence of the distributions gives
\begin{align*}
& \int_\D \rho_i(\ds) \dd(\hat{\ES}^{\D,N}_{c_1,i}(\ds)) \to \int_\D \rho_i(\ds) \dd(\ES^\D_{c_1,i}(\ds)) \\
& \int_\D \rho_i(\ds) \dd(\hat{\ES}^{\D, \ell N}_{c_2,i}(\ds)) \to \int_\D \rho_i(\ds) \dd(\ES^\D_{c_2,i}(\ds))
\end{align*}
and thus Definition 6.8.(i) of~\cite{Vil08} is satisfied. % chktex 36
As $\D$ is a Polish space, by~\cite[Theorem 6.9]{Vil08} we obtain that
\[
\hat{\ES}^{\D,N}_{c_l,i} \to \ES^\D_{c_l,i} \text{ implies }
\Wasserstein(m^\D_{\rho,i})(\ES^\D_{c_1,i},\ES^\D_{c_2,i})
=
\lim_{N \to \infty} \Wasserstein(m^\D_{\rho,i})(\hat{\ES}^{\D,N}_{c_1,i},\hat{\ES}^{\D,\ell N}_{c_2,i})
\enspace .
\]

\item {\sc Proof of Equation~\eqref{proof:two}.}

For this part of the proof we follow~\cite{TK09}.
Since the ranking function is continuous, it is in particular $\borel_\D$ measurable and therefore for any distribution $\mu$ on $(\D,\borel_\D)$ we obtain that
\[
F_{\mu,\rho_i}(r) := \mu(\{\rho_i(\ds) < r\})
\]
is a well defined cumulative distribution function.
In particular, for $\mu = \hat{\ES}^{\D,N}_{c_1,i}$ we have that
\[
F_{\hat{\ES}^{\D,N}_{c_1,i}, \rho_i}(r) = \hat{\ES}^{\D,N}_{c_1,i}(\{\rho_i(\ds) < r\}) =
\frac{\left| \{ \config{\proc_1^j}{\ds_1^j}{\E_1} \in E_{1,i} \mid \rho_i(\ds_1^j) < r \}\right|}{N}
\enspace .
\]
Since, moreover, we can always assume that the values $\rho_i(\ds_1^j)$ are sorted, so that $\rho_i(\ds_i^j) \le \rho_i(\ds_1^{j+1})$ for each $j = 1,\dots,N-1$, we can express the counter image of the cumulative distribution function as
\begin{equation}%
\label{eq:cdf_inverse}
F^{-1}_{\hat{\ES}^{\D,N}_{c_1,i}, \rho_i}(r) = \rho_i(\ds_1^j)
\text{ whenever } \frac{j-1}{N} < r \le \frac{j}{N}
\enspace .
\end{equation}
A similar reasoning holds for $F^{-1}_{\hat{\ES}^{\D,\ell N}_{c_2,i}, \rho_i}(r)$.

Then, by~\cite[Proposition 6.2]{FR18}, for each $N$ we have that
\[
\Wasserstein(m^\D_{\rho,i})(\hat{\ES}^{\D,N}_{c_1,i}, \hat{\ES}^{\D,\ell N}_{c_2,i}) =
\int_0^1
\max
\left\{
F^{-1}_{\hat{\ES}^{\D,\ell N}_{c_2,i}, \rho_i}(r) -
F^{-1}_{\hat{\ES}^{\D,N}_{c_1,i}, \rho_i}(r),
0
\right\}
\dd r
\enspace .
\]
Let us now partition the interval $[0,1]$ into $\ell N$ intervals of size $\frac{1}{\ell N}$, thus obtaining
\[
\Wasserstein(m^\D_{\rho,i})(\hat{\ES}^{\D,N}_{c_1,i}, \hat{\ES}^{\D,\ell N}_{c_2,i}) =
\sum_{h=1}^{\ell N}
\left(
\int_{\frac{h-1}{\ell N}}^{\frac{h}{\ell N}}
\max\left\{
F^{-1}_{\hat{\ES}^{\D,\ell N}_{c_2,i}, \rho_i}(r) -
F^{-1}_{\hat{\ES}^{\D,N}_{c_1,i}, \rho_i}(r),
0
\right\}
\dd r
\right)
\enspace .
\]
From Equation~\eqref{eq:cdf_inverse}, on each interval $\lopen{\frac{h-1}{\ell N},\frac{h}{\ell N}}$ it holds that $F^{-1}_{\hat{\ES}^{\D,N}_{c_1,i}, \rho_i}(r) = \rho_i(\ds_1^{\lceil \frac{h}{\ell} \rceil})$ and $F^{-1}_{\hat{\ES}^{\D,\ell N}_{c_2,i}, \rho_i}(r) = \rho_i(\ds_2^h)$.
Moreover, both functions are constant on each the interval so that the value of the integral is given by the difference multiplied by the length of the interval:
\begin{align*}
\Wasserstein(m^\D_{\rho,i})(\hat{\ES}^{\D,N}_{c_1,i}, \hat{\ES}^{\D,\ell N}_{c_2,i}) ={} &
\sum_{h=1}^{\ell N} \frac{1}{\ell N}
\max\left\{\rho_i(\ds_2^h) - \rho_i(\ds_1^{\lceil \frac{h}{\ell} \rceil}), 0 \right\} \\
= & \sum_{h=1}^{\ell N} \frac{1}{\ell N} \max\left\{\nu_h - \omega_{\lceil \frac{h}{\ell} \rceil}, 0 \right\}
\enspace .
\end{align*}
By substituting the last equality into Equation~\eqref{proof:one} we obtain the thesis.
\qedhere
\end{itemize}
\end{proof}
%------------------------------------------------------

\begin{figure}[tbp]
\begin{minipage}[t]{0.5\textwidth}
\small
\begin{algorithmic}[1]
\Function{Distance}{$c_1,c_2,\rho,\lambda,\OT,N,\ell$}
\State $k \gets \max\{i \in \OT\}$
\State $E_{1,1},\ldots,E_{1,k}\gets \Call{Estimate}{c_1,k,N}$
\State $E_{2,1},\ldots,E_{2,k}\gets \Call{Estimate}{c_2,k,\ell N}$
\State $m \gets 0$
\ForAll{$i\in \OT$}
	\State $m_i \gets \Call{ComputeW}{E_{1,i},E_{2,i},\rho_i}$
	\State $m \gets \max\{ m, \lambda(i)\cdot m_i \}$
\EndFor
\State \Return $m$
\EndFunction
\end{algorithmic}
\end{minipage}\hfill
\small
\begin{minipage}[t]{0.5\textwidth}
\begin{algorithmic}[1]
\Function{ComputeW}{$E_1,E_2,\rho$}
\State $(\config{\proc^1_1}{\ds^1_1}{\E_1},\ldots,\config{\proc^N_1}{\ds^N_1}{\E_1})\gets E_1$
\State $(\config{\proc^1_2}{\ds^1_2}{\E_2},\ldots,\config{\proc_2^{\ell N}}{\ds_2^{\ell N}}{\E_2})\gets E_2$
\State $\forall j: (1\leq j\leq N): \omega_j\gets\rho(\ds^j_1)$
\State $\forall h: (1\leq h\leq \ell N): \nu_h\gets\rho(\ds^h_2)$
\State re index $\{\omega_j\}$ s.t.\ $\omega_j\leq \omega_{j+1}$
\State re index $\{\nu_h\}$ s.t.\ $\nu_h\leq \nu_{h+1}$
\State \Return $\frac{1}{\ell N}\sum_{h=1}^{\ell N}\max\{ \nu_h - \omega_{\lceil \frac{h}{\ell}\rceil}, 0\}$
\EndFunction
\end{algorithmic}
\end{minipage}
\caption{Functions used to estimate the \spell metric on systems.}%
\label{alg:computedistance}
\end{figure}

\noindent
The procedure outlined above is realised by functions $\Call{Distance}{}$ and $\Call{ComputeW}{}$ in Figure~\ref{alg:computedistance}.
The former takes as input the two configurations to compare, the penalty function (seen as the sequence of the $i$-penalty functions), the discount function $\lambda$, the set $\OT$ of observation times which we assume to be bounded, and the parameters $N$ and $\ell$ used to obtain the samplings of computation.

Function $\Call{Distance}{}$ collects the samples $E_i$ of possible computations during the observation period $[0,\max_{\OT}]$, where $\max_{\OT}$ denotes the last observation time.

Then, for each observation time $i \in \OT$, the distance at time $i$ is computed via the function $\Call{ComputeW}{E_{1,i},E_{2,i},\rho_i}$.

Since the penalty function allows us to reduce the evaluation of the Wasserstein distance in $\real^n$ to its evaluation on $\real$, due to the sorting of $\{\nu_h \mid h \in [1,\dots,\ell N]\}$ the complexity of function $\Call{ComputeW}{}$ is $O(\ell N \log(\ell N))$ (cf.~\cite{TK09}).
We refer the interested reader to~\cite[Corollary 3.5, Equation (3.10)]{SFGSL12} for an estimation of the approximation error given by the evaluation of the Wasserstein distance over $N$ samples.

%==============================================================
%===============================================================

\subsection{Analysis of the three-tanks experiment}

Our aim is now to use the \spell metric, and the algorithms introduced above, to study various properties of the behaviour of the three-tanks system.
In particular, we consider the following two objectives:
\begin{enumerate}
\item Comparison of the behaviour of two systems generated from the interaction of the same program with two different environments, under the same initial conditions.
\item Comparison of the behaviour of two systems generated from the interaction of two different programs with the same environment, under the same initial conditions.
\end{enumerate}

%====================================================

\subsubsection{Same program, different environments}

Consider the systems $\system_1$ and $\system_2$ introduced in Example~\ref{ex:simulation}.
We recall that the initial configuration of $\system_i$ is $c_i = \config{\mathsf{P_{Tanks}}}{\ds_0}{\E_i}$, for $i =1,2$.
Notice that since $\system_1$ and $\system_2$ are distinguished only by the environment function, a comparison of their behaviour by means of the \spell metric allows us to establish in which scenario the program $\mathsf{P_{Tanks}}$ performs better with respect to the aforementioned tasks.

\begin{figure}[t]
\begin{subfigure}{0.47\textwidth}
\includegraphics[scale=0.45]{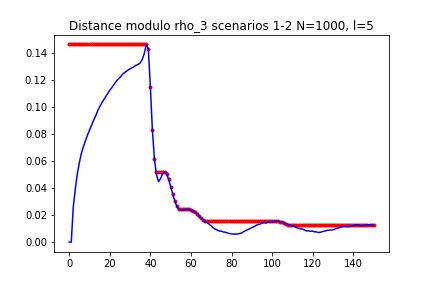}
\caption{Evaluation of $\m^{\uno}_{\rho^{l_3},\OT}(c_1,c_2) (\sim 0.15)$.}%
\label{fig:diff_env_12}
\end{subfigure}
\hfill
\begin{subfigure}{0.47\textwidth}
\includegraphics[scale=0.45]{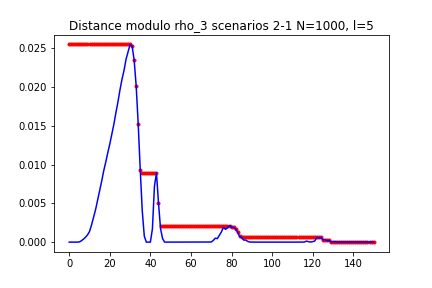}
\caption{Evaluation of $\m^{\uno}_{\rho^{l_3},\OT}(c_2,c_1) (\sim 0.02)$.}%
\label{fig:diff_env_21}
\end{subfigure}
\caption{Evaluation of the distance between $\system_1$ and $\system_2$ (and $\system_2,\system_1$) with respect to the target $l_3 = l_{goal}$.
A \textcolor{red}{red dot} at coordinates $(x,y)$ means that $\m^{\uno}_{\rho^{l_3},\OT_x}(c_1,c_2) = y$, where $\OT_x = \{\tau \in \OT \mid \tau \ge x\}$.
The \textcolor{blue}{blue line} is the pointwise Wasserstein distance between the distributions in the \tracciones{} of $c_i$ and $c_{3-i}$.}%
\label{fig:esperimento1}
\end{figure}

Firstly, we focus on a single tank by considering the target $l_3 = l_{goal}$, and thus the penalty function $\rho^{l_3}$.
We do not consider any discount (so $\lambda$ is the constant function $\uno$, i.e., $\lambda(\tau) = 1$ for all $\tau$) and we let $\OT = \nat \cap [0,150]$.
In Figure~\ref{fig:esperimento1} we show the evaluation of $\m^{\uno}_{\rho^{l_3},\OT}$ between $\system_1$ and $\system_2$ (Figure~\ref{fig:diff_env_12}) and between $\system_2$ and $\system_1$ (Figure~\ref{fig:diff_env_21}) over a simulation with $5000$ samples ($N=1000$, $l=5$).
Notice that the first distance is $\m^{\uno}_{\rho^{l_3},\OT}(c_1,c_2) \sim 0.15$, whereas the second one is $\m^{\uno}_{\rho^{l_3},\OT}(c_2,c_1) \sim 0.02$.
The different scale of the two distances (the latter is less than 1/7 of the former) already gives us an intuition of the advantages the use of a hemimetric gives us over a (pseudo)metric in terms of expressiveness. % chktex 36

Consider Figure~\ref{fig:diff_env_12}.
Recall that, given the definition of our hemimetric on \datastates, the value of $\m^{\uno}_{\rho^{l_3},\OT}(c_1,c_2)$ expresses how much the probabilistic evolution of the data in the \traccione{} of $c_2$ is worse than that of $c_1$, with respect to the task identified by $\rho^3$.
In other words, the higher the value of $\m^{\uno}_{\rho^{l_3},\OT}(c_1,c_2)$, the worse the performance of $c_2$ with respect to $c_1$.
Yet, for those who are not familiar with the Wasserstein distance, the plot alone can be a little foggy, so let us shed some light.
Notice that, although after the first $40$ steps the distance decreases considerably, it never reaches $0$.
This means that at each time step, the distribution in the \traccione{} of $c_2$ assigns positive probability to (at least) one measurable set of \datastates{} in which $l_3$ is farther from $l_{goal}$ than in all the measurable sets of \datastates{} in the support of the distribution reached at the same time by $c_1$.
This is perfectly in line with our observations on the differences between the variances of the distributions in Figure~\ref{fig:probestimation} (cf.\ the bottom row pertaining $l_3$).

A natural question is then why the distance $\m^{\uno}_{\rho^{l_3},\OT}(c_2,c_1)$, given in Figure~\ref{fig:diff_env_21}, is not equal to the constant $0$.
This is due to the combination of the Wasserstein distance with the hemimetric $m^{\D_{3T}}_{\rho^{l_3}}$.
Since $m^{\D_{3T}}_{\rho^{l_3}}(\ds_1,\ds_2) = 0$ whenever $\rho^{l_3}(\ds_2) < \rho^{l_3}(\ds_1)$ and the Wasserstein lifting is defined as the infimum distance over the couplings, the (measurable sets of) \datastates{} reached by $c_2$ that are better, with respect to $\rho^{l_3}$, than those reached by $c_1$, are ``hidden'' in the evaluation of $\m^{\uno}_{\rho^{l_3},\OT}(c_1,c_2)$.
However, they do contribute to the evaluation of $\m^{\uno}_{\rho^{l_3},\OT}(c_2,c_1)$.
Clearly, once we compute both $\m^1_{\rho^{l_3},\OT}(c_1,c_2)$ and $\m^{\uno}_{\rho^{l_3},\OT}(c_2,c_1)$, we can observe that the latter distance is less than $1/7$ of the former and we can conclude thus that $\system_1$ shows a better behaviour than $\system_2$.

As a final observation on Figure~\ref{fig:diff_env_21}, we notice that, for instance, in the time interval $I = [50,70]$ the pointwise distance between the \tracciones{} is $0$.
This means that for each (measurable set) \datastate{} $\ds_2$ in the support of $\ES^{\D_{3T}}_{c_2,\tau}$, $\tau \in I$, there is one (measurable set) \datastate{} $\ds_1$ in the support of $\ES^{\D_{3T}}_{c_1,\tau}$, such that $m^{\D_{3T}}_{\rho^{l_3}}(\ds_2,\ds_1)=0$.

\begin{figure}
\begin{subfigure}{0.47\textwidth}
\includegraphics[scale=0.45]{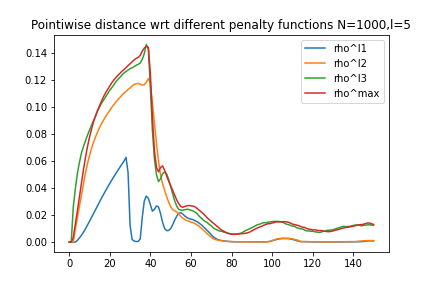}
\caption{Evaluation of $\Wasserstein(m^{\D_{3T}}_{\rho})(\ES^{\D_{3T}}_{c_1,\tau},\ES^{\D_{3T}}_{c_2,\tau})$ for $\rho \in \{\rho^{l_1},\rho^{l_2},\rho^{l_3},\rho^{\max}\}$.}%
\label{fig:diff_rho_12}
\end{subfigure}
\hfill
\begin{subfigure}{0.47\textwidth}
\includegraphics[scale=0.45]{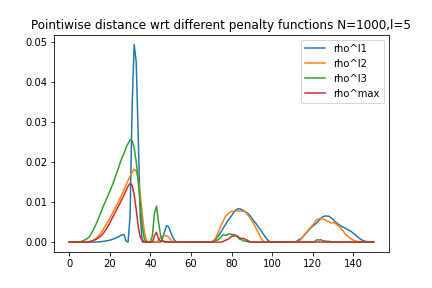}
\caption{Evaluation of $\Wasserstein(m^{\D_{3T}}_{\rho})(\ES^{\D_{3T}}_{c_2,\tau},\ES^{\D_{3T}}_{c_1,\tau})$ for $\rho \in \{\rho^{l_1},\rho^{l_2},\rho^{l_3},\rho^{\max}\}$.}%
\label{fig:diff_rho_21}
\end{subfigure}
\caption{Evaluation of the \spell metric with respect to different targets.}%
\label{fig:diff_rho}
\end{figure}

Clearly, one can repeat a similar analysis for the other tasks of the system.
To give a general idea of the difference in the behaviour of the system in the two scenarios, in Figure~\ref{fig:diff_rho} we report the evaluation of the pointwise distance between the \tracciones{} of $\system_1$ and $\system_2$ (Figure~\ref{fig:diff_rho_12}), and viceversa (Figure~\ref{fig:diff_rho_21}).

%=======================================================

\subsubsection{Different programs, same environment}

We introduce two new programs: $\mathsf{P_{Tanks}^+}$ and $\mathsf{P_{Tanks}^-}$.
They are defined exactly as $\mathsf{P_{Tanks}}$, the only difference being the value of $\Delta_l$ used in the if/else guards in $\mathsf{P_{in}}$ and $\mathsf{P_{out}}$.
For $\mathsf{P_{Tanks}}$ we used $\Delta_l=0.5$; we use $\Delta_l^+ = 0.7$ for $\mathsf{P_{Tanks}^+}$, and $\Delta_l^-=0.3$ for $\mathsf{P_{Tanks}^-}$.
Our aim is to compare the behaviour of $\system_1$ with that of systems $\system^+$ and $\system^-$ having as initial configurations, respectively, $c^+=\config{\mathsf{P_{Tanks}^+}}{\ds_0}{\E_1}$ and $c^-=\config{\mathsf{P_{Tanks}^-}}{\ds_0}{\E_1}$.
Notice that the three systems are characterised by the same environment evolution $\E_1$ and the same initial \datastate{} $\ds_0$.
Moreover, as outlined above, we see the value of $\Delta_l$ as an inherent property of the program.

So, while the outcome of the experiment will be exactly the one the reader expects, this example shows that with our framework we can compare the ability of different programs to fulfil their objectives when operating on equal footing, thus offering us the possibility of choosing the most suitable one according to (possibly) external factors.
For instance, we can imagine that while $\Delta_l^-$ allows for a finer control, it also causes more frequent modifications to the flow rate $q_1$ than $\Delta_l$ does.
These may in turn entail a higher risk of a breakdown of the pump attached to the first tank.
In this case, we can use the \spell metric to decide whether the difference in the performance of the programs justifies the higher risk, or if $\mathsf{P_{Tanks}}$ is still ``good enough'' for our purposes.

\begin{figure}
\begin{subfigure}{0.47\textwidth}
\includegraphics[scale=0.45]{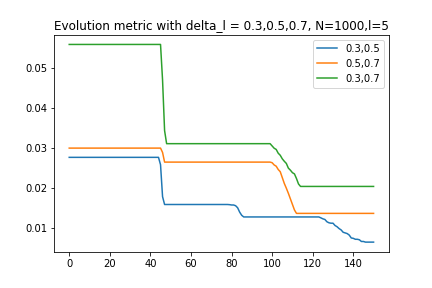}
\caption{Evolution metrics.}%
\label{fig:diff_prog_ev}
\end{subfigure}
\hfill
\begin{subfigure}{0.47\textwidth}
\includegraphics[scale=0.45]{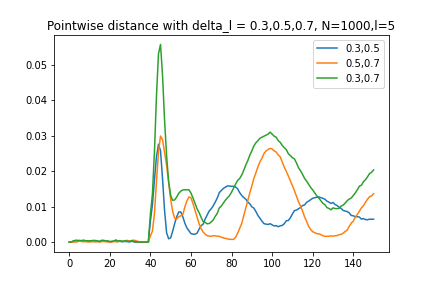}
\caption{Pointwise distances.}%
\label{fig:diff_prog_pt}
\end{subfigure}
\caption{Evaluation of the distances between $\mathsf{P_{Tanks}^-}, \mathsf{P_{Tanks}}, \mathsf{P_{Tanks}^+}$ with respect to $\rho^{l_3}$.}%
\label{fig:diff_prog}
\end{figure}

In Figure~\ref{fig:diff_prog} we present the results of the evaluations of $\m^{\uno}_{\rho^{l_3},\OT}(c^-,c_1)$, $\m^{\uno}_{\rho^{l_3},\OT}(c_1,c^+)$, and $\m^{\uno}_{\rho^{l_3},\OT}(c^-,c^+)$, where the target of the programs is $l_3 = l_{goal}$ (i.e., the penalty function $\rho^{l_3}$ is considered).
In detail, in Figure~\ref{fig:diff_prog_ev} we report the values of the three \spell metrics (see Figure~\ref{fig:esperimento1} for an explanation of how to read the graph), while in Figure~\ref{fig:diff_prog_pt} we report the pointwise (with respect to time) Wasserstein distances between the distributions in the \tracciones{} of the systems.

%==============================================
% sec - properties
%===============================================

\section{Robustness of programs}%
\label{sec:properties}

We devote this section to the analysis of the \emph{robustness} of programs with respect to a \datastate{} and an environment.
We also introduce two sensitivity properties of programs, which we call \emph{adaptability} and \emph{reliability}, entailing the ability of the program to induce a similar behaviour in systems that start operating from similar initial conditions.

%==================================================
%====================================================

\subsection{Robustness}

A system is said to be \emph{robust} if it is able to function correctly even in the presence of uncertainties.
In our setting, we formalise the \emph{robustness} of programs by \emph{measuring} their capability to \emph{tolerate} perturbations in the environmental conditions.

First of all, we characterise the perturbations with respect to which we want the program to be robust.
Given a program $\proc$, a \datastate{} $\ds$, and an environment $\E$, we consider all \datastates{} $\ds'$ such that the behaviour of the original configuration $\config{\proc}{\ds}{\E}$ is at distance at most $\eta_1$, for a chosen $\eta_1 \in \ropen{0,1}$, from the behaviour of $\config{\proc}{\ds'}{\E}$.
The distance is evaluated as the \spell metric with respect to a chosen penalty function $\rho$ and a time interval $I$.
The idea is that $\rho$ and $I$ can be used to characterise any particular situation we might want to study.
For instance, if we think of a drone autonomously setting its trajectory, then $I$ can be the time window within which a gust of wind occurs, and $\rho$ can capture the distance from the intended trajectory.

We can then proceed to measure the robustness.
Let $\rho'$ be the penalty function related to the (principal) target of the system, and let $\tilde{\tau}$ be an observable time instant.
We evaluate the \spell metric, with respect to $\rho'$, between the behaviour of $\config{\proc}{\ds}{\E}$ shown after time $\tilde{\tau}$, and that of $\config{\proc}{\ds'}{\E}$ for all the perturbations $\ds'$ obtained above.
We say that the robustness of $\proc$ with respect to $\ds$ and $\E$ is $\eta_2 \in \ropen{0,1}$, if $\eta_2$ is an upper bound for all those distances.

Notice that we use two penalty functions, one to identify the perturbations ($\rho$), and one to measure the robustness ($\rho'$), which can be potentially related to different targets of the system.
For instance, in the aforementioned case of the drone, $\rho'$ can be related to the distance from the final location, or to battery consumption.
This allows for an analysis of programs robustness in a general setting, where the perturbations and the (long term) system behaviour take into account different sets of variables.
Moreover, the time instant $\tilde{\tau}$ plays the role of a \emph{time bonus} given to the program to counter the effects of perturbations: differences detectable within time $\tilde{\tau}$ are not considered in the evaluation of the robustness threshold.

Definition~\ref{def:robustness} formalises the intuitions given above.

\begin{defi}
[Robustness]%
\label{def:robustness}
Let $\rho,\rho'$ be two penalty functions,
$I$ a time interval, $\tilde{\tau} \in \OT$, and $\eta_1,\eta_2 \in \ropen{0,1}$.
We say that $\proc$ is $(\rho,\rho',I,\tilde{\tau},\eta_1,\eta_2)$-\emph{robust} with respect to the \datastate{} $\ds$ and the environment evolution $\E$ if
\[
\forall\, \ds' \in \D \text{ with }
\m^{\lambda}_{\rho, \OT \cap I}(\config{\proc}{\ds}{\E},\config{\proc}{\ds'}{\E}) \le \eta_1
\]
it holds that
\[
\m^{\lambda}_{\rho',\{\tau \in \OT \mid \tau \ge \tilde{\tau}\}}(\config{\proc}{\ds}{\E},\config{\proc}{\ds'}{\E}) \le \eta_2
\enspace .
\]
\end{defi}

We can use our algorithm to measure the robustness of a given program.
Given a configuration $\config{\proc}{\ds}{\E}$, a set $\OT$ of observation times, a given threshold $\eta_1 \ge 0$, a penalty function $\rho$, and a time interval $I$, we can sample $M$ variations $\{ \ds_1,\ldots,\ds_M \}$ of $\ds$, such that for any $i \in \{1,\dots, M\}$, $\m^{\lambda}_{\rho,\OT\cap I}(\ds,\ds_i)\leq \eta_1$.
Then, for each sampled \datastate{} $\ds_i$ we can estimate the distance between $c = \config{\proc}{\ds}{\E}$ and $c_i = \config{\proc}{\ds_i}{\E}$, with respect to $\rho'$, after time $\tilde{\tau}$, namely $\m^{\lambda}_{\rho',\{\tau \in \OT \mid \tau \ge \tilde{\tau}\}}(c,c_i)$ for any $\tilde{\tau}\in \OT$.
Finally, for each $\tilde{\tau}\in \OT$, we let
\[
\xi_{\tilde{\tau}}=
\max_{i \in \{1,\dots, M\}} \{ \m^{\lambda}_{\rho',\{\tau \in \OT \mid \tau \ge \tilde{\tau}\}}(c,c_i) \}
\enspace .
\]
For the chosen $\eta_1,\rho,I,\rho'$ and $\tilde{\tau}$, each $\xi_{\tilde{\tau}}$ gives us a lower bound to $\eta_2$, and thus the robustness of the program.

\begin{figure}
\includegraphics[scale=0.5]{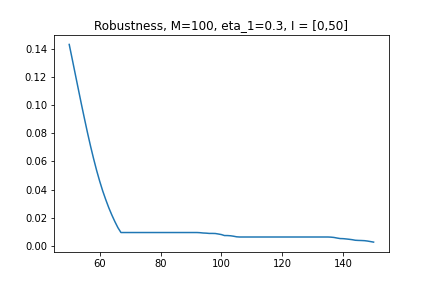}
\caption{Robustness of $\mathsf{P_{Tanks}}$ with respect to $\ds_s$ and $\E_1$, for $M = 100$, $\eta_1 = 0.3$, and $I = [0,50]$.}%
\label{fig:tanks_robust_scen1}
\end{figure}

\begin{exa}
We can use our definition of robustness to show that whenever the program is able to (initially) keep the difference with respect to the target $l_3 = l_{goal}$ bounded, then also the distance with respect to the target $l_1=l_{goal}=l_2$ will be bounded.
Formally, we instantiate the parameters of Definition~\ref{def:robustness} as follows: $\rho = \rho^{l_3}$, $\rho' = \frac{1}{2}\rho^{l_1} + \frac{1}{2}\rho^{l_2}$, $I = [0,50]$, and $\eta_1 = 0.3$.
Figure~\ref{fig:tanks_robust_scen1} reports the evaluation of $\xi_{50}$ with respect to $M = 100$ variations $\ds'$ satisfying $\m^{\uno}_{\rho,\OT\cap I}(\config{\mathsf{P_{Tanks}}}{\ds_s}{\E_1}, \config{\mathsf{P_{Tanks}}}{\ds'}{\E_1}) \le 0.3$.
\hfill\scalebox{0.85}{$\LHD$}
\end{exa}

%=======================================================
%=========================================================

\subsection{Adaptability and reliability}

We define the sensitivity properties of \emph{adaptability} and \emph{reliability} as particular instances of the notion of robustness.
Briefly, adaptability considers only one penalty function, and reduces the time interval $I$ to the sole initial time step $0$.
Then, reliability also fixes $\tilde{\tau} = 0$.
Let us present them in detail.

The notion of adaptability imposes some constraints on the \emph{long term} behaviour of systems, disregarding their possible initial dissimilarities.
Given the thresholds $\eta_1,\eta_2 \in \ropen{0,1}$ and an observable time $\tilde{\tau}$, we say that a program $\proc$ is adaptable with respect to a \datastate{} $\ds$ and an environment evolution $\E$ if whenever $\proc$ starts its computation from a \datastate{} $\ds'$ that differs from $\ds$ for at most $\eta_1$, then we are guaranteed that the distance between the \tracciones{} of the two systems after time $\tilde{\tau}$ is bounded by $\eta_2$.

\begin{defi}
[Adaptability]%
\label{def:adaptability}
Let $\rho$ be a penalty function, $\tilde{\tau} \in \OT$ and $\eta_1,\eta_2 \in \ropen{0,1}$.
We say that $\proc$ is $(\tilde{\tau},\eta_1,\eta_2)$-\emph{adaptable} with respect to the \datastate{} $\ds$ and the environment evolution $\E$ if
$\forall\, \ds' \in \D$ with $m^\D_{\rho,0}(\ds,\ds') \le \eta_1$ it holds $\m^{\lambda}_{\rho,\{\tau \in \OT \mid \tau \ge \tilde{\tau}\}}(\config{\proc}{\ds}{\E},\config{\proc}{\ds'}{\E}) \le \eta_2$.
\end{defi}

We remark that one can always consider the \datastate{} $\ds$ as the ideal model of the world used for the specification of $\proc$, and the \datastate{} $\ds'$ as the real world in which $\proc$ has to execute.
Hence, the idea behind adaptability is that even if the initial behaviour of the two systems is quite different, $\proc$ is able to reduce the gap between the real evolution and the desired one within the time threshold $\tilde{\tau}$.
Notice that being $(\tilde{\tau},\eta_1,\eta_2)$-adaptable for $\tilde{\tau} = \min \{\tau \mid \tau \in \OT\}$ is equivalent to being $(\tau,\eta_1,\eta_2)$-adaptable for all $\tau \in \OT$.

The notion of reliability strengthens that of adaptability by bounding the distance on the \tracciones{} from the beginning.
A program is reliable if it guarantees that small variations in the initial conditions cause only bounded variations in its evolution.

\begin{defi}
[Reliability]%
\label{def:reliability}
Let $\rho$ be a penalty function and $\eta_1,\eta_2 \in \ropen{0,1}$.
We say that $\proc$ is $(\eta_1,\eta_2)$-\emph{reliable} with respect to the \datastate{} $\ds$ and the environment evolution $\E$ if
$\forall\, \ds' \in \D$ with $m^\D_{\rho,0}(\ds,\ds') \le \eta_1$ it holds $\m^{\lambda}_{\rho,\OT}(\config{\proc}{\ds}{\E},\config{\proc}{\ds'}{\E}) \le \eta_2$.
\end{defi}

The same strategy depicted above for the evaluation of the robustness can be applied also in the case of adaptability and reliability.
Briefly, we simulate the behaviour of $M$ variations satisfying the requirement on the initial distance between \datastates{} $m^{\D}_{\rho,0}(\ds,\ds_i)\leq \eta_1$, and use the bounds $\xi_{\tilde{\tau}}$ to estimate the adaptability bound $\eta_2$.
Similarly, for $\tau_{\min}= \min_{\OT} \tau$, $\xi_{\tau_{\min}}$ gives a lower bound for its reliability.

\begin{figure}
\begin{subfigure}{0.47\textwidth}
\includegraphics[scale=0.45]{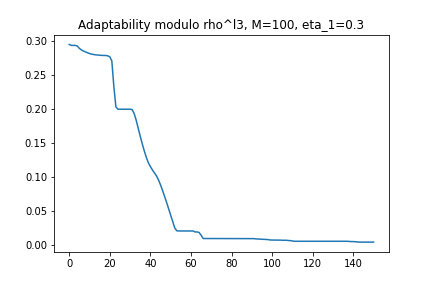}
\caption{Adaptability modulo $\rho^{l_3}$.}%
\label{fig:tanks_adapt_r3}
\end{subfigure}
\begin{subfigure}{0.47\textwidth}
\includegraphics[scale=0.45]{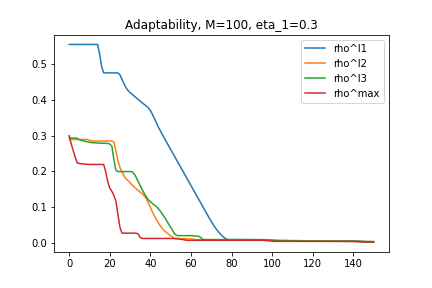}
\caption{Comparison of adaptability modulo different penalty functions.}%
\label{fig:tanks_adapt_all}
\end{subfigure}
\caption{Adaptability of $\mathsf{P_{Tanks}}$ with respect to $\ds_s$ and $\E_1$, for $M=100$ and $\eta_1 = 0.3$.}%
\label{fig:tanks_adaptability}
\end{figure}

\begin{exa}%
\label{ex:tanks_adaptability}
We provide an example of the evaluation of the adaptability of $\mathsf{P_{Tanks}}$.
We consider the environment evolution $\E_1$, i.e., the one corresponding to scenario 1, and, as initial \datastate, the \datastate{} $\ds_s$ such that $\ds_s(q_i) = 0$ and $\ds_s(l_i) = 5$, for $i = 1,2,3$.
The rationale to consider $\ds_s$ in place of $\ds_0$ is that $\rho(\ds_0) = 1$ and thus $m^{\D_{3T}}_{\rho}(\ds_0,\ds') = 0$ for all \datastates{} $\ds'$, for any penalty function $\rho \in \{\rho^{l_1},\rho^{l_2},\rho^{l_3},\rho^{\max}\}$.
Conversely, there are several \datastates{} $\ds'$ such that $m^{\D_{3T}}_{\rho}(\ds_s,\ds') > 0$, so that we can take into account the behaviour of the system when starting from ``worse'' initial conditions.
Let us stress that since the penalty functions $\rho$ are evaluated only on (some of) the values $l_i$, for $i = 1,2,3$, the \datastates{} $\ds'$ that are within distance $\eta_1$ from $\ds_s$ can have any value in $\D_q$ initially assigned to $q_j$, $j =1,2,3$ (and any value in $\D_l$ initially assigned to the $l_j$ that are possibly not taken into account in the evaluation of $\rho$).

To obtain the results reported in Figure~\ref{fig:tanks_adaptability} we considered, for each $\rho \in \{\rho^{l_1},\rho^{l_2},\rho^{l_3},\rho^{\max}\}$, $M = 100$ variations of the \datastate{} $\ds_s$ within the (closed) ball of radius $\eta_1 = 0.3$ with respect to $m^{\D_{3T}}_{\rho}$, and performed each simulation over $5000$ samples ($N=1000$, $\ell=5$), keeping $k=150$ as time horizon.
Figure~\ref{fig:tanks_adapt_r3} represents the adaptability of $\mathsf{P_{Tanks}}$ with respect to $\ds$ and $\E_1$ for $\rho = \rho^{l_3}$.
We can then infer that $\mathsf{P_{Tanks}}$ is $(30,0.3,0.2)$-adaptable and $(50,0.3,0.05)$-adaptable with respect to $\ds_s$ and $\E_1$, when considering the target $l_3 = l_{goal}$.
Please notice that one should always take into consideration the approximation error related to the computation of the Wasserstein distance~\cite[Corollary 3.5, Equation (3.10)]{SFGSL12} for an exact evaluation.
Yet, in our presentation we are more interested in showing the method rather than making exact calculations.
Figure~\ref{fig:tanks_adapt_all} offers a comparison of the adaptability modulo the various penalty functions defined for the three-tanks system.
\hfill\scalebox{0.85}{$\LHD$}
\end{exa}

%====================================================================
% case study: motore
%===================================================================

\section{Case-study: engine system}%
\label{sec:engine}

In this section we analyse a case study proposed in~\cite{LMMT21} and consisting in two supervised self-coordinating refrigerated engine systems that are subject to cyber-physical attacks aiming to inflict \emph{overstress of equipment}~\cite{GGIKLW2015}.

Since the impact of the attacks can be viewed as a flaw on system behaviour, in our context it will be quantified by employing suitable penalty functions.
In particular, the impact of attacks aiming to overstress the equipment can be quantified by adopting two different approaches: by simply measuring the level of equipment's stress or by comparing such a value with the alarm signals generated by the system when attempting to perform attack detection.
In the former case one focuses only on damages inflicted by attacks, in the latter case one is mainly interested in false positives and false negatives raised by the detection system.
In the former case we will use penalty functions quantifying the level of stress of system's equipment, in the latter case penalty functions will quantify false negatives and false positives that are produced when attempting to detect attacks and to raise alarm events.
In both cases, our simulations will allow us to \emph{estimate the distance between the genuine system and the system under attack}.
Essentially, these simulations will allow us to estimate the impact of attacks, thus quantifying how worse the system under attack behaves with respect to the genuine one.

%======================================================
%========================================================

\subsection{Modelling of the engine system}

We start by describing a single engine.
Then, we will describe the system obtained by combining two identical engines with a supervisor.
Finally, we will introduce cyber-physical attacks tampering with a single engine or with the whole system.
Interestingly, our choice of modelling the program and the environment separately well supports the specifications of these attacks:
since the attacks consist in pieces of malicious code that are somehow injected in the logic of the system, an attack specified by a process $A$ can be easily injected in a genuine system
$\config{\proc}{\ds}{\E}$
by composing $\proc$ and $A$ in parallel, thus giving
$\config{\proc \cmerge A}{\ds}{\E}$.

The engine we consider is a CPS whose logic aims to perform three tasks:
\begin{enumerate*}[(i)]
\item regulate the speed,
\item maintain the temperature within a specific range by means of a cooling system, and
\item detect anomalies.
\end{enumerate*}
The first two tasks are on charge of a controller, the last one is took over by an intrusion detection system, henceforth IDS\@.

The \dataspace{} consists of the following variables:
\begin{enumerate}[i.]
\item $\mathit{temp}$, representing a sensor detecting the temperature of the engine;
\item $\mathit{cool}$, representing an actuator to turn $\mathsf{on}$/$\mathsf{off}$ the cooling system;
\item $\mathit{speed}$, representing an actuator to regulate the engine speed, with values in the set $\{ \mathsf{slow}, \mathsf{half}, \mathsf{full}\}$, where $\mathsf{half}$ is the normal setting;
\item $\mathit{ch\_speed}$, representing a channel used by the IDS to command the controller to set the actuator $\mathit{speed}$ at $\mathsf{slow}$, when an anomaly is detected, or at $\mathsf{half}$, otherwise;
\item $\mathit{ch\_in}$ and $\mathit{ch\_out}$, representing two channels used for communication between the two engines:
when an IDS commands to the controller of its own engine to operate at $\mathsf{slow}$ speed, in order to compensate the consequent lack of performance it
asks to the controller of the other engine to operate at $\mathsf{full}$ speed;
\item $\mathit{stress}$, denoting the level of stress of the engine, due to its operating conditions;
\item six variables $\mathit{p_k}$, for $ 1 \leq k \leq 6$, recording the temperatures in the last six time instants;
\item $\mathit{temp\_fake}$, representing a noise introduced by malicious activity of attackers tempering with sensor $\mathit{temp}$:
we assume that the IDS can access the genuine value $\mathit{temp}$, whereas the controller receives the value of $\mathit{temp}$ through an insecure channel $\mathit{ch\_temp}$ that may be compromised and whose value is obtained by summing up the noise $\mathit{temp\_fake}$ and $\mathit{temp}$;
\item
$\mathit{ch\_warning}$, representing a channel used by the IDS to raise anomalies.
\end{enumerate}

\noindent
The environment evolution $\E$ affects these variables as follows:
\begin{enumerate}
\item $\mathit{temp}$ changes by a value determined probabilistically according to a uniform distribution ${\mathcal U}(I)$, where $I$ is the interval $[-1.2,-0.8]$ if the cooling system is active ($\mathit{cool} = \mathsf{on}$), whereas if the cooling system is inactive ($\mathit{cool} = \mathsf{off}$) then $I$ depends on the speed:
if $\mathit{speed}$ is $\mathsf{slow}$ then $I=[0.1,0.3]$,
if $\mathit{speed}$ is $\mathsf{half}$ then $I=[0.3,0.7]$,
if $\mathit{speed}$ is $\mathsf{full}$ then $I=[0.7,1.2]$;
\item the variables $\mathit{p_k}$, for $k \in 1 \dots 6$, are updated as expected to record the last six temperatures detected by sensor $\mathit{temp}$;
\item $\mathit{stress}$ remains unchanged if the temperature was below a constant $\mathsf{max}$, set to 100 in our experiments, for at least $3$ of the last $6$ time instants; otherwise, the variable is increased (reaching at most the value $1$) by a constant $\mathsf{stressincr}$ that we set at $0.02$;
\item all variables representing channels or actuators are not modified by $\E$, since they are under the control of the program.
\end{enumerate}
Summarising, the dynamics of $\mathit{p_k}$ and $\mathit{stress}$ can be modelled via the following set of stochastic difference equations, with sampling time interval $\Delta\tau =1$:
\begin{equation}%
\label{eq:enginedynamics}
\begin{array}{rcl}
\mathit{p_k}(\tau + 1)
& = &
\begin{cases}
\begin{array}{cl}
\mathit{temp}(\tau)
&
\text{ if } k=1
\\
\mathit{p_{k-1}}(\tau) &
\text{ if }
k=2,\dots,6
\end{array}
\end{cases}
%%%
\\
%%%
\mathit{stress}(\tau + 1)
& = &
\begin{cases}
\begin{array}{cl}
\mathit{stress}(\tau) + \mathsf{stressincr}
&
\text{ if }
|\{k \mid \mathit{p_k}(\tau) \ge \mathsf{max}\} | > 3
\\
\mathit{stress}(\tau)
&
\text{ otherwise}
\end{array}
\end{cases}
\end{array}
\end{equation}
and $\mathit{temp}$ varies by a value that is uniformly distributed in an interval depending on the state of actuators $\mathit{cool}$ and $\mathit{speed}$:
\begin{equation}%
\label{eq:temp}
\begin{array}{rcl}
\mathit{temp}(\tau + 1)
& = &
\mathit{temp}(\tau ) + v
\\
v & \sim &
\begin{cases}
\begin{array}{cl}
{\mathcal U}[-1.2,-0.8]
&
\text{ if }
\mathit{cool}(\tau) = \mathsf{on}
\\
{\mathcal U}[0.1,0.3]
&
\text{ if }
\mathit{cool}(\tau) = \mathsf{off}
\text{ and }
\mathit{speed}(\tau) = \mathsf{slow}
\\
{\mathcal U}[0.3,0.7]
&
\text{ if }
\mathit{cool}(\tau) = \mathsf{off}
\text{ and }
\mathit{speed}(\tau) = \mathsf{half}
\\
{\mathcal U}[0.7,1.2]
&
\text{ if }
\mathit{cool}(\tau) = \mathsf{off}
\text{ and }
\mathit{speed}(\tau) = \mathsf{full}.
\end{array}
\end{cases}
\end{array}
\end{equation}

The whole engine, the controller and the IDS are modelled by the following processes $\mathsf{Eng}$, $\mathsf{Ctrl}$ and $\mathsf{IDS}$, respectively:
\[
\begin{array}{lll}
\mathsf{Eng}
&
\stackrel{\mathit{def}}{=}
&
\mathsf{Ctrl} \cmerge \mathsf{IDS}
\\
\mathsf{Ctrl}
&
\stackrel{\mathit{def}}{=}
&
\mathrm{if}\;
[\mathit{ch\_temp} < \mathsf{threshold}]\;
\surd.
\mathsf{Check}
\;
\mathrm{else}
\;
\wact{\mathsf{on}}{\mathit{cool}}.
\mathsf{Cooling}
\\
\mathsf{Cooling}
&
\stackrel{\mathit{def}}{=}
&
\surd . \surd . \surd . \surd .
\mathsf{Check}
\\
\mathsf{Check}
&
\stackrel{\mathit{def}}{=}
&
\mathrm{if}\;
[\mathit{ch\_speed} = \mathsf{slow}]\;
(\wact{\mathsf{slow}}{\mathit{speed}},
\wact{\mathsf{off}}{\mathit{cool}}).
\mathsf{Ctrl} \;
\\ & &
\mathrm{else}\;
(\wact{\mathit{ch\_in}}{\mathit{speed}},
\wact{\mathsf{off}}{\mathit{cool}}).
\mathsf{Ctrl} \;
\\
\mathsf{IDS}
&
\stackrel{\mathit{def}}{=}
&
\mathrm{if}\;
[\mathit{temp} > \mathsf{max} \wedge \mathsf{cool} = \mathsf{off}]\;
\\ & &
(\wact{\mathsf{hot}}{\mathit{ch\_warning}},
\wact{\mathsf{low}}{\mathit{ch\_speed}},
\wact{\mathsf{full}}{\mathit{ch\_out}}).
\mathsf{IDS}
\\ & &
\mathrm{else}\;
(\wact{\mathsf{ok}}{\mathit{ch\_warning}},
\wact{\mathsf{half}}{\mathit{ch\_speed}},
\wact{\mathsf{half}}{\mathit{ch\_out}}).
\mathsf{IDS}
\end{array}
\]
where we write $((e_1 \to x_1),\dots,(e_n \to x_n))$ to denote the prefix $(e_1\dots e_n \to x_1\dots x_n)$.

At each scan cycle $\mathsf{Ctrl}$ checks if the value of temperature as carried by channel $\mathit{ch\_temp}$ is too high, namely above $\mathsf{threshold}$, which is a constant $\le \mathsf{max}$ and set to 99.8 in our experiments.
If this is the case, then $\mathsf{Ctrl}$ activates the coolant for $5$ consecutive time instants, otherwise the coolant remains off.
In both cases, before re-starting  its scan cycle, $\mathsf{Ctrl}$ waits for instructions/requests to change the speed: if it receives instructions through channel $\mathit{ch\_speed}$ from the IDS to slow down the engine  then it  commands so through actuator $\mathit{speed}$. Otherwise, if the IDS of the other engine requests through channel $\mathit{ch\_in}$ to work at $\mathsf{full}$ of $\mathsf{half}$ power, then $\mathsf{Ctrl}$ sets $\mathit{speed}$ accordingly.

The process $\mathsf{IDS}$ checks whether the cooling system is  active when the temperature is above $\mathsf{max}$.
If this safety condition is violated then:
\begin{enumerate*}[(i)]
\item it raises the presence of an anomaly, via channel $\mathit{ch\_warning}$;
\item it commands $\mathsf{Ctrl}$ to slow down the engine, via  $\mathit{ch\_speed}$;
\item it requests to the other engine to run at full power, via channel $\mathit{\mathit{ch}\_{\mathit{out}}}$.
\end{enumerate*}
Otherwise, $\mathsf{IDS}$ asks (both local and external) controllers to set their engines at $\mathsf{half}$ power.

An engine system can be built by combining two engines, the ``left'' and the ``right'' one, interacting with a supervisory component that checks their correct functioning.
The composite system can be defined as
\begin{displaymath}
\mathsf{Eng\_Sys} \; = \; \mathsf{Eng}\_{\mathsf{L}} \cmerge \mathsf{Eng}\_{\mathsf{R}} \cmerge \mathsf{SV}
\end{displaymath}
\noindent
where processes $\mathsf{Eng}\_{\mathsf{L}}$ and $\mathsf{Eng}\_{\mathsf{R}}$ are obtained from $\mathsf{Eng}$ by renaming the variables as follows:
\begin{enumerate*}[(i)]
\item both $\mathit{ch\_in}$ in $\mathsf{Eng\_L}$ and $\mathit{ch\_out}$ in $\mathsf{Eng\_R}$ are renamed as $\mathit{ch\_speed\_R\_to\_L}$;
\item both $\mathit{ch\_in}$ in $\mathsf{Eng\_R}$ and $\mathit{ch\_out}$ in $\mathsf{Eng\_L}$ are renamed as $\mathit{ch\_speed\_L\_to\_R}$;
\item all remaining variables $\mathit{x}$ used in $\mathit{Eng}$ are renamed as $\mathit{x\_L}$ in $\mathsf{Eng\_L}$ and as $\mathit{x\_R}$ in $\mathsf{Eng\_R}$.
\end{enumerate*}
Then, we introduce a new channel $\mathit{ch\_alarm}$ which is used by $\mathsf{SV}$ to forward to the external environment
the information obtained from the IDSs through $\mathit{ch\_warning\_L}$ and $\mathit{ch\_warning\_R}$.
In detail, $\mathit{ch\_alarm}$ carries a value in the set $\{\mathsf{none}, \mathsf{left}, \mathsf{right}, \mathsf{both}\}$ depending on the anomaly warnings received from the left and right IDS:\@
\[
\begin{array}{lll}
\mathsf{SV}
&
\stackrel{\mathit{def}}{=}
&
\mathrm{if}\;
[\mathit{ch\_warning\_L} = \mathsf{hot} \wedge \mathit{ch\_warning\_R} = \mathsf{hot}]\;
\wact{\mathsf{both}}{\mathit{ch\_alarm}}. \mathsf{SV}
\\ & &
\mathrm{else}\;
\mathrm{if}\;
[\mathit{ch\_warning\_L} = \mathsf{hot}]\;
\wact{\mathsf{left}}{\mathit{ch\_alarm}}. \mathsf{SV}
\\ &&
\; \; \; \; \; \; \; \mathrm{else}\;
\mathrm{if}\;
[\mathit{ch\_warning\_R} = \mathsf{hot}]\;
\wact{\mathsf{right}}{\mathit{ch\_alarm}}. \mathsf{SV} \;
\\ &&
\; \; \; \; \; \; \; \; \; \; \; \; \; \;  \mathrm{else}\;
\wact{\mathsf{none}}{\mathit{ch\_alarm}}. \mathsf{SV}\;
\end{array}
\]
Following~\cite{LMMT21}, cyber-physical attacks can be modelled as processes that run in parallel with systems under attack.
Intuitively, they represent malicious code that has been injected in the logic of the system.
Three examples of attack to $\mathsf{Eng\_Sys}$ are the following, where $\mathsf{X} \in \{\mathsf{L},\mathsf{R}\}$:
\[
\begin{array}{lll}
\mathsf{Att\_Act\_X}
&
\stackrel{\mathit{def}}{=}
&
\mathrm{if}\;
[\mathit{temp\_X} < \mathsf{max} - \mathsf{AC}]\;
\wact{\mathsf{off}}{\mathit{cool\_X}}. \mathsf{Att\_Act\_X}\;
\mathrm{else}\;
\surd. \mathsf{Att\_Act\_X}
\\
\mathsf{Att\_Sen\_X}
&
\stackrel{\mathit{def}}{=}
&
\wact{-\mathsf{TF}}{\mathit{temp\_fake\_X}}. \mathsf{Att\_Sen\_X}
\\
\mathsf{Att\_Saw\_X}
&
\stackrel{\mathit{def}}{=}
&
(\wact{\mathit{left}}{\mathit{AW\_l}},
\wact{\mathit{right}}{\mathit{AW\_r}}).
\mathsf{Att\_Saw\_X'}
\\
\mathsf{Att\_Saw\_X'}
&
\stackrel{\mathit{def}}{=}
&
\mathrm{if}\;
[\mathit{cs} \in [\mathit{AW_l},\mathit{AW}_r]]\;
\wact{-\mathsf{TF}}{\mathit{temp\_fake\_X}}. \mathsf{Att\_Saw\_X'}\;
\\
& &
\mathrm{else}\;
\wact{0}{\mathit{temp\_fake\_X}}. \mathsf{Att\_Saw\_X'}
\end{array}
\]
Process $\mathsf{Att\_Act\_X}$ models an integrity attack on actuator $\mathit{cool\_X}$.
The five-instants cooling cycle started by $\mathsf{Ctrl\_X}$ is maliciously interrupted as soon as $\mathit{temp\_X}$ goes below $\mathsf{max} - \mathsf{AC}$, with $\mathsf{AC}$ a constant with value in $[1.0,3.0]$, in order to let the temperature of the engine rise quickly above $\mathsf{max}$, accumulating stress in the system, and requiring continuous activation of the cooling system.
We recall that the quantification of stress accumulated by engine $\mathit{Eng\_X}$ is recorded in variable $\mathit{stress\_X}$, which is incremented by $\E$ when $|\{k \mid 1 \le k \le 6 \wedge p_k \ge \mathsf{max}\} | > 3$.
In~\cite{LMMT21} this attack is classified as \emph{stealth}, since the condition $\mathit{temp\_X} > \mathsf{max} \wedge \mathit{cool\_X} = \mathsf{off}$ monitored by the IDS never holds.
Notice that $\mathsf{AC}$ is a parameter of the attack.

Then, $\mathsf{Att\_Sen\_X}$ is an integrity attack to sensor $\mathit{temp\_X}$.
The attack adds a negative offset $-\mathsf{TF}$, with $\mathsf{TF}$ a positive constant, to the temperature value carried by $\mathit{ch\_temp\_X}$, namely $ch\_temp\_X$ becomes $\mathit{temp\_X} - \mathsf{TF}$.
This attack aims to prevent some activation of the cooling system by $\mathsf{Ctrl\_X}$.
Since the IDS raises a warning on $\mathit{ch\_warning\_X}$ each time we have $\mathsf{max} < \mathit{temp\_X} \leq \mathsf{max} + \mathsf{TF}$ and $\mathit{cool\_X}=\mathsf{off}$, this attack is not stealth.
In this case, $\mathsf{TF}$ is a parameter of the attack.

Finally, $\mathsf{Att\_Saw\_X}$  performs the same attack of $\mathsf{Att\_Sen\_X}$ but only in a precise attack window.
We assume that $\mathit{cs}$ is a variable that simply counts the number of computation steps (the initial value is $0$ and $\mathit{cs}(\tau+1) = \mathit{cs}(\tau)+1$ is added to Equation~\ref{eq:enginedynamics}).
Then, $[\mathit{AW_l},\mathit{AW_r}]$ is an interval, where $\mathit{AW_l}=0=\mathit{AW_r}$ in the initial data state and we assume that the variables $\mathit{left}$ and $\mathit{right}$ take random non-negative integer values satisfying the property $0 \le \mathit{right} - \mathit{left} \le \mathsf{AWML}$, where $\mathsf{AWML}$ is a parameter of the attack representing the maximal length of the attack window.

%===================================================
%======================================================

\subsection{Simulation of the engine system}

We have applied our simulation strategy to the engine system.
Figure~\ref{fig:average_temp} reports the results of six simulations of the single engine $\mathsf{Eng\_L}$, where each simulation is obtained by simulating $100$ runs consisting in $10000$ steps, with the following initial \datastate{} $\ds_0$:
\begin{enumerate*}[(i)]
\item $\mathit{temp\_L}=\mathit{p}_1\_L= \cdots = \mathit{p}_6\_L = 95$;
\item $\mathit{cool\_L} = \mathsf{off}$;
\item $\mathit{speed\_L} = \mathsf{half}$;
\item $\mathit{stress\_L} = 0$;
\item $\mathit{ch\_speed\_L}$ = $\mathit{ch\_speed\_L\_to\_R}$= $\mathit{ch\_speed\_R\_to\_L}$= $\mathsf{half}$;
\item $\mathit{temp\_fake\_L} = 0$; and
\item $\mathit{ch\_warning\_L} = \mathsf{ok}$.
\end{enumerate*}
We considered three different scenarios:
\begin{enumerate*}[(i)]
\item no attack, modelled by processes $\mathsf{Eng\_L}$,
\item attack on actuator $\mathit{cool\_L}$, modelled by $\mathsf{Eng\_L} \cmerge \mathsf{Att\_Act\_L}$, with constant $\mathsf{AC}$ set to $1.8$, and
\item attack on sensor $\mathit{temp\_L}$, modelled by $\mathsf{Eng\_L} \cmerge \mathsf{Att\_Sen\_L}$, with constant $\mathsf{TF}$ set to $0.4$.
\end{enumerate*}
The six pictures report the value of average temperature detected by sensor $\mathit{temp\_L}$ and the average stress level carried by variable $\mathit{stress\_L}$ in these three scenarios.

Clearly, by comparing the plots in Figure~\ref{fig:average_temp} it turns out that both attacks are able to cause an average size-up of the temperature of the engine and to take the engine to the maximal level of stress after 10000 steps.

\begin{figure}[tbp]
\centering
\includegraphics[width=0.32\textwidth]{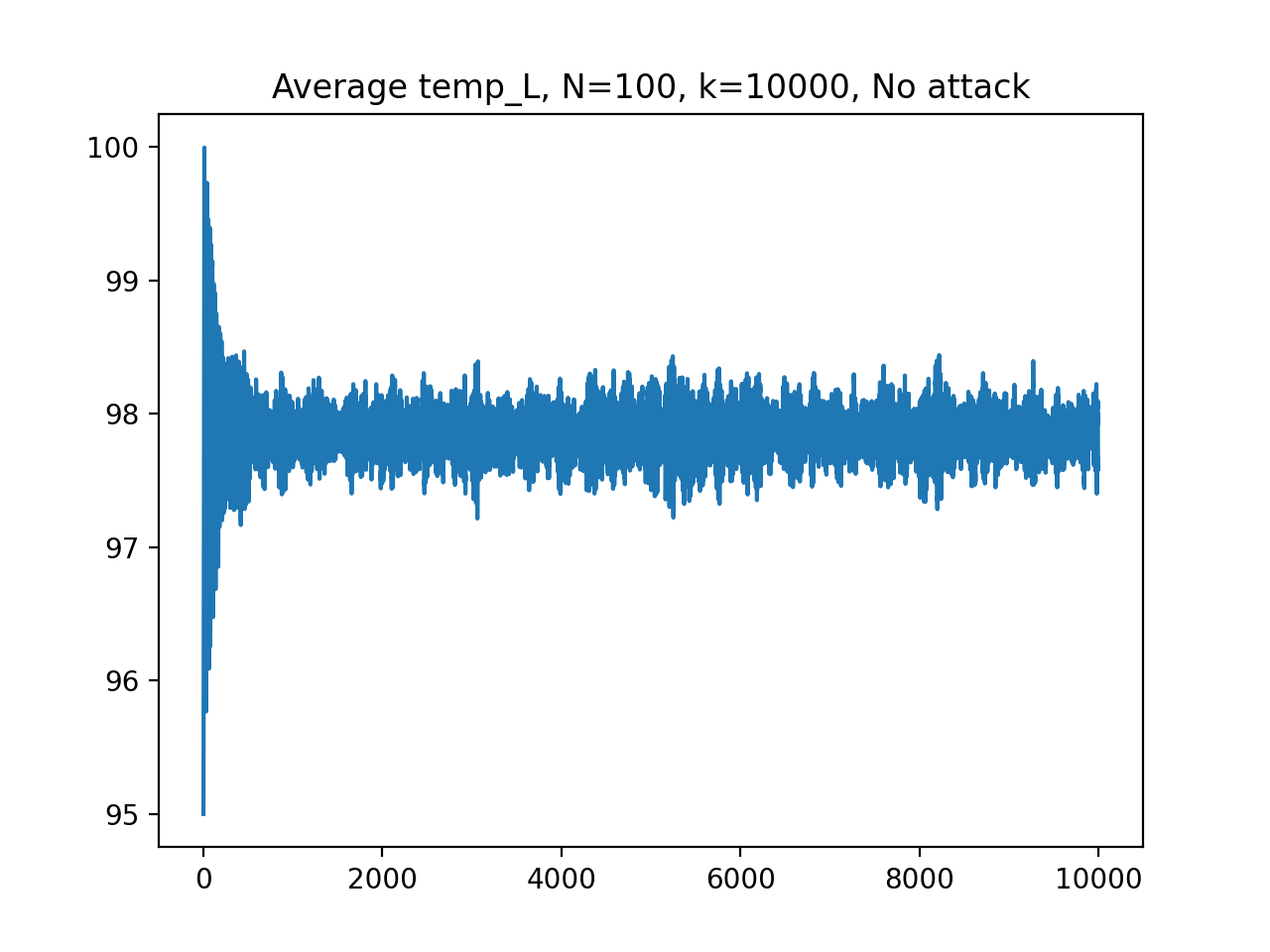}
\includegraphics[width=0.32\textwidth]{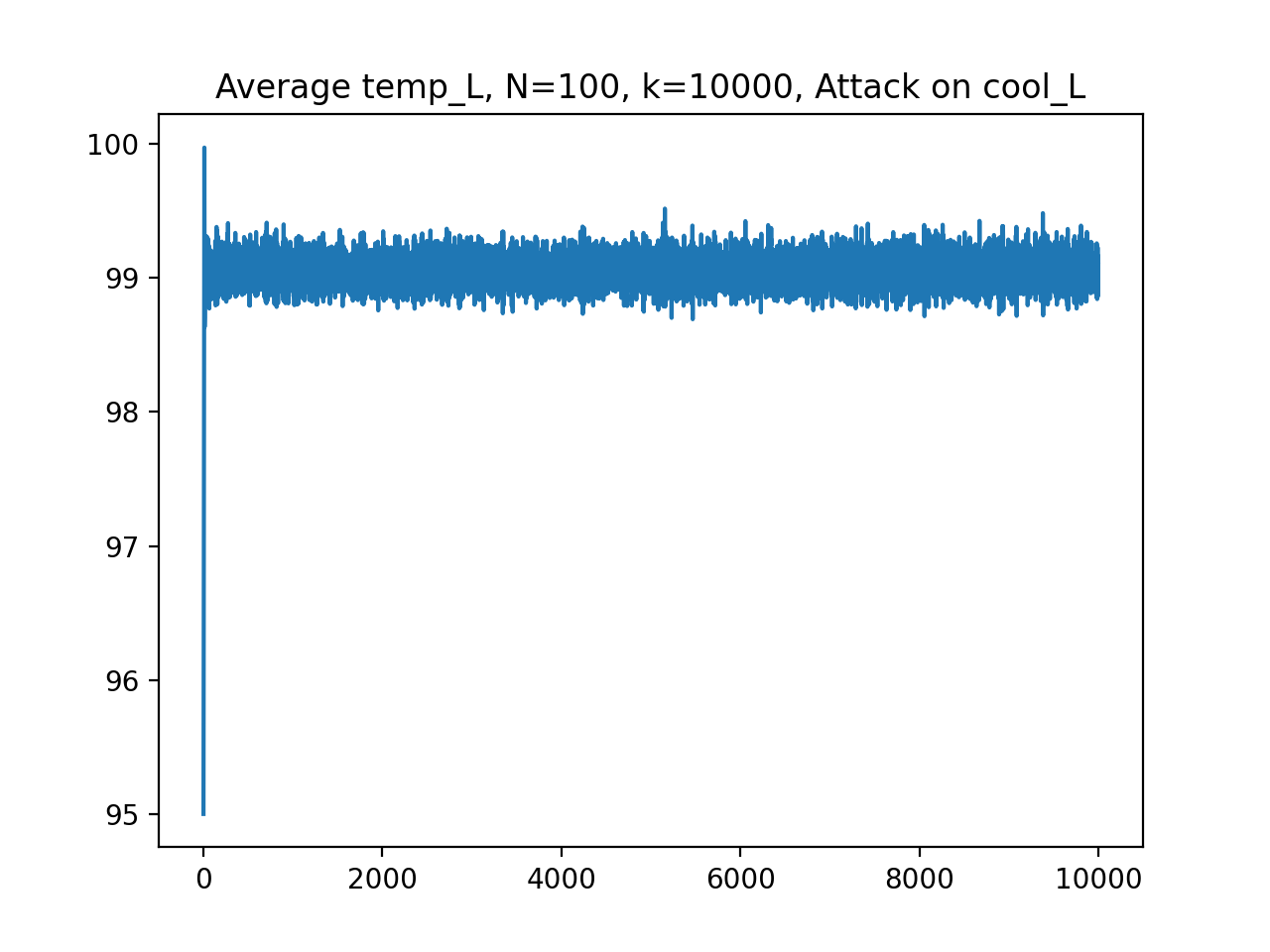}
\includegraphics[width=0.32\textwidth]{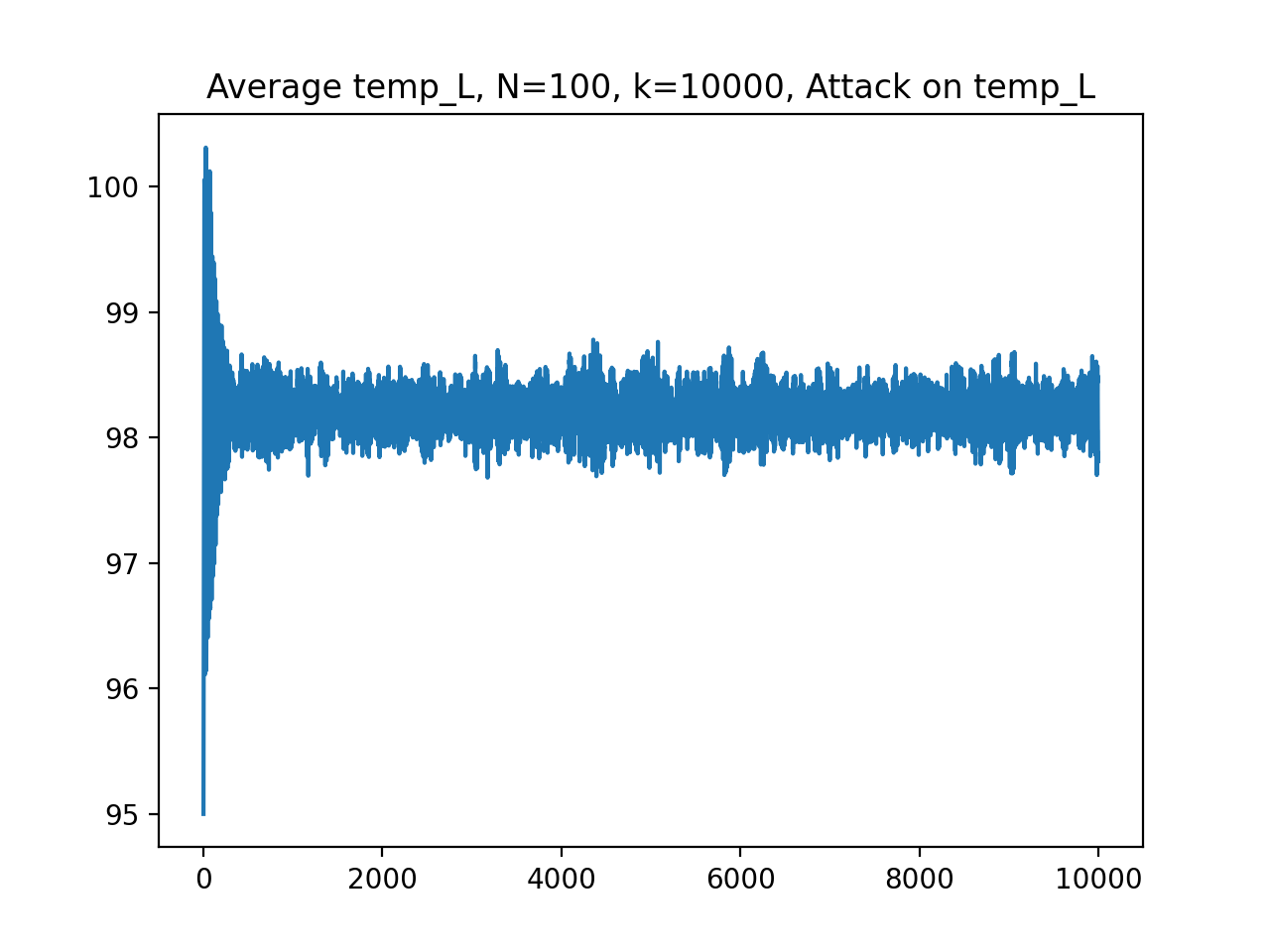}
\includegraphics[width=0.32\textwidth]{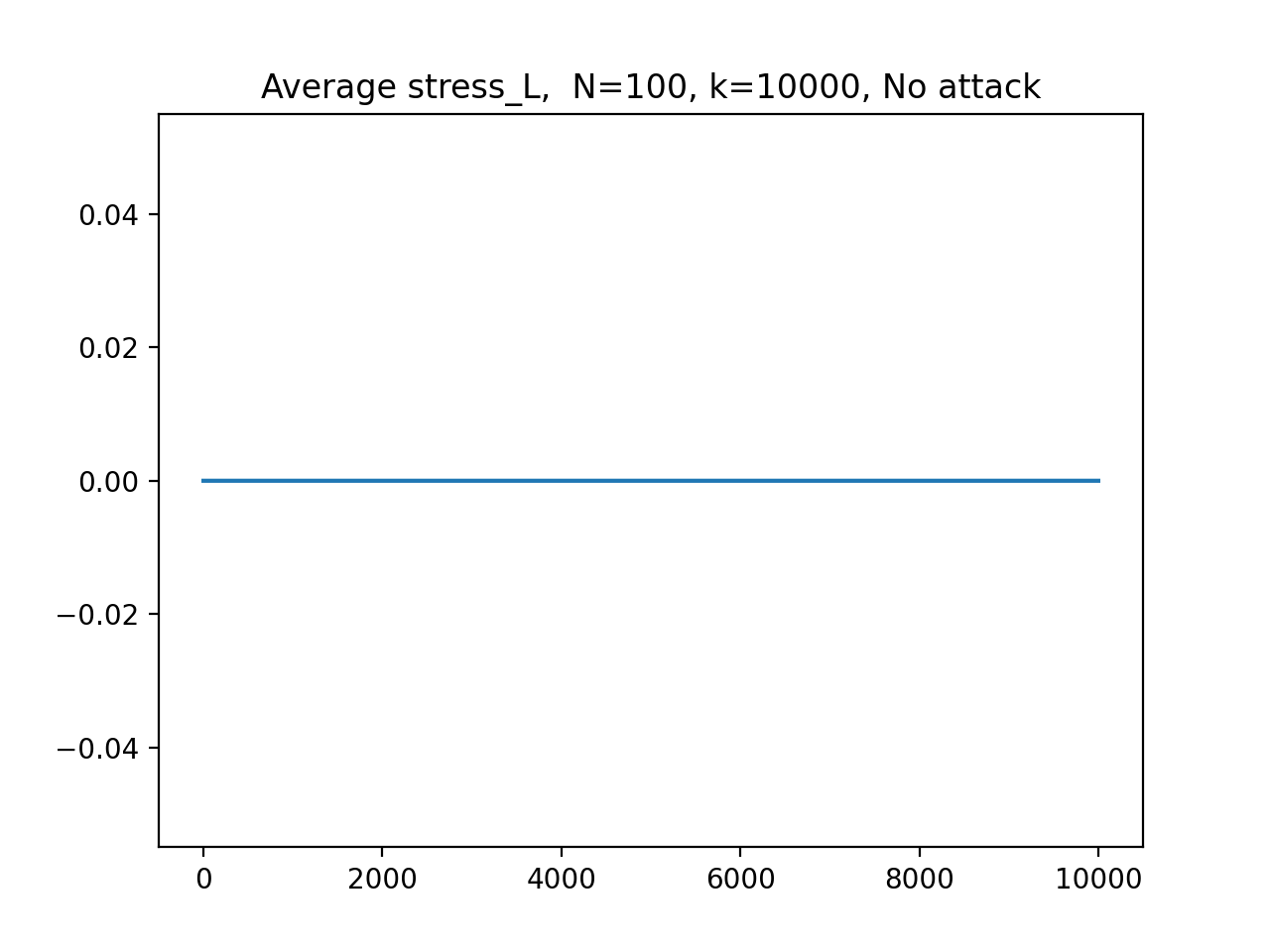}
\includegraphics[width=0.32\textwidth]{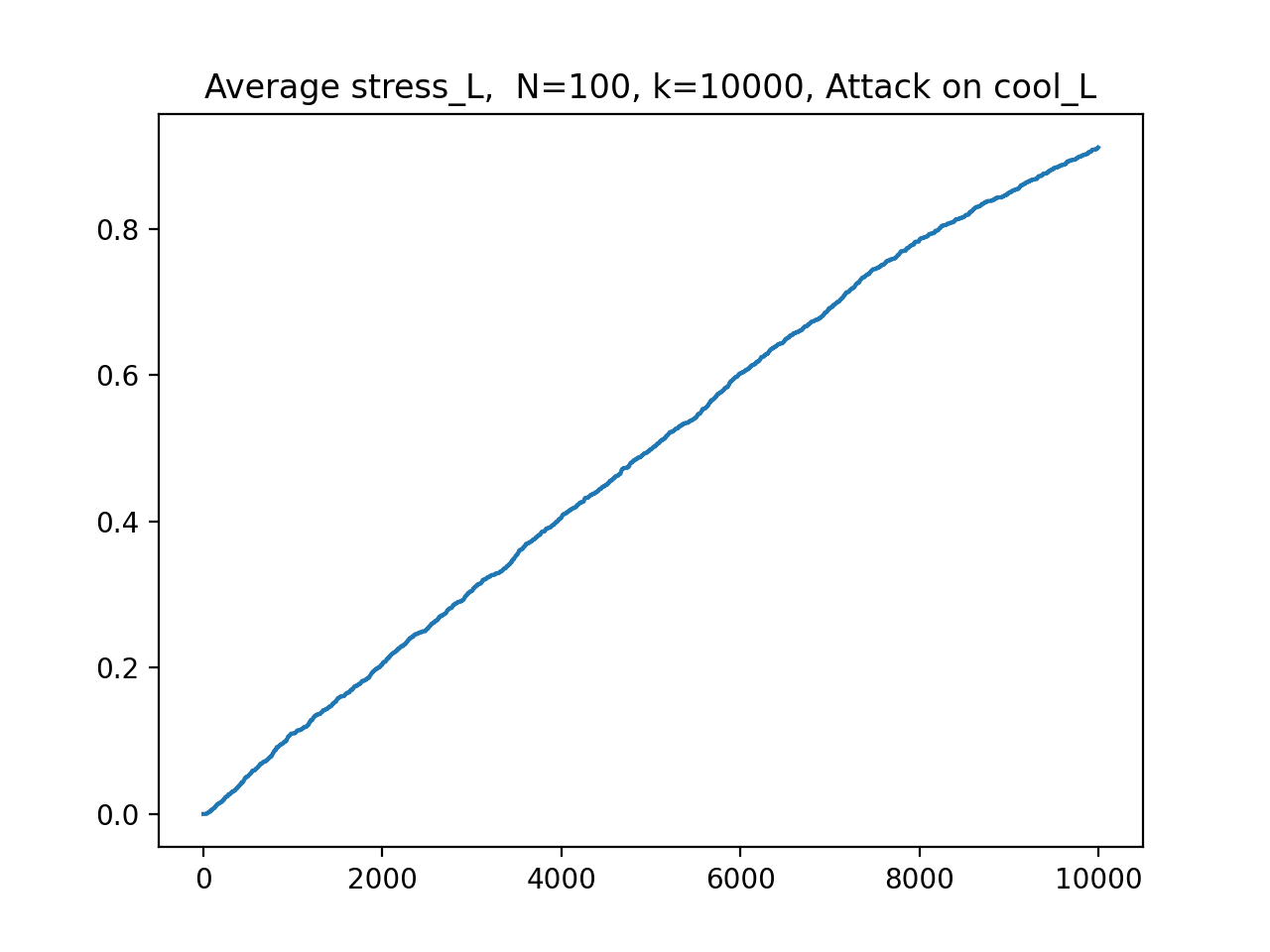}
\includegraphics[width=0.32\textwidth]{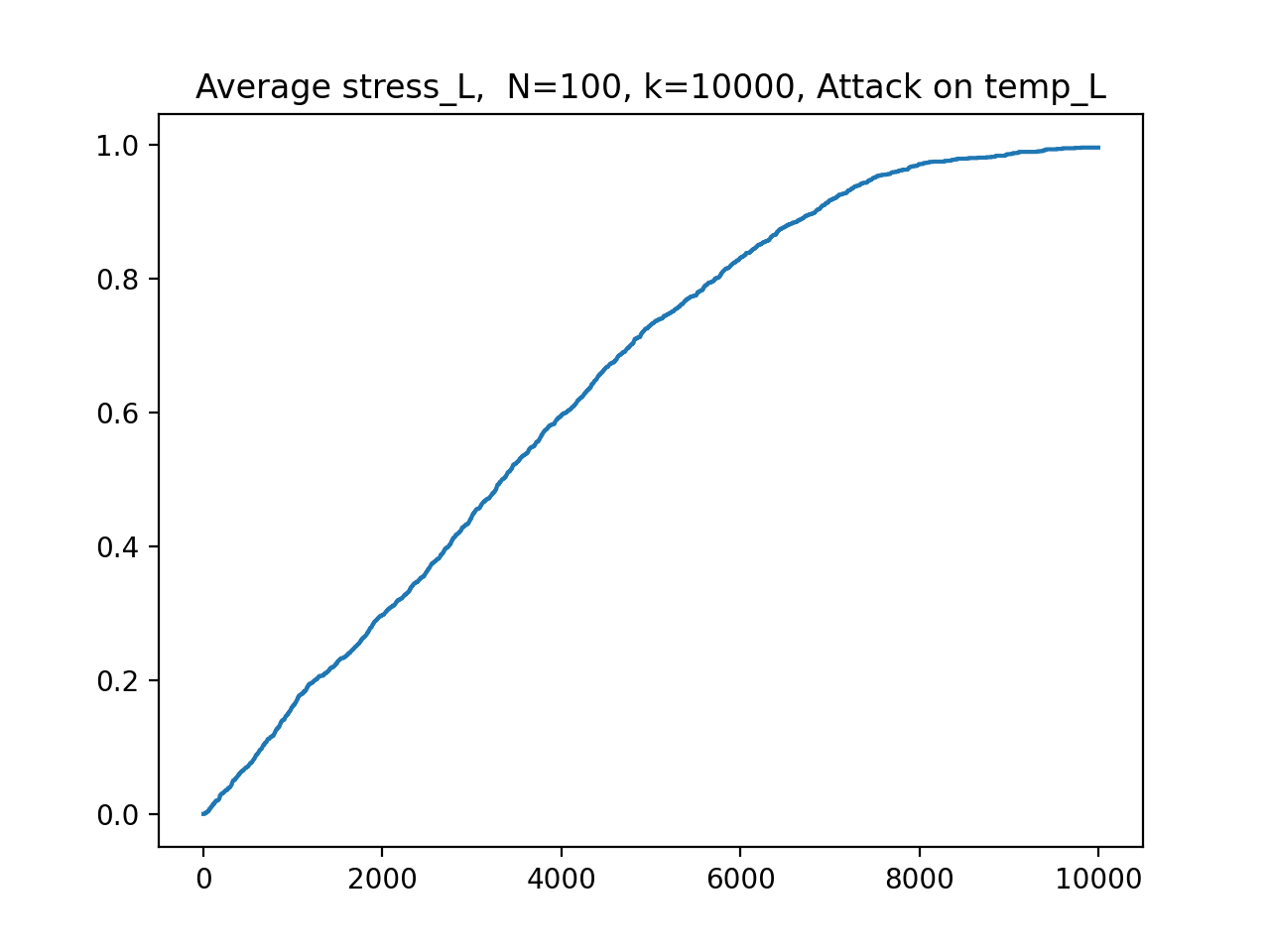}
\caption{Average temperature (sensor $\mathit{temp\_L}$) and stress (variable $\mathit{stress\_L}$) for left engine running in three scenarios: no attack, attack on actuator $\mathit{cool\_L}$, attack on sensor $\mathit{temp\_L}$.\\
Simulation settings: $k=10000$ steps, $N=100$ runs.\\
Initial sensor/actuator settings:
$\mathit{temp\_L}=95$,
$\mathit{cool\_L}=\mathsf{off}$,
$\mathit{speed\_L}=\mathsf{half}$.\\
Values assigned to constants:
$\mathsf{AC}=1.8$,
$\mathsf{TF}=0.4$.}%
\label{fig:average_temp}
\end{figure}

%==================================================
%=================================================

\subsection{Analysis of the engine system}

Our aim is to use evolution metrics in order to study the impact of the attacks.
This requires to compare the evolution sequences of the genuine system and those of the system under attack.
The two approaches to the quantification of impact of attacks outlined above, consisting in evaluating only damages caused by attacks or both damages and alert signals, will be realised by employing different penalty functions.

We start with quantifying the impact of cyber-physical attacks by focusing on the runtime values of \emph{false negatives} and \emph{false positives}.
Intuitively, if we focus on a single engine $\mathsf{Eng\_X}$, false negatives and false positives represent the average \emph{effectiveness}, and the average \emph{precision}, respectively, of the IDS modelled by $\mathsf{IDS\_X}$ to signal through channel $\mathit{ch\_warning\_X}$ that the engine is under stress.
In our setting, we can add two variables $\mathit{fn\_X}$ and $\mathit{fp\_X}$ that quantify false negatives and false positives depending on $\mathit{stress\_X}$ and
$\mathit{ch\_warning\_X}$.
Both these variables are updated by $\E$.
In detail, $\E$ acts on $\mathit{fn\_X}$ and $\mathit{fp\_X}$ as follows:
\[
\begin{array}{lll}
\mathit{fn\_X}(\tau+1) & ={} & \frac{\tau * \mathit{fn\_X}(\tau) +
\max(0,\mathit{stress\_X(\tau)}-\mathit{ch\_warning\_X(\tau)})}{\tau+1}
\\[1.0 ex]
\\
\mathit{fp\_X}(\tau+1) & ={} & \frac{\tau * \mathit{fp\_X}(\tau) +
\max(0,\mathit{ch\_warning\_X(\tau)}-\mathit{stress\_X}(\tau))}{\tau+1}
\end{array}
\]
where $\mathit{ch\_warning\_X} \in \{\mathsf{hot},\mathsf{ok}\}$ is viewed as a real by setting $\mathsf{hot}=1$ and $\mathsf{ok}=0$.
In this way, as reported in~\cite{LMMT21}, $\mathit{fn\_X}$ and $\mathit{fp\_X}$ carry the average value, in probability and time, of the difference between the value of the stress of the system and of the warning raised by the IDS\@.

False negatives and false positives for the same simulations described in Figure~\ref{fig:average_temp} are reported in
Figure~\ref{fig:average_fp_fn_left}.

\begin{figure}[tbp]
\centering
\includegraphics[width=0.32\textwidth]{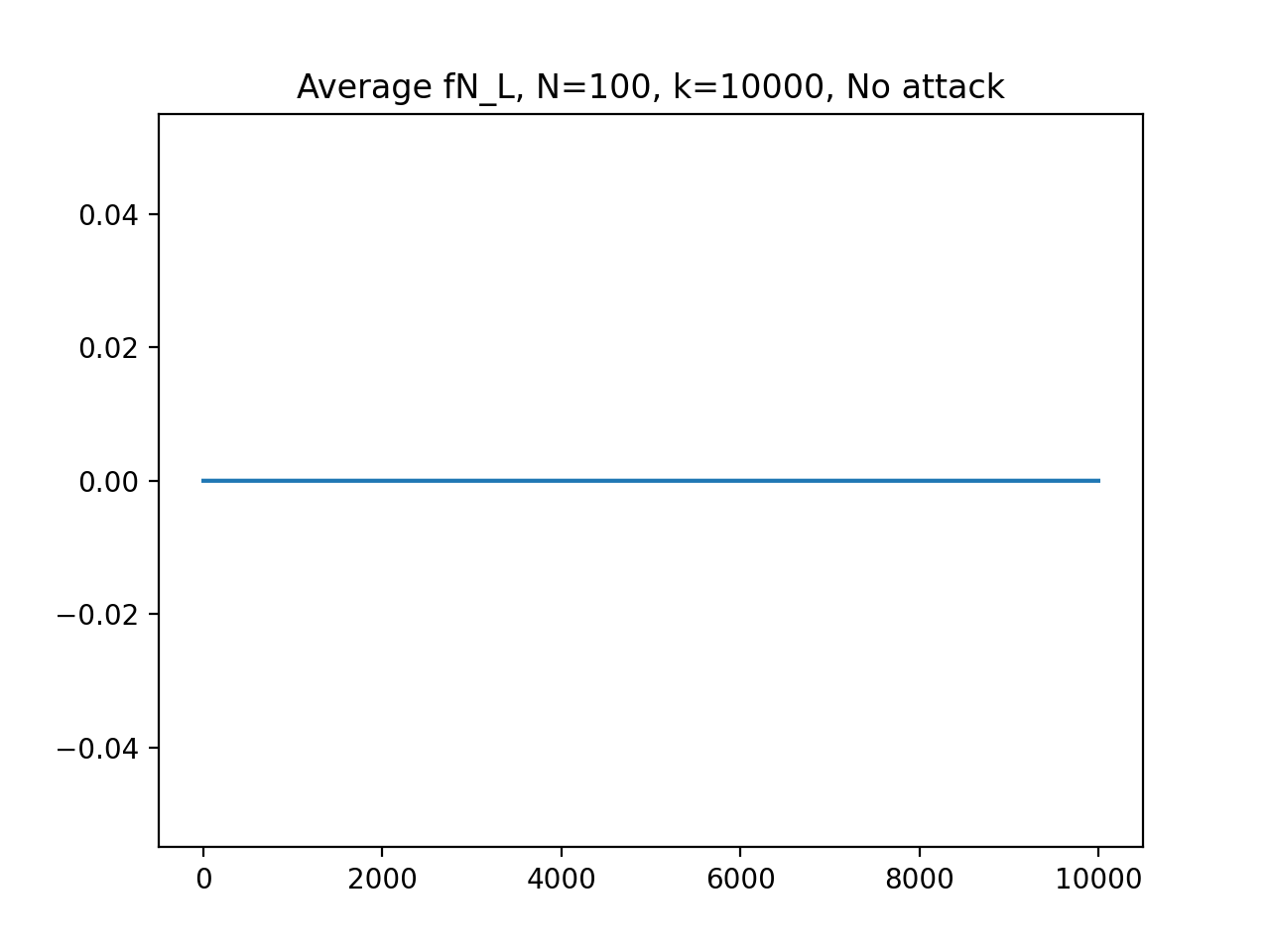}
\includegraphics[width=0.32\textwidth]{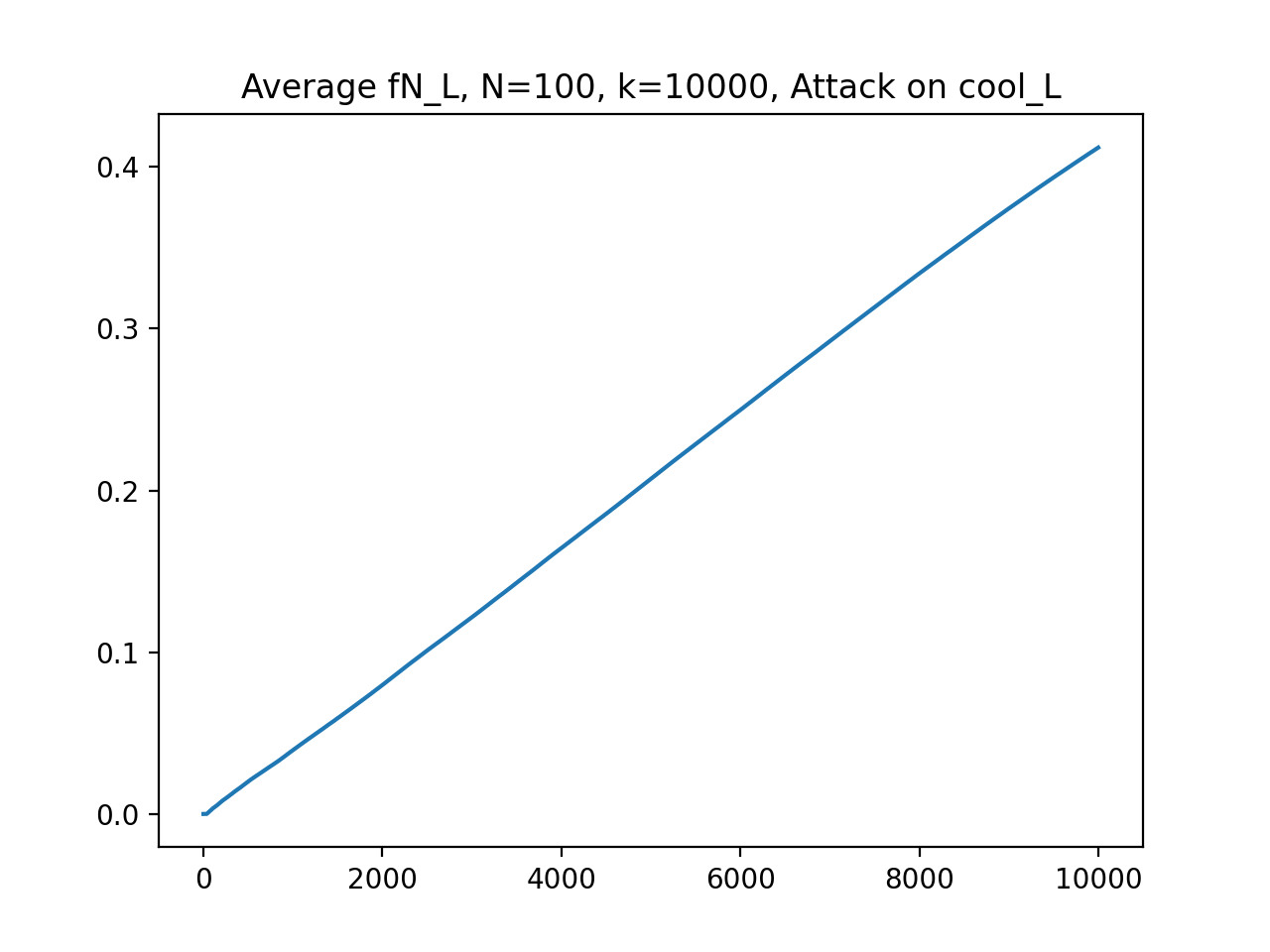}
\includegraphics[width=0.32\textwidth]{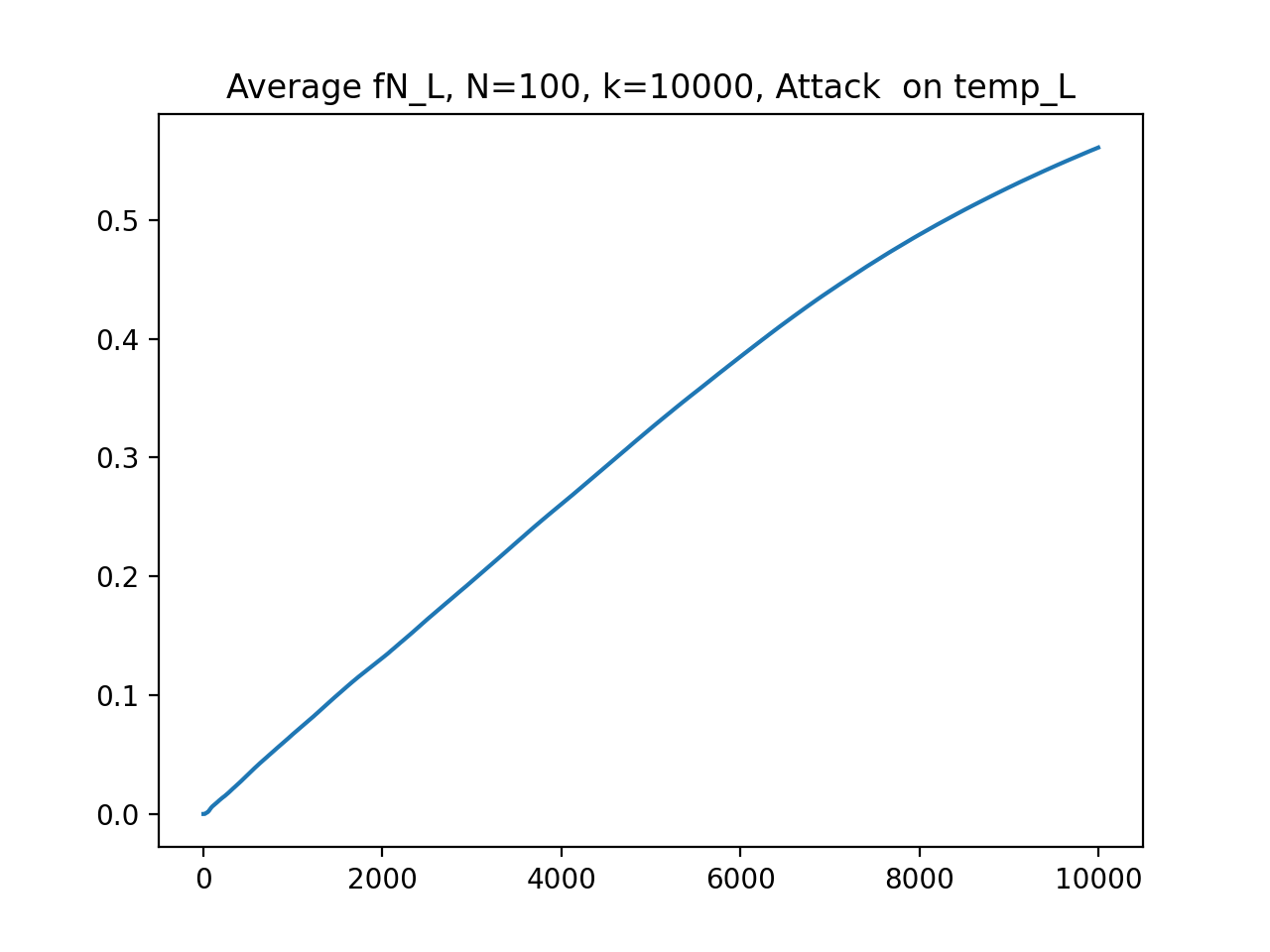}
\includegraphics[width=0.32\textwidth]{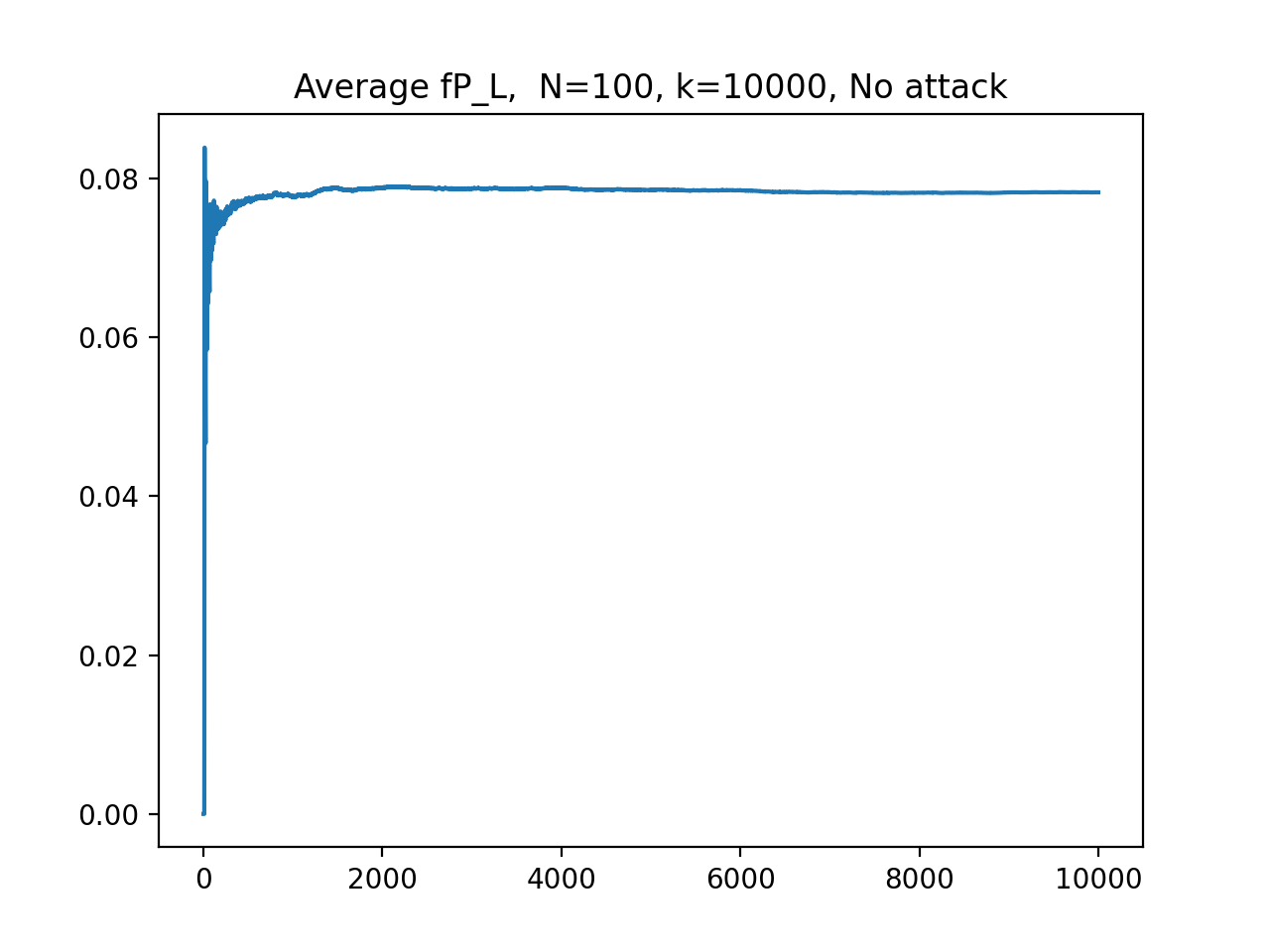}
\includegraphics[width=0.32\textwidth]{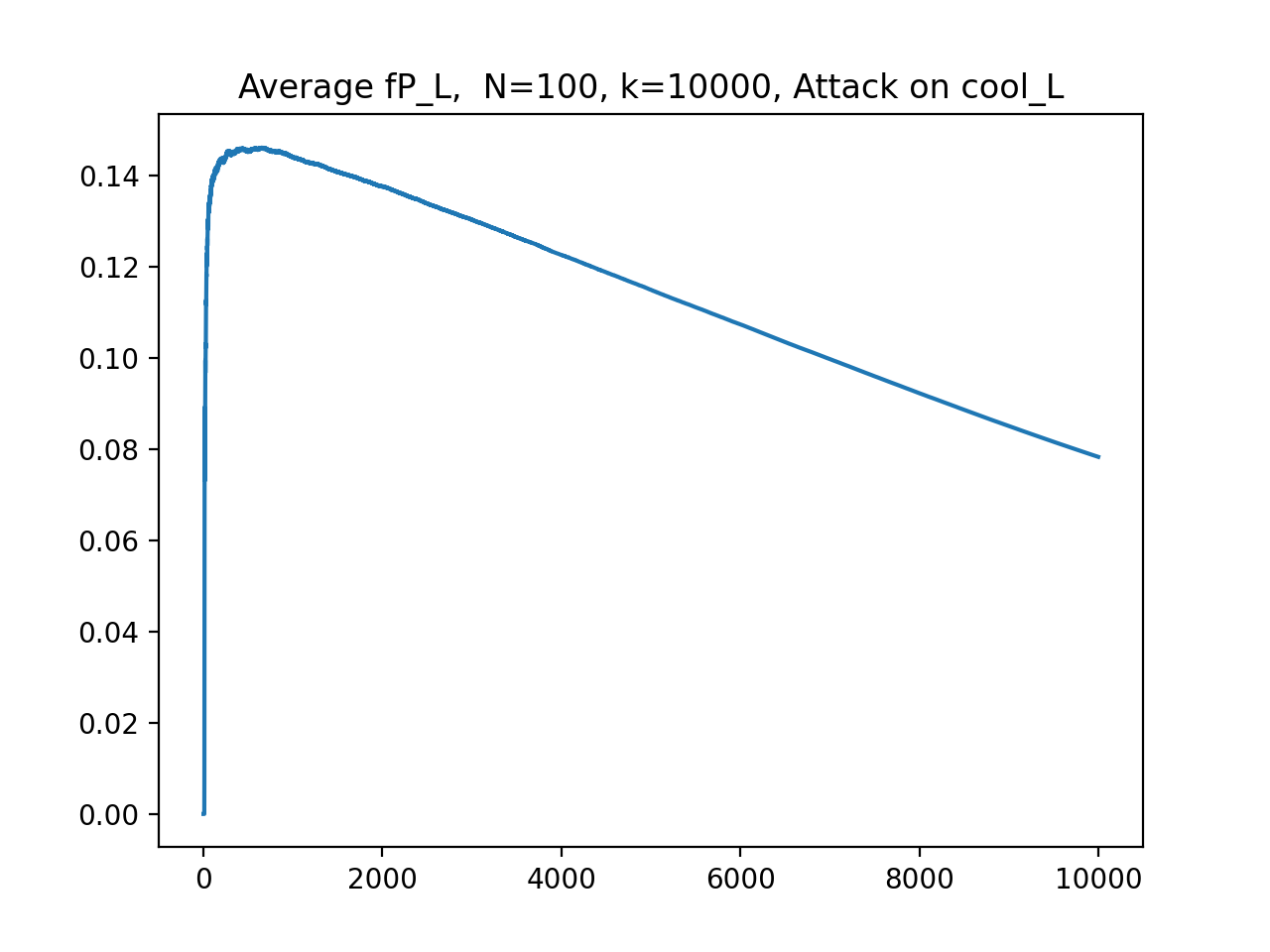}
\includegraphics[width=0.32\textwidth]{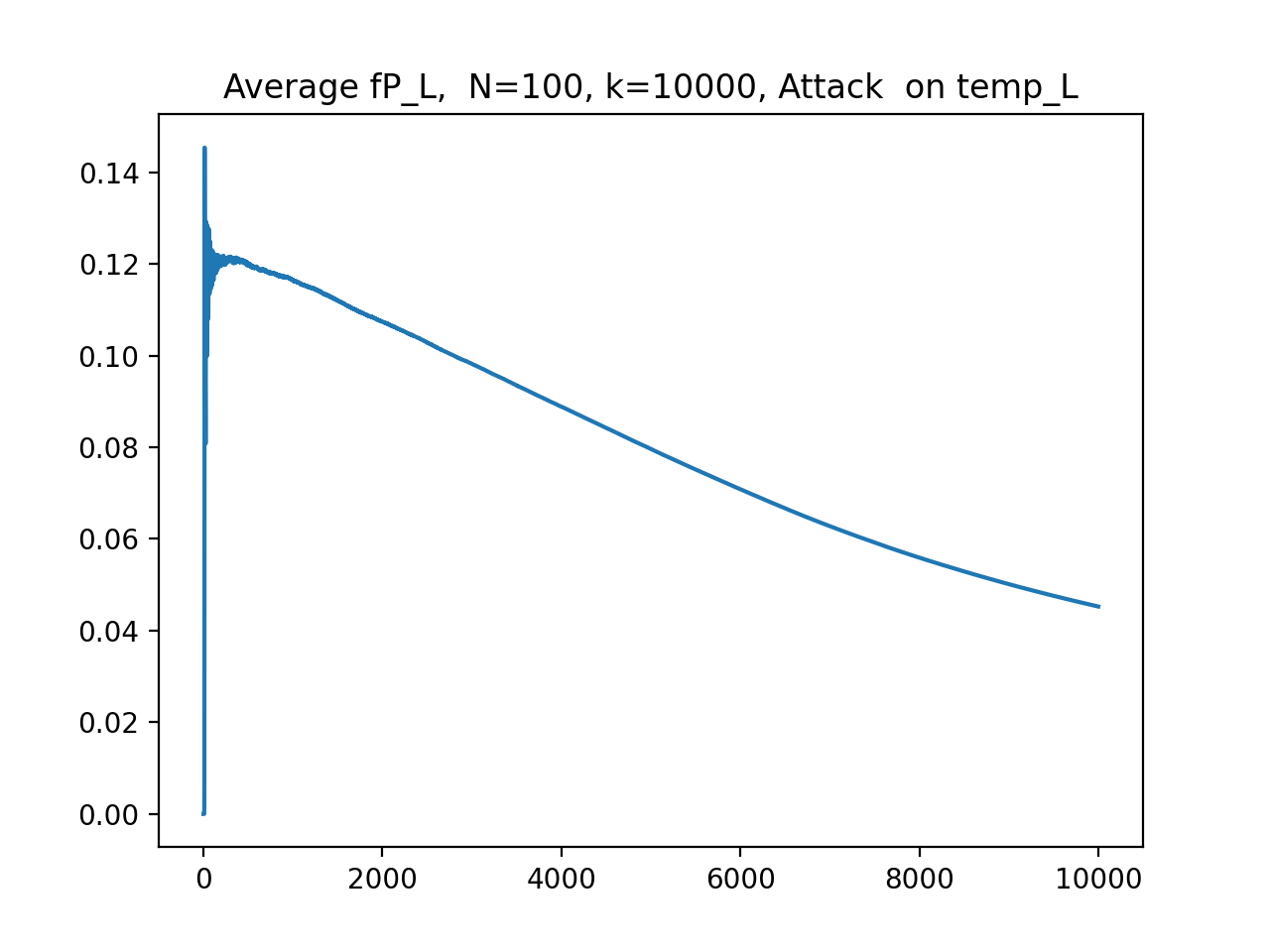}
\caption{Average false negatives (variable $\mathit{fn\_L}$) and false positives (variable $\mathit{fp\_L}$) for the simulations of Figure~\ref{fig:average_temp}.}%
\label{fig:average_fp_fn_left}
\end{figure}

Clearly, if we take the value of variable $\mathit{fn\_X}$  (resp. $\mathit{fp\_X}$) as the value returned by our penalty function, we can quantify how worse $\mathit{Eng\_X}$ under attack behaves with respect to the genuine version of $\mathit{Eng\_X}$, where the evaluation of these systems is done by taking into account their false negatives (resp.\ false positives).
Notice that other approaches are possible.
For instance, by defining a penalty function returning a convex combination of the values of $\mathit{fn\_X}$ and $\mathit{fp\_X}$, systems are evaluated by suitable weighting their false negatives and false positives.

In Figure~\ref{fig:distance_attack_left} we report the distances between the genuine system $\mathsf{Eng\_L}$ and the compromised systems $\mathsf{Eng\_L} \cmerge \mathsf{Att\_Act\_L}$ and $\mathsf{Eng\_L} \cmerge \mathsf{Att\_Sen\_L}$ in terms of false negatives  and false positives.
The settings for these simulations are the same of those used for the simulations described in Figure~\ref{fig:average_temp}.

\begin{figure}[tbp]
\centering
\begin{subfigure}{0.42\textwidth}
\includegraphics[width=1.00\textwidth]{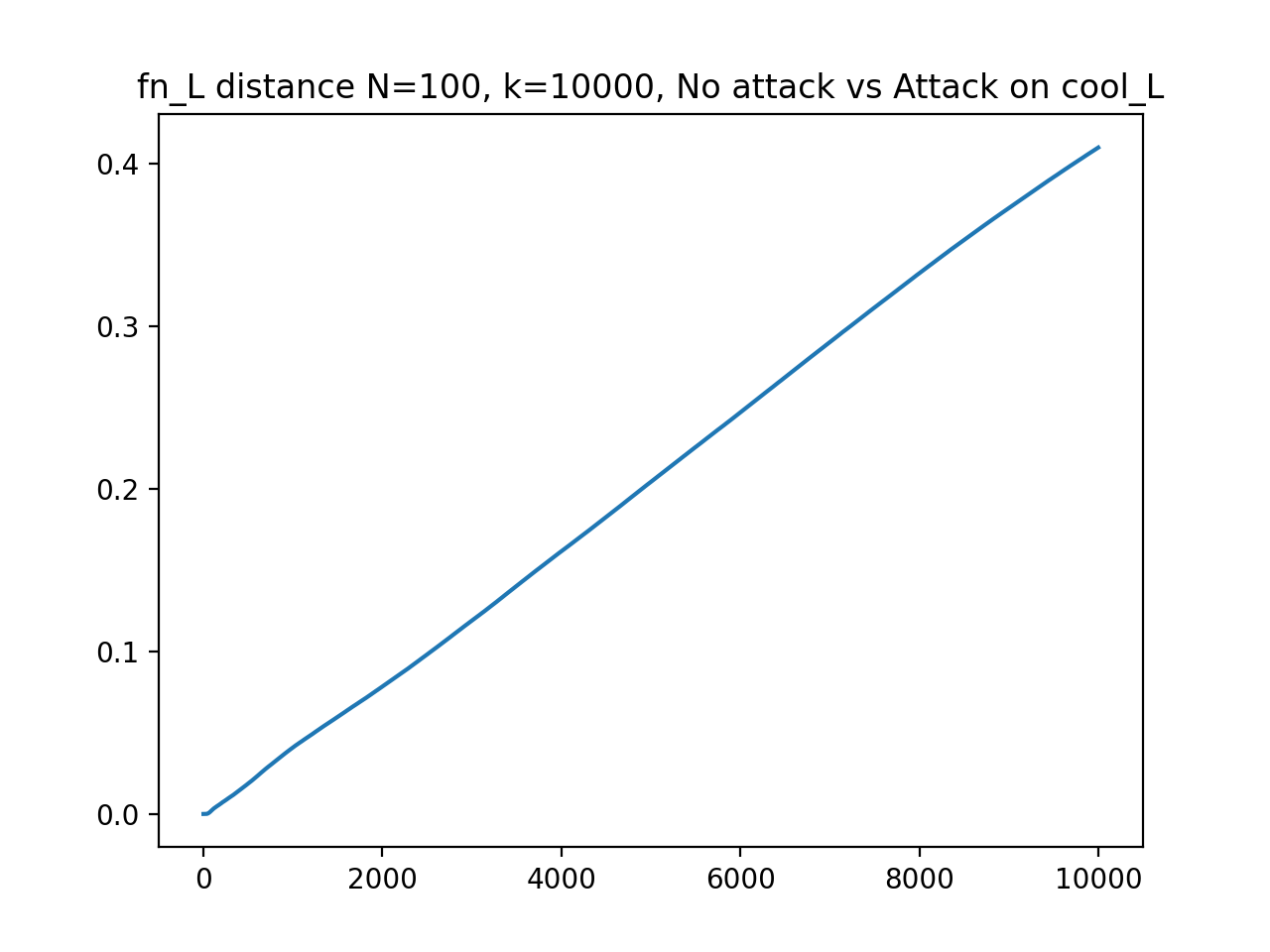}
\caption{$\mathit{fn\_L}$ distance, attack on $\mathsf{cool\_L}$}%
\label{fig_fn_fp_engL_01}
\end{subfigure}
\hfill
\begin{subfigure}{0.42\textwidth}
\includegraphics[width=1.00\textwidth]{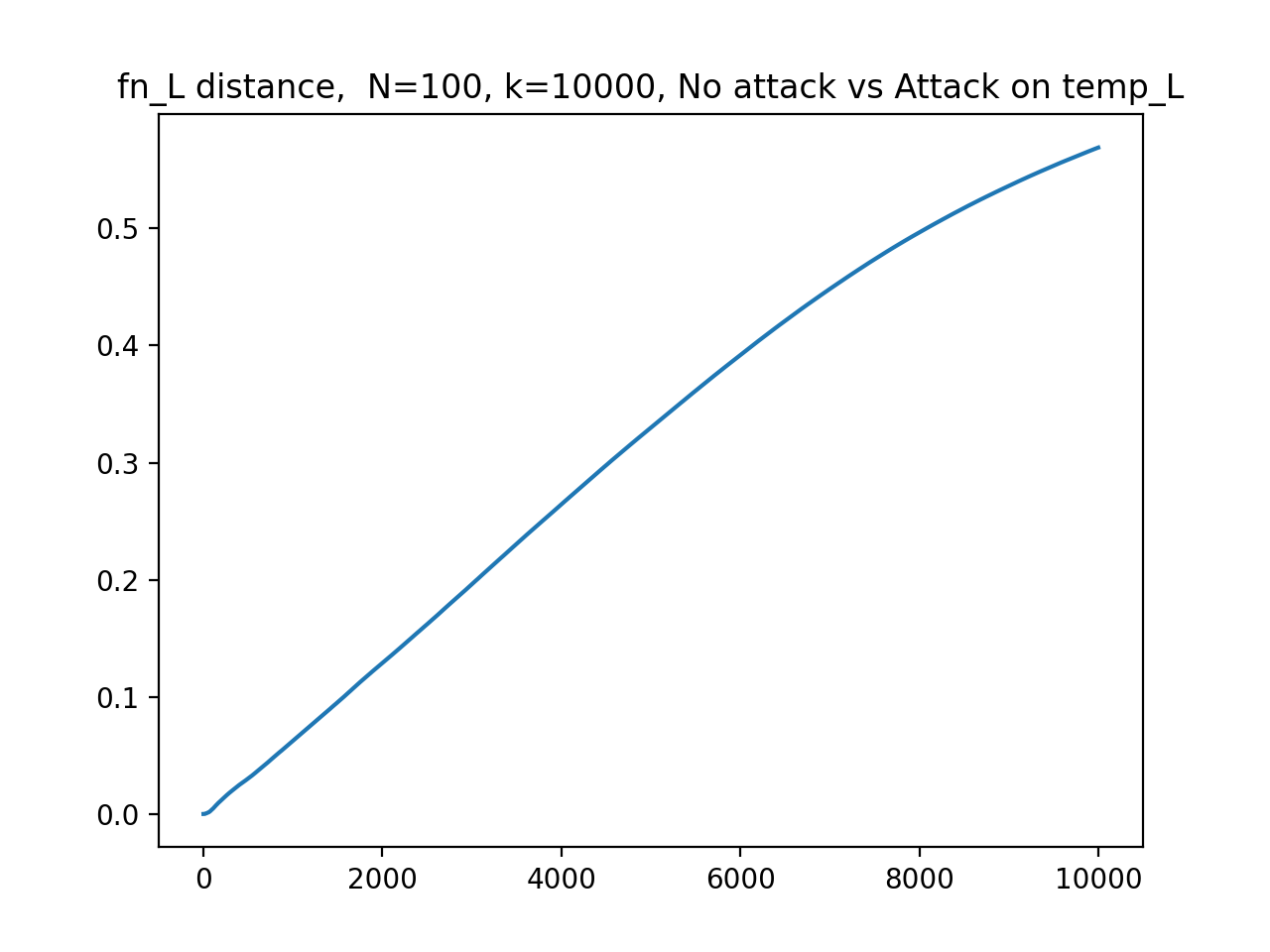}
\caption{$\mathit{fn\_L}$ distance, attack on $\mathsf{temp\_L}$}%
\label{fig_fn_fp_engL_02}
\end{subfigure}
\\
\begin{subfigure}{0.42\textwidth}	\includegraphics[width=1.0\textwidth]{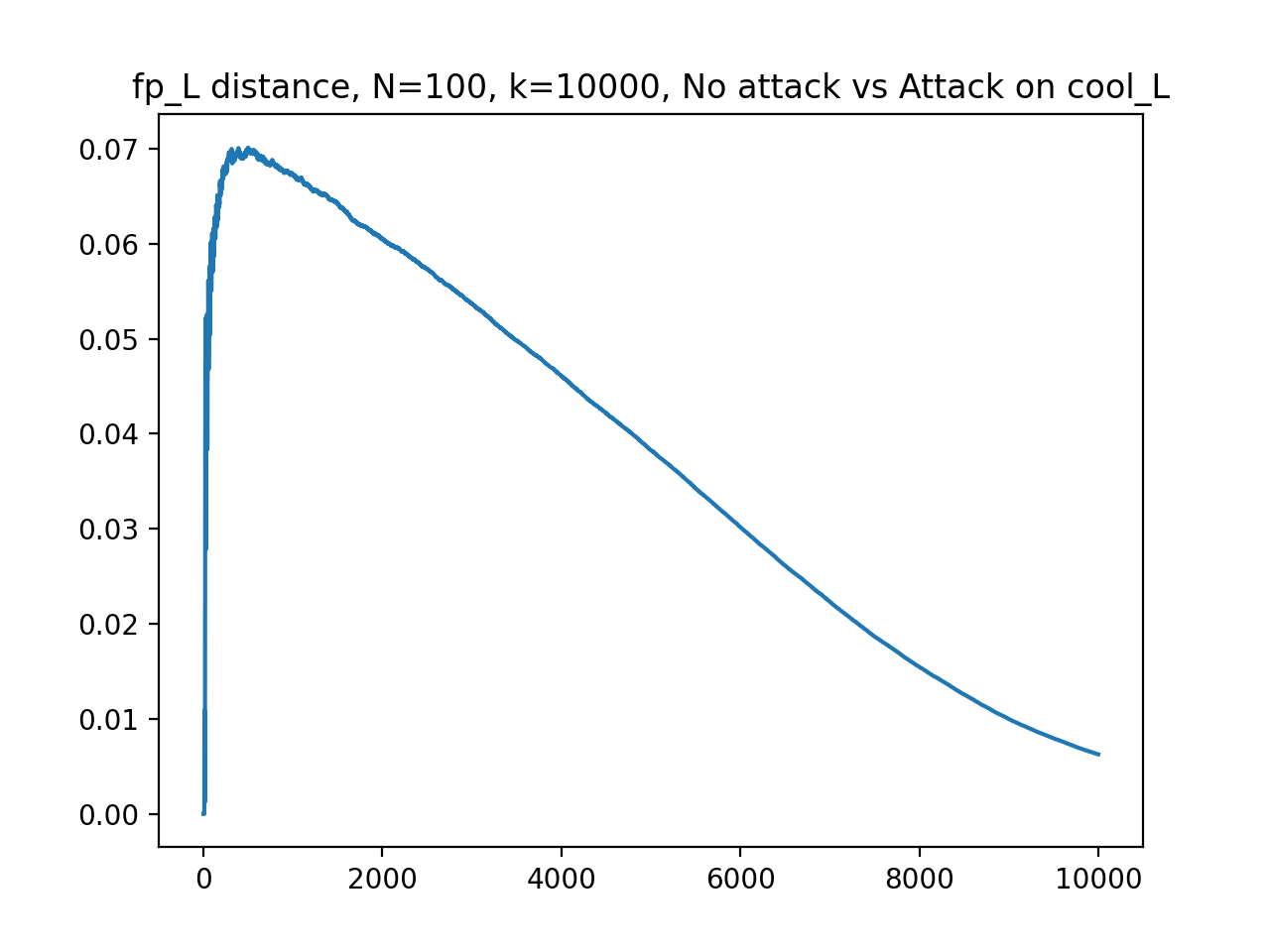}
\caption{$\mathit{fp\_L}$ distance, attack on $\mathsf{cool\_L}$}%
\label{fig_fn_fp_engL_03}
\end{subfigure}
\hfill
\begin{subfigure}{0.42\textwidth}	\includegraphics[width=1.00\textwidth]{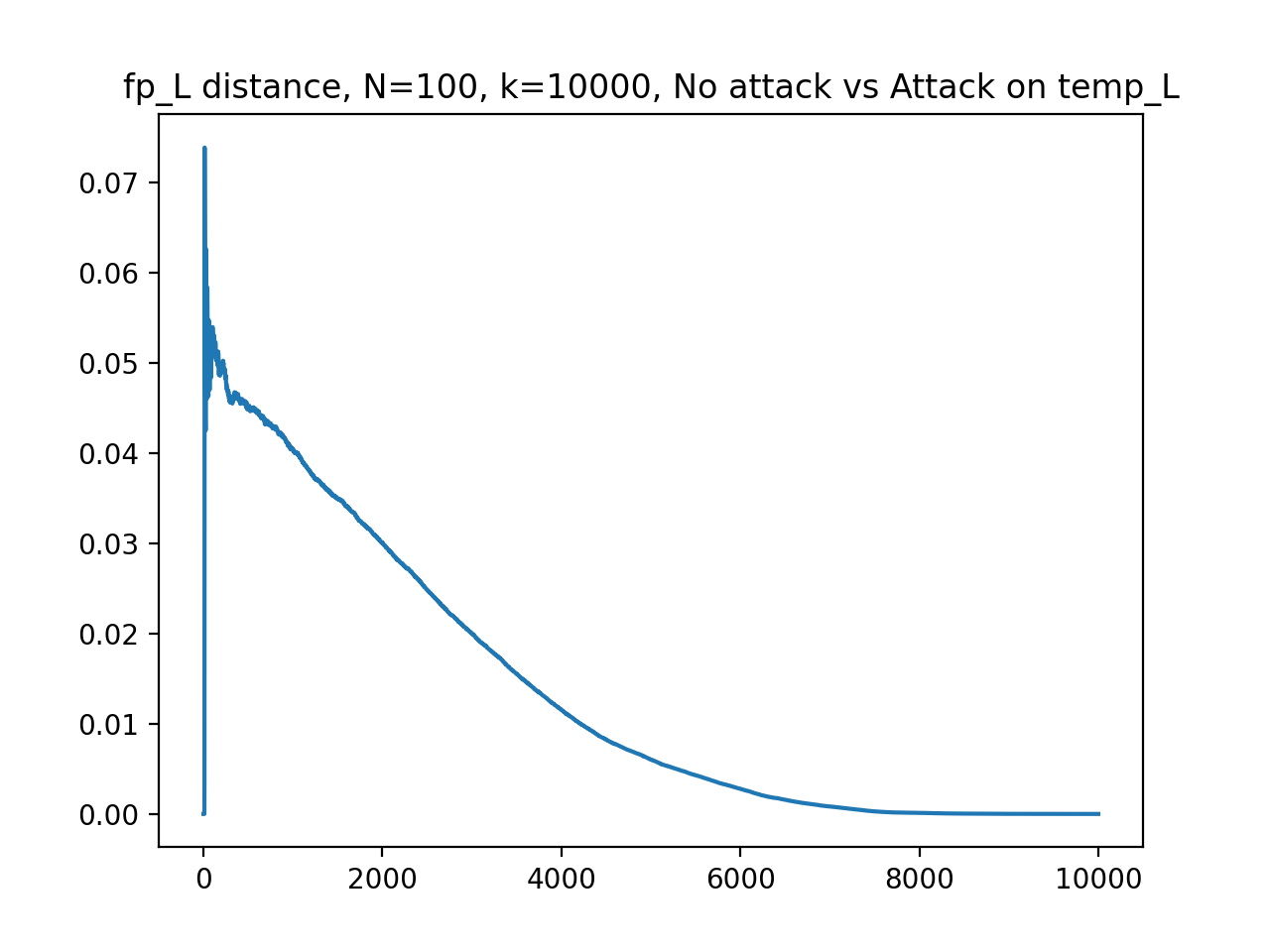}
\caption{$\mathit{fp\_L}$ distance, attack on $\mathsf{temp\_L}$}%
\label{fig_fn_fp_engL_04}
\end{subfigure}
\caption{False negative and false positive distance between genuine $\mathsf{Eng\_L}$ and $\mathsf{Eng\_L}$ under actuator or sensor attack, for the simulations of Figure~\ref{fig:average_temp}.}%
\label{fig:distance_attack_left}
\end{figure}

If we move to the whole system, the false negatives and false positives are carried by variables $\mathit{fn}$ and $\mathit{fp}$, that are updated by $\E$ as follows:
\[
\begin{array}{lll}
\mathit{fn}(\tau +1) & = &\frac{\tau * \mathit{fn}(\tau)  +
\min(1,
(0.7 * \max(0,\mathit{stress\_L(\tau)}-\mathit{l})
+
0.7 * \max(0,\mathit{stress\_R(\tau)}-\mathit{r})))
}{\tau+1}
\\[0.5 ex]
\mathit{fp}(\tau +1) & = &\frac{\tau * \mathit{fp}(\tau)   +
\min(1,
(0.7 * \max(0,\mathit{l}-\mathit{stress\_L}(\tau))
+
0.7 * \max(0,\mathit{r}-\mathit{stress\_R}(\tau))))}{\tau+1}
\end{array}
\]
where
$l = \begin{cases} 1 & \text{if } \mathit{ch\_alarm}(\tau) \in \{\mathsf{both},\mathsf{left}\}
\\
0 & \text{if } \mathit{ch\_alarm}(\tau) \in \{\mathsf{none},\mathsf{right}\}
\end{cases}$
and
$r = \begin{cases} 1 & \text{if } \mathit{ch\_alarm}(\tau) \in \{\mathsf{both},\mathsf{right}\}
\\
0 & \text{if } \mathit{ch\_alarm}(\tau) \in \{\mathsf{none},\mathsf{left}\}
\end{cases}$.

Here, false negatives and false positives represent the average effectiveness, and the average precision, respectively, of the supervisors modelled by $\mathsf{SV}$ to signal through channel $\mathit{ch\_alarm}$ that the system is under stress.
We remark that the values 0.7 are the values chosen in~\cite{LMMT21} to quantify the \emph{security clearance} of the two engines, where, in general, the clearance of a CPS component quantifies the importance of that component from a security point of view.
In this case, the two engines have the same criticality and, consequently, the same security clearance.

\begin{figure}[tbp]
\centering
\includegraphics[width=0.42\textwidth]{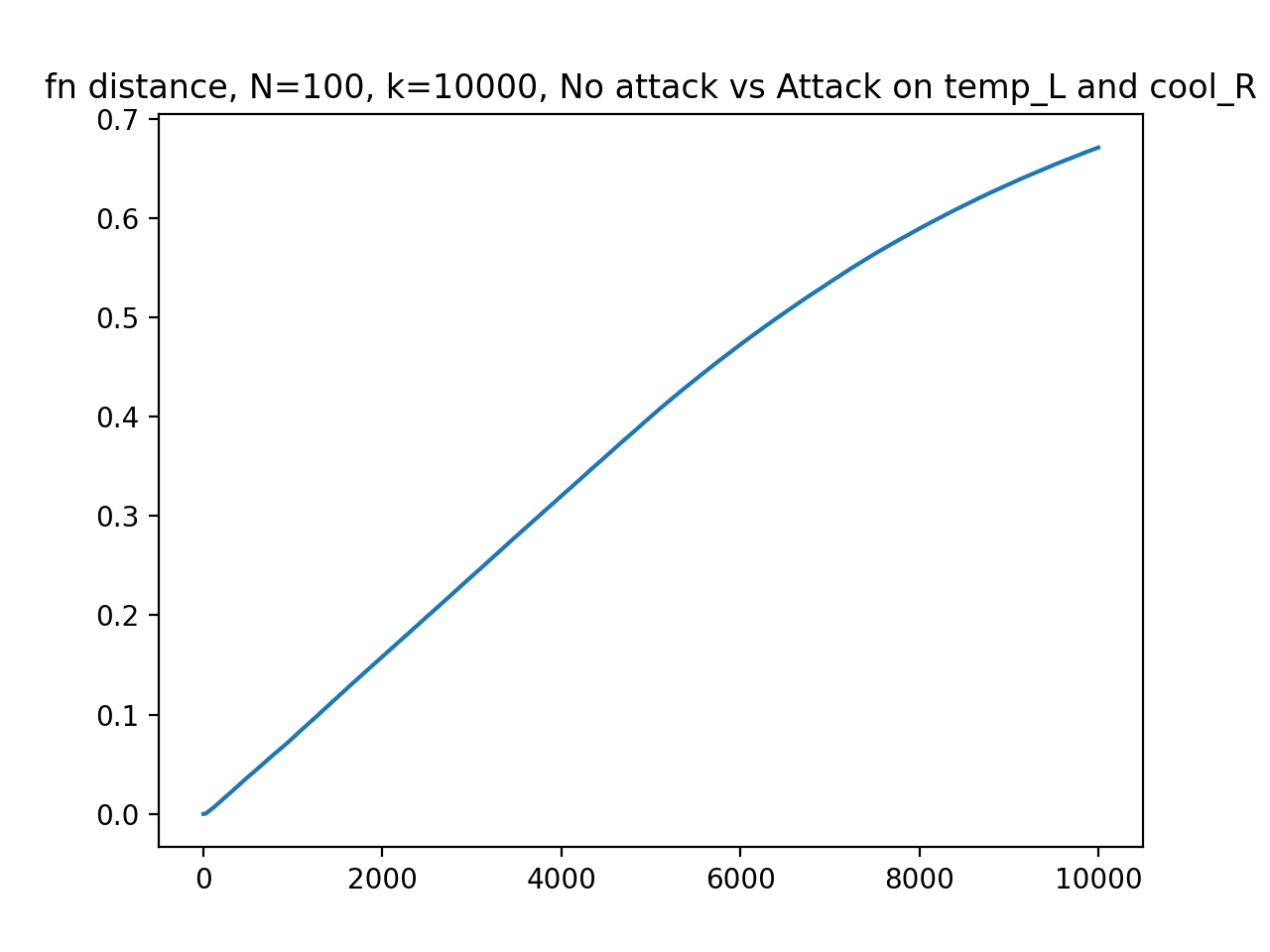}
\hfill
\includegraphics[width=0.42\textwidth]{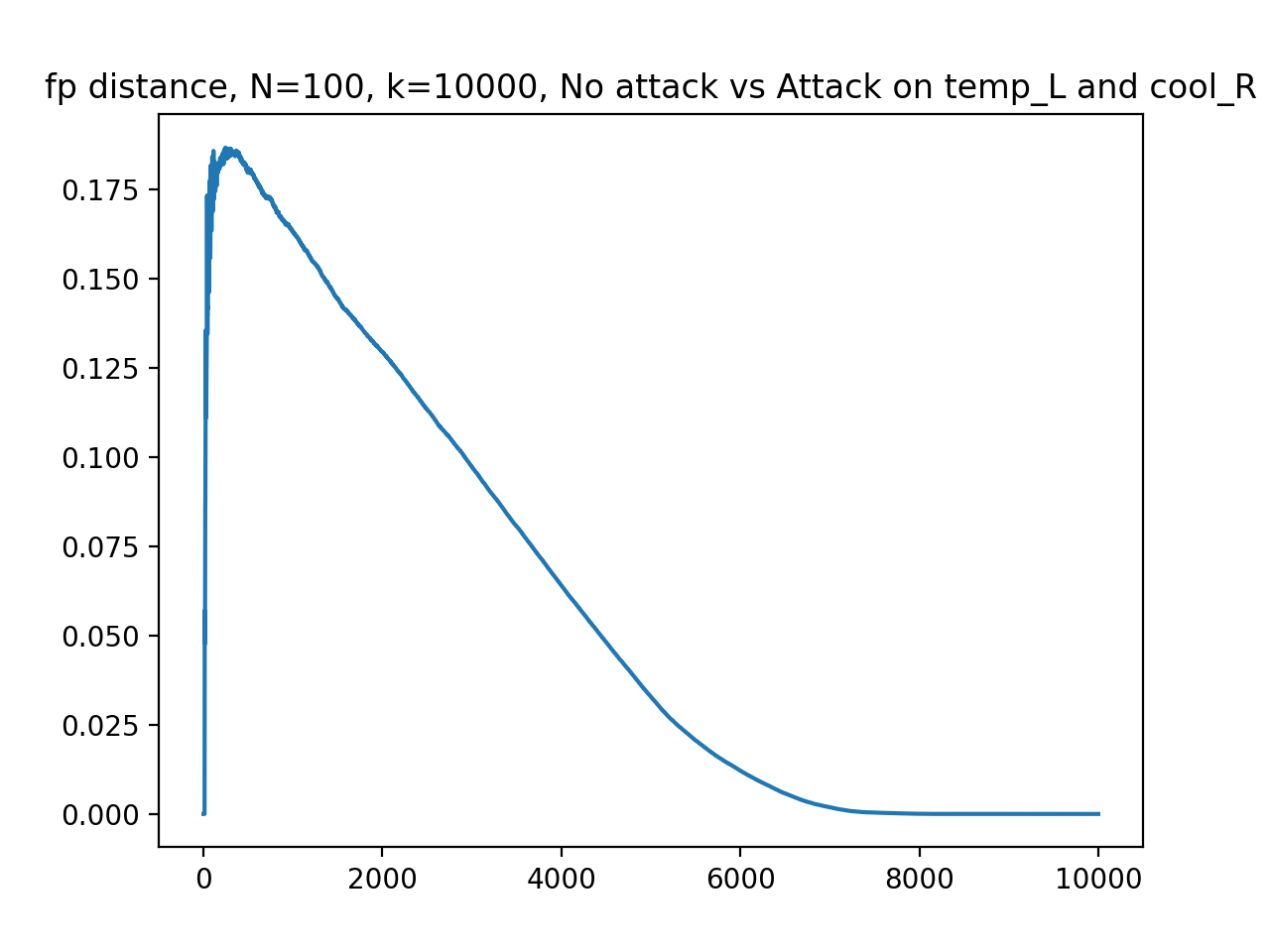}
\caption{False negative and false positive distance between genuine $\mathsf{Eng\_Sys}$ and $\mathsf{Eng\_Sys}$ under attack on left sensor and right actuator.}%
\label{fig:distance_attack_mix}
\end{figure}

In Figure~\ref{fig:distance_attack_mix}  we report the distance, in terms of false negatives and false positives, between the system $\mathsf{Eng\_Sys}$ and the compromised system $\mathsf{Eng\_Sys} \cmerge \mathsf{Att\_Act\_L} \cmerge \mathsf{Att\_Sen\_R}$ that is subject to both an attack on the actuator $\mathit{cool\_L}$ of the left engine and an attack on the sensor $\mathit{temp\_R}$ of the right engine.
Again, the settings for these simulations are the same of those used for simulations described in Figure~\ref{fig:average_temp}.

Let us focus now on the third attack, namely the one implemented by process $\mathsf{Att\_Saw\_X}$.
Intuitively, we expect a correlation between the length of the attack window and the impact of the attack.
In order to study such a relation, we may proceed as follows:
\begin{enumerate*}[(i)]
\item We define a penalty function $\rho$ so that its value depends on the length of the attack window $[\mathit{AW_l},\mathit{AW_r}]$.
Actually, we may define $\rho$ to return $\sfrac{(\mathit{AW_r} - \mathit{AW_l})}{\mathsf{AWML}}$.
In particular, if we consider the genuine system then $\rho$ returns $0$, since $\mathit{AW_l}=0=\mathit{AW_r}$.
\item We consider a penalty function $\rho'$ expressing the impact of the attack.
\item We estimate the $(\rho,\rho',[0,0],0,\eta_1,\eta_2)$-robustness of the genuine system.
\end{enumerate*}
More in detail, we may sample $M$ variations of the system at $\rho$-distance from the genuine one bounded by $\eta_1$.
This means sampling $M$ versions of $\mathsf{Eng\_X \cmerge Att\_Saw\_X}$ characterised by attack windows of length bounded by $\eta_1 \cdot \mathsf{AWML}$.
Then, by applying our simulations, we can estimate how worse these compromised systems behave with respect to the genuine system $\mathsf{Eng\_X}$, according to the impact expressed by $\rho'$.

In Figure~\ref{fig:robustness_tempfake_vs_fn} we report the results of such an estimation for $\rho'$ coinciding with $\mathit{fn\_L}$, $[\mathit{AW_l},\mathit{AW_r}] = [0,aw]$ an attack window of length $aw$ starting at first instant, $\mathsf{AWML}$ set to 1000 and $\eta_1=0.1$, which means that $aw$ is chosen probabilistically according to a uniform distribution on interval $[0,100]$.

Here we consider three different offset values assigned to $\mathsf{TF}$: 0.4, 0.6, and 0.8.
In all cases, we estimate to have robustness, for values for $\eta_2$ that are $0.012$, $0.05$, and $0.11$, respectively.

\begin{figure}[tbp]
\centering
\begin{subfigure}{0.32\textwidth}
\includegraphics[width=1.00\textwidth]{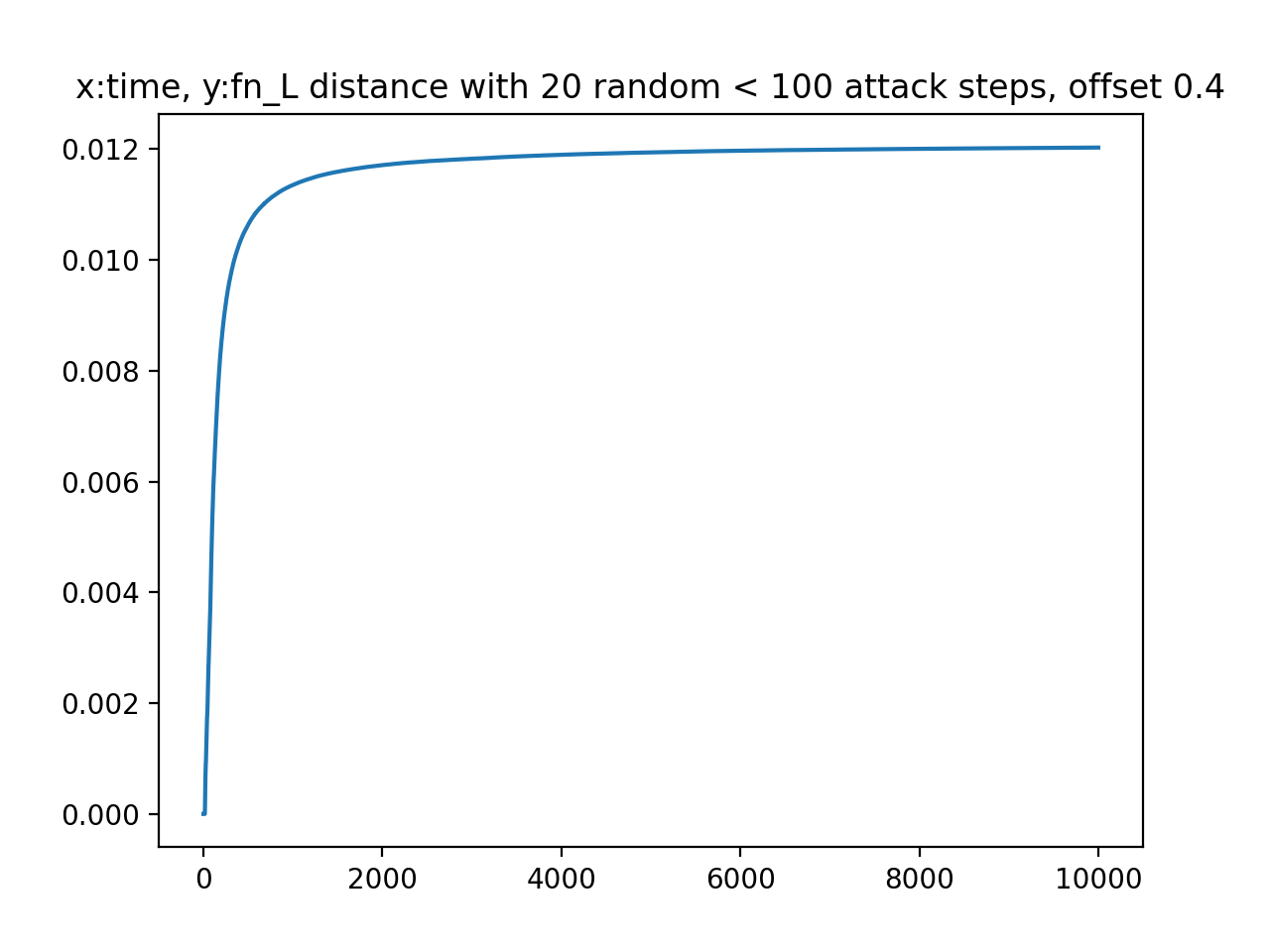}
\caption{$\mathit{TF}=0.4$}%
\label{fig04}
\end{subfigure}
\begin{subfigure}{0.32\textwidth}
\centering
\includegraphics[width=1.00\textwidth]{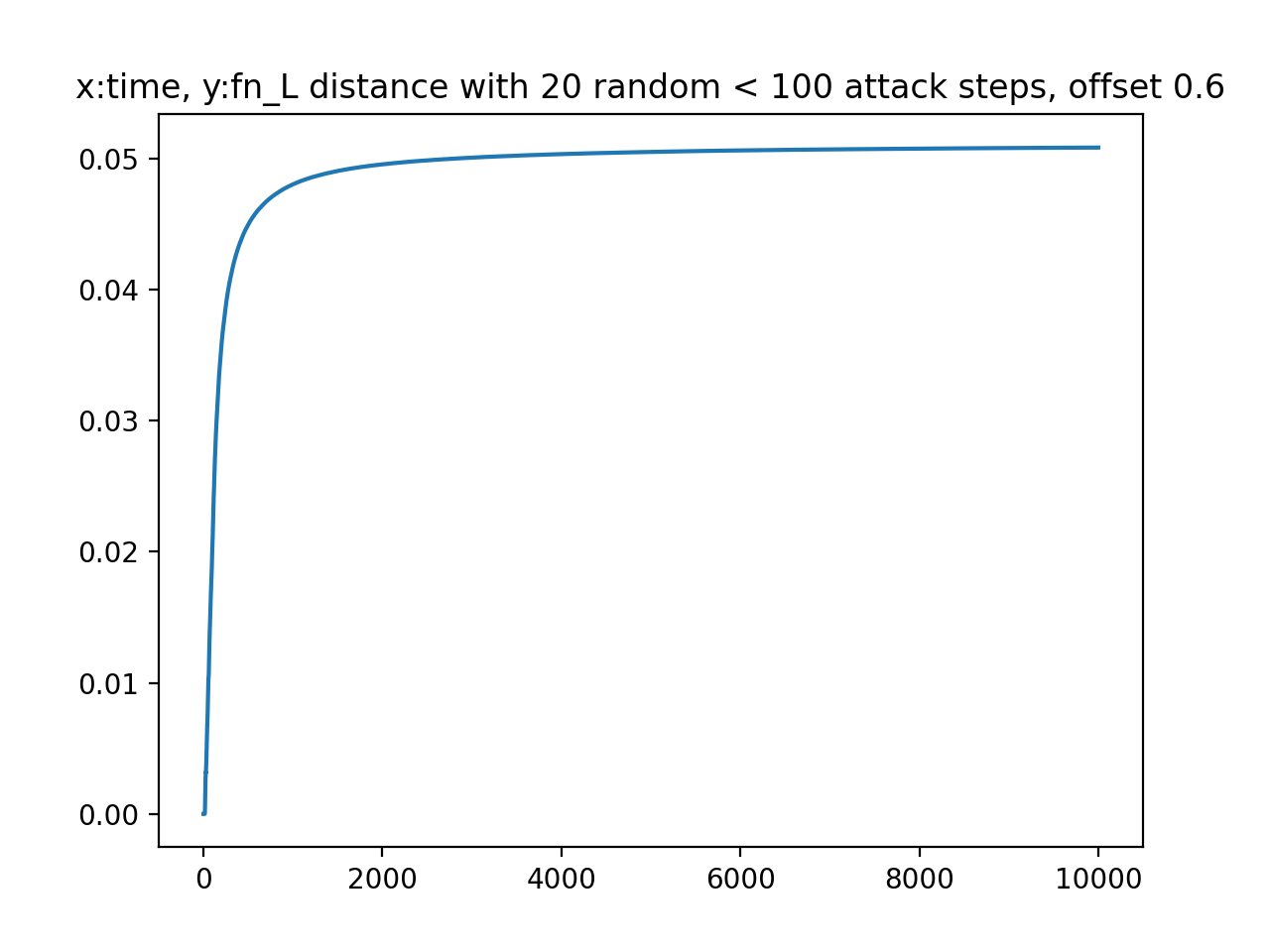}
\caption{$\mathit{TF}=0.6$}%
\label{fig06}
\end{subfigure}
\begin{subfigure}{0.32\textwidth}
\centering
\includegraphics[width=1.00\textwidth]{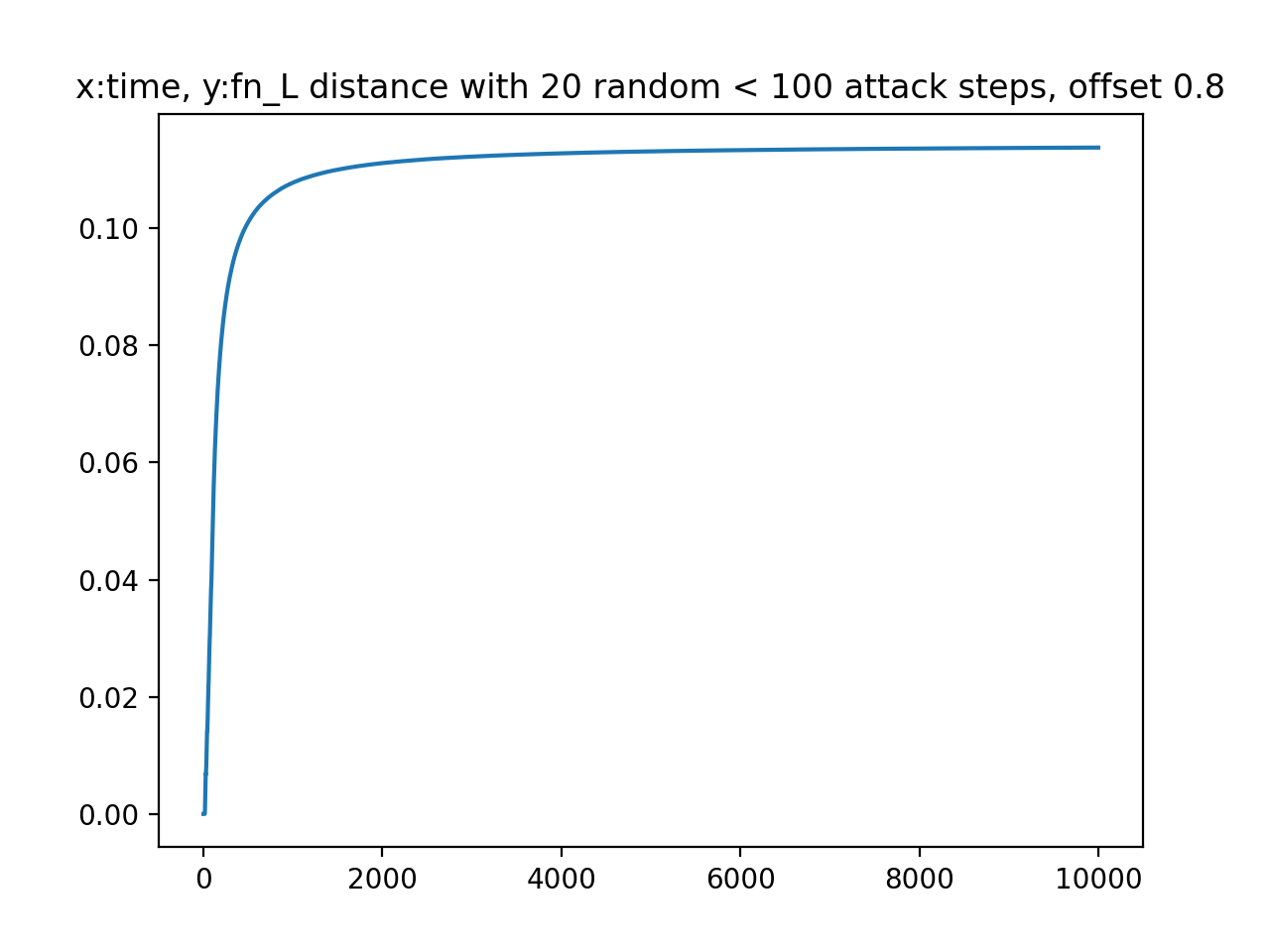}
\caption{$\mathit{TF}=0.8$}%
\label{fig08}
\end{subfigure}
\caption{Estimation of false negative distance between a genuine engine and an engine subject to tampering with sensor $\mathit{temp\_L}$.
The attacker adds a negative offset $-0.4$ (Fig.~\ref{fig04}), or $-0.6$ (Fig.~\ref{fig06}), or $-0.8$ (Fig.~\ref{fig08}) to sensor, on the attack window $[0,aw]$, with $aw$ an integer uniformly distributed in $[0,100]$.
Simulation settings for each offset:
$M=20$ different sampled values for each $aw$,
$N=100$ runs by the genuine engine and
$N*l=1000$ runs by the compromised system for each $aw$,
$k=10000$ steps per run.}%
\label{fig:robustness_tempfake_vs_fn}
\end{figure}

Now, let us change the approach for attack quantification and assume that only the inflicted overstress is considered, without taking into account the raised alarm signals.
In particular, this overstress is quantified at a given computation step $n$ as the ratio between the number of computation steps where the stress condition $|\{k \mid \mathit{p_k} \ge \mathsf{max}\} | > 3$ holds and $n$.
Therefore, the penalty function $\rho'$ is defined now as the ratio between the values of two variables, one counting the number of instants with the stress condition and one counting the total number of steps.
Robustness can be estimated as in the earlier case.
In Figure~\ref{fig:average_stressL} we report the robustness of the left engine $\mathsf{Eng\_L}$.
Here we use the same simulation settings used in Figure~\ref{fig:robustness_tempfake_vs_fn} and the same offset values.
The figure shows that we have robustness for low values of $\eta_2$, which, in essence, tells us that attacks in a given window are unable to generate overstress in the long run.

\begin{figure}[tbp]
\centering
\begin{subfigure}{0.32\textwidth}
\includegraphics[width=1.00\textwidth]{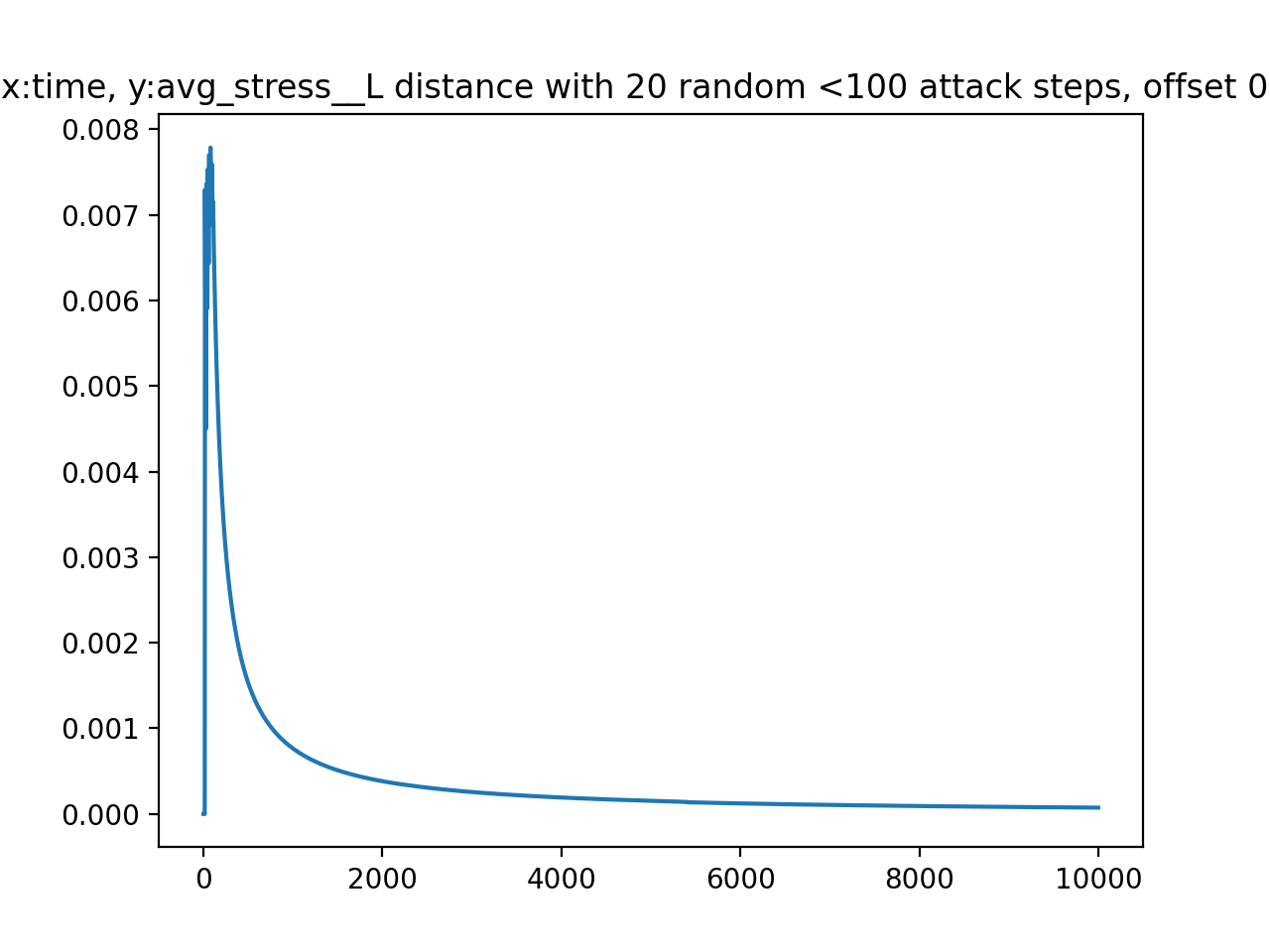}
\caption{$\mathit{TF}=0.4$}%
\label{fig04s}
\end{subfigure}
\begin{subfigure}{0.32\textwidth}
\centering
\includegraphics[width=1.00\textwidth]{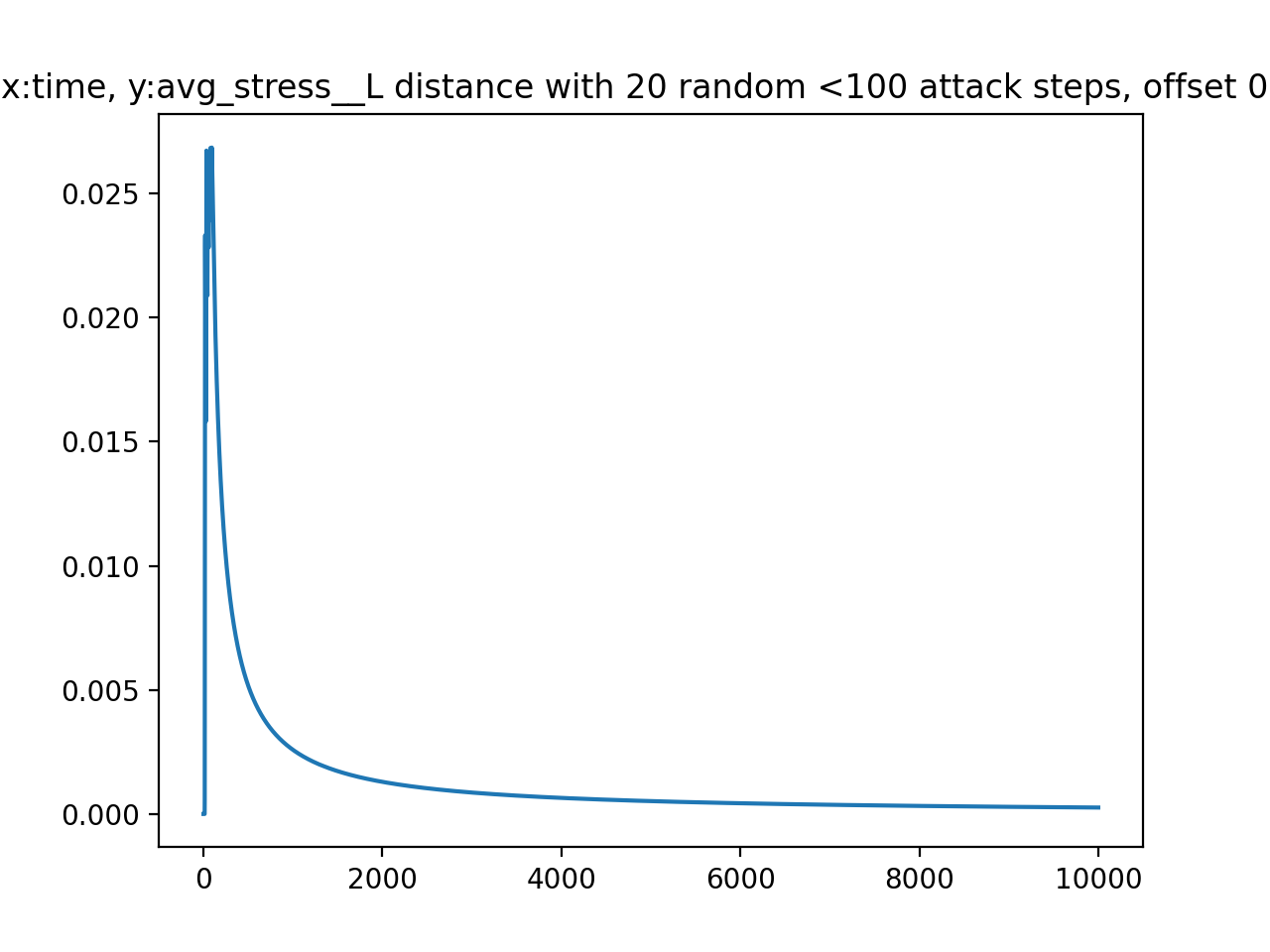}
\caption{$\mathit{TF}=0.6$}%
\label{fig06s}
\end{subfigure}
\begin{subfigure}{0.35\textwidth}
\centering
\includegraphics[width=1.00\textwidth]{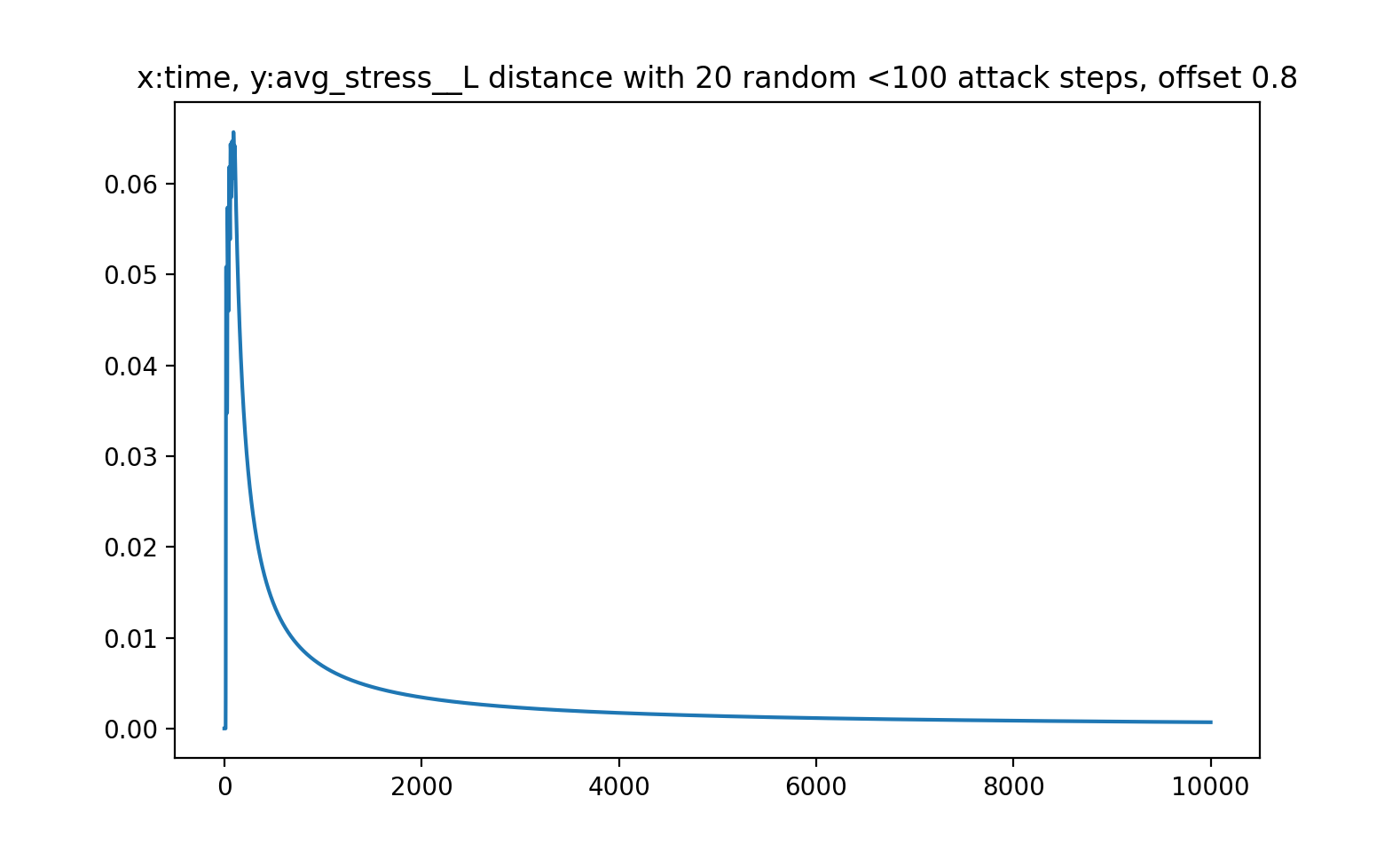}
\caption{$\mathit{TF}=0.8$}%
\label{fig08s}
\end{subfigure}
\caption{Estimation of average stress distance between a genuine engine and an engine subject to tampering with sensor $\mathit{temp\_L}$.
The simulation settings are the same of Figure~\ref{fig:robustness_tempfake_vs_fn}.}%
\label{fig:average_stressL}
\end{figure}

%===============================================
% sec - conclusion
%================================================

\section{Concluding Remarks}%
\label{sec:conclusion}

The notion of robustness is used in different contexts, from control
theory~\cite{ZD98} to biology~\cite{Ki07}.
Commonly, it is meant as the ability of a system to maintain its functionalities against external and internal perturbations.
In the context of CPSs,~\cite{FKP16} individuates five kinds of robustness:
\begin{enumerate*}[(i)]
\item\label{iorob} input/output robustness;
\item robustness with respect to system parameters;
\item robustness in real-time system implementation;
\item\label{unpred-env} robustness due to unpredictable environment;
\item robustness to faults.
\end{enumerate*}
The notions of adaptability and reliability considered in this paper fall in category~\ref{unpred-env}.
In particular, in the example of the engine system the attacks are the source of environment's unpredictability.

We notice that our notions of robustness differ from classical ones like those in~\cite{FP09,DM10}.
Our notions require to compare all behaviours of two different systems, whereas~\cite{FP09,DM10} compare a single behaviour of a single system with the set of the behaviours that satisfy a given property, which is specified by means of a formula expressed in a suitable temporal logic.

Moreover, the generality of our notions of robustness allows us to capture classical properties of linear systems.
For instance, our notion of adaptability can be used to encode the overshoot~\cite{ZNTH19} and settling time~\cite{Do08} of signals.
Indeed, there are substantial differences between them.
Specifically:
\begin{enumerate*}[(i)]
\item adaptability is a property of CPSs in general, whereas overshoot and settling time are referred to signals;
\item adaptability is defined over the \traccione{} of a system, i.e., over all possible trajectories, whereas the other notions are dependency measures of a single signal;
\item using the notion of Definition~\ref{def:adaptability}, the overshoot can be expressed as the parameter $\eta_1$, and the settling time as $\tilde{\tau}$.
\end{enumerate*}
This means that not only our notion of adaptability captures both properties, but it can be used to study their correlation.

As a first step for future research we will provide a simple logic, defined in the vein of \emph{Signal Temporal Logic} (STL)~\cite{MN04}, that can be used to specify requirements on the \tracciones{} of a system.
Our intuition is that we can exploit the \spell metric, and the algorithm we have proposed, to develop a quantitative model checking tool for this type of systems.
Recently we have developed the \emph{Robustness Temporal Logic} (\emph{RobTL})~\cite{CLT23arXiv}, together with its model checker that is integrated in our \emph{Software Tool for the Analysis of Robustness in the unKnown environment} (\emph{\textsc{Stark}})~\cite{CLT23stark}.
This logic allows us to specify temporal requirements on the evolution of distances between the nominal behaviour of a system and its perturbed version, thus including properties of robustness against uncertainties of programs.
However, we plan to propose a temporal logic allowing us to deal with the presence of uncertainties from another point of view: our aim is to express requirements on the \emph{expected} behaviour of the system in the presence of uncertainties and perturbations, by using probability measures over \datastates{} as atomic propositions.

Moreover, we would like to enhance the modelling of time in our framework.
Firstly we could relax the timing constraints on the \spell metric by introducing a \emph{time tolerance} and defining a \emph{stretched \spell metric} as a Skorokhod-like metric~\cite{Sk56}, as those used for conformance testing~\cite{DMP17}.
Then, we could provide an extension of our techniques to the case in which also the program shows a continuous time behaviour.

The use of metrics for the analysis of systems stems from~\cite{GJS90,DGJP04,KN96} where, in a process algebraic setting, it is argued that metrics are indeed more informative than behavioural equivalences when quantitative information on the behaviour is taken into account.
The Wasserstein lifting has then found several successful applications: from the definition of \emph{behavioural metrics} (e.g.,~\cite{B05,CLT20,GLT16,GT18}), to privacy~\cite{CGPX14,CCP18,CCP20} and machine learning (e.g.,~\cite{ACB17,GAADC17,TBGS18}).
Usually, one can use behavioural metrics to quantify how well an implementation ($I$) meets its specification ($S$).
In~\cite{CHR12} the authors do so by setting a two players game with weighted choices, and the cost of the game is interpreted as the distance between $I$ and $S$.
Hence the authors propose three distance functions: \emph{correctness}, \emph{coverage}, and \emph{robustness}.
Correctness expresses how often $I$ violates $S$, coverage is its dual,and robustness measures how often $I$ can make an unexpected error with the resulting behaviour still meeting $S$.
A similar game-based approach is used in~\cite{CDDP19} to define a \emph{masking fault-tolerant} distance.
Briefly, a system is masking fault-tolerant if faults do not induce any observable behaviour in the system.
Hence, the proposed distance measures how many faults are tolerated by $I$ while being masked by the states of the system.
Notice that the notion of robustness from~\cite{CHR12} and the masking fault-tolerant distance from~\cite{CDDP19} are quite different from our notions of reliability and robustness.
In fact, we are not interested in counting how many times an error occurs, but in checking whether the system is able to regain the desired behaviour after the occurrence of an error.

One of the main objectives in the control theory research area is the synthesis of controllers satisfying some desired property, such as safety, stability, robustness, etc.
Conversely, our purpose in this paper was to provide some tools for the \emph{analysis} of the interaction of a \emph{given} program with an environment.
The idea is that these tools can be used to test a synthesised controller against its deployment in various environments.
While a direct application of our framework to the synthesis of programs does not seem feasible, it would be interesting to investigate whether a combination of it with learning techniques can be used for the design and synthesis of robust controllers.
Our confidence in a potential application relies on some recently proposed metric-based approaches to controllers synthesis~\cite{GBSST19,AP11}, in which, however, the environment is only modelled deterministically.

To conclude, we remark that our \tracciones{} are not a rewriting of Markov processes as \emph{transformers of distributions}~\cite{KA04,KVAK10}.
Roughly, in the latter approach, one can interpret state-to-state transition probabilities as a single distribution over the state space, so that the behaviour of the system is given by the sequence of the so obtained distributions.
While the two approaches may seem similar, there are some substantial differences.
Firstly, the state space in~\cite{KA04,KVAK10} is finite and discrete, whereas here we are in the continuous setting (cf. Remark~\ref{rmk:continuity}).
Secondly, the transformers of distributions consider the behaviour of the system as a whole, i.e., it is not possible to separate the logical component from the environment.

%===============================================
% ACK
%================================================

\section*{Acknowledgements}

This work is supported by the project \emph{Programs in the wild: uncertainties, adaptability and verification} (ULTRON) of the Icelandic Research Fund (grant No.~228376-051).
Moreover, V. Castiglioni has been supported by the project \emph{Open Problems in the Equational Logic of Processes} (OPEL) of the Icelandic Research Fund (grant No.~196050-051).

\bibliographystyle{alphaurl}
\bibliography{Traccione}

\end{document}